\newif\ifsubmission
\ifsubmission
\documentclass{llncs}
\else
\documentclass[11pt]{article}
\usepackage{palatino}
\usepackage{fullpage}
\fi

\newif\iffull

\usepackage{cmap} 
\usepackage[pdfstartview=FitH,pdfpagemode=None,colorlinks,linkcolor=blue,filecolor=blue,citecolor=Violet,urlcolor=red]{hyperref}
\usepackage{multirow}
\usepackage{float}
\usepackage[section]{placeins}
\usepackage{mathtools}  
\usepackage[dvipsnames]{xcolor}
\usepackage{amsmath,amssymb,algpseudocode,algorithm,cryptocode}\ifsubmission\else\usepackage{amsthm}\fi
\usepackage{mathrsfs}
\usepackage[capitalize]{cleveref}
\usepackage{xspace}
\usepackage{braket}
\usepackage{comment}
\usepackage{hhline}
\usepackage{tablefootnote}
\usepackage{tikz-cd}
\usepackage{url}
\usepackage{xcolor}
\usepackage{tcolorbox}
\usepackage{authblk}

\definecolor{darkgreen}{rgb}{0,0.5,0}
\newif\ifnotes
\notestrue

\ifsubmission
\else
\newtheorem{theorem}{Theorem}[section]
\newtheorem{importedtheorem}[theorem]{Imported Theorem}

\newtheorem{claim}[theorem]{Claim}
\newtheorem{fact}[theorem]{Fact}
\newtheorem{lemma}[theorem]{Lemma}

\newtheorem{definition}[theorem]{Definition}

\newtheorem{remark}[theorem]{Remark}

\theoremstyle{remark}
\fi

\Crefname{importedtheorem}{Imported Theorem}{Imported Theorems}
\Crefname{theorem}{Theorem}{Theorems}
\Crefname{proposition}{Proposition}{Propositions}
\Crefname{claim}{Claim}{Claims}
\Crefname{fact}{Fact}{Fact}
\Crefname{lemma}{Lemma}{Lemmas}
\Crefname{conjecture}{Conjecture}{Conjectures}
\Crefname{corollary}{Corollary}{Corollaries}
\Crefname{construction}{Construction}{Constructions}
\Crefname{property}{Property}{Properties}

\Crefname{definition}{Definition}{Definitions}
\Crefname{assumption}{Assumption}{Assumptions}
\Crefname{notation}{Notation}{Notations}

\Crefname{question}{Question}{Questions}
\Crefname{remark}{Remark}{Remarks}
\Crefname{comment}{Comment}{Comments}

\newcommand{\secp}{\lambda}

\def\cA{{\cal A}}

\def\cD{{\cal D}}

\def\cF{{\cal F}}

\def\cH{{\cal H}}

\def\cO{{\cal O}}

\def\cQ{{\cal Q}}
\def\cR{{\cal R}}
\def\cS{{\cal S}}

\def\regA{{\cal A}}
\def\regB{{\cal B}}
\def\regC{{\cal C}}

\def\regG{{\cal G}}

\def\regK{{\cal K}}

\def\regQ{{\cal Q}}
\def\regR{{\cal R}}
\def\regS{{\cal S}}

\def\regU{{\cal U}}
\def\regV{{\cal V}}

\def\regX{{\cal X}}
\def\regY{{\cal Y}}
\def\regZ{{\cal Z}}

\def\cCK{{\mathsf C}{\mathsf K}}
\def\cPK{{\mathsf P}{\mathsf K}}
\def\bCK{{\mathbf C}{\mathbf K}}

\def\sA{{\mathsf A}}
\def\sB{{\mathsf B}}
\def\sC{{\mathsf C}}
\def\sD{{\mathsf D}}

\def\sH{{\mathsf H}}

\def\sP{{\mathsf P}}

\def\sS{{\mathsf S}}

\def\sU{{\mathsf U}}
\def\sV{{\mathsf V}}

\def\bbC{{\mathbb C}}

\def\bbF{{\mathbb F}}

\def\bbI{{\mathbb I}}

\def\bbN{{\mathbb N}}

\def\bbR{{\mathbb R}}

\newcommand{\bc}{\mathbf{c}}

\newcommand{\bB}{\mathbf{B}}

\newcommand{\bdk}{\mathbf{dk}}
\newcommand{\bck}{\mathbf{ck}}

\def\poly{{\rm poly}}

\def\negl{{\rm negl}}

\newcommand{\pk}{\mathsf{pk}}
\newcommand{\sk}{\mathsf{sk}}

\newcommand{\KeyGen}{\mathsf{KeyGen}}
\newcommand{\Setup}{\mathsf{Setup}}

\newcommand{\Prove}{\mathsf{Prove}}

\newcommand{\Com}{\mathsf{Com}}

\newcommand{\Sim}{\mathsf{Sim}}
\newcommand{\Sign}{\mathsf{Sign}}
\newcommand{\Enc}{\mathsf{Enc}}
\newcommand{\Dec}{\mathsf{Dec}}

\newcommand{\Verify}{\mathsf{Verify}}

\newcommand{\ct}{\mathsf{ct}}

\newcommand{\Eval}{\mathsf{Eval}}

\DeclareMathOperator*{\expectation}{\mathbb{E}}
\newcommand{\E}{\expectation}

\newcommand{\union}{\cup}

\newcommand{\PRF}{\mathsf{PRF}}

\newcommand{\Gen}{\mathsf{Gen}}

\newcommand{\pp}{\mathsf{pp}}
\newcommand{\sparam}{\mathsf{sp}}

\newcommand{\vk}{\mathsf{vk}}
\newcommand{\aux}{\mathsf{aux}}

\newcommand{\ketbra}[2]{\mathinner{|{#1}\rangle\,\langle{#2}|}}

\newcommand{\proref}[1]{Protocol~\protect\ref{#1}}

\newenvironment{boxfig}[2]{\begin{figure}[#1]\fbox{
    \begin{minipage}{\linewidth}
    \vspace{0.2em}\makebox[0.025\linewidth]{}    \begin{minipage}{0.95\linewidth}{{#2 }}
    \end{minipage}\vspace{0.2em}\end{minipage}}}{\end{figure}}

\newcommand{\pprotocol}[4]{
\begin{boxfig}{h!}{
\begin{center}
\textbf{#1}
\end{center}
    {\footnotesize#4}
\vspace{0.2em} } \caption{\label{#3} #2}
\end{boxfig}
}

\newcommand{\protocol}[4]{
\pprotocol{#1}{#2}{#3}{#4} }

\newcommand{\crs}{\mathsf{crs}}

\newcommand{\out}{\mathsf{out}}

\newcommand{\Ver}{\mathsf{Ver}}
\newcommand{\Open}{\mathsf{Open}}

\newcommand{\state}{\mathsf{state}}

\newcommand{\parl}{\mathsf{parl}}

\newcommand{\QFHE}{\mathsf{QFHE}}

\newcommand{\TD}{\mathsf{TD}}

\newcommand{\Exp}{\mathsf{EXP}}

\newcommand{\Tr}{\mathsf{Tr}}

\newcommand{\msk}{\mathsf{msk}}
\newcommand{\Obf}{\mathsf{Obf}}

\newcommand{\Out}{\mathsf{Out}}

\newcommand{\Invert}{\mathsf{Invert}}
\newcommand{\Ch}{\mathsf{Check}}

\newcommand{\Valid}{\mathsf{Valid}}
\newcommand{\IsValid}{\mathsf{IsValid}}

\newcommand{\CV}{\mathsf{CV}}
\newcommand{\QV}{\mathsf{QV}}
\newcommand{\PV}{\mathsf{PV}}

\newcommand{\Prep}{\mathsf{Prep}}

\newcommand{\OpenZ}{\mathsf{OpenZ}}
\newcommand{\OpenX}{\mathsf{OpenX}}
\newcommand{\DecZ}{\mathsf{DecZ}}
\newcommand{\DecX}{\mathsf{DecX}}

\newcommand{\ck}{\mathsf{ck}}
\newcommand{\dk}{\mathsf{dk}}
\newcommand{\inside}{\mathsf{in}}
\newcommand{\Tok}{\mathsf{Tok}}
\newcommand{\PFC}{\mathsf{PFC}}
\newcommand{\TCF}{\mathsf{TCF}}
\newcommand{\nonnegl}{\mathsf{non}\text{-}\mathsf{negl}}

\newcommand{\Amplify}{\mathsf{Amplify}}
\newcommand{\err}{\mathsf{err}}
\newcommand{\Maj}{\mathsf{Maj}}

\newcommand{\TestOut}{\mathsf{TestRoundOutputs}}
\newcommand{\CObf}{\mathsf{CObf}}
\newcommand{\CEval}{\mathsf{CEval}}
\newcommand{\QObf}{\mathsf{QObf}}
\newcommand{\QEval}{\mathsf{QEval}}
\newcommand{\Meas}{\mathsf{Meas}}
\newcommand{\CVGen}{\mathsf{CVGen}}
\newcommand{\CCR}{\mathsf{CCR}}

\newcommand{\co}{\mathsf{co}}
\newcommand{\PFCI}{\mathsf{PFC}\text{-}\mathsf{I}}
\newcommand{\PFCII}{\mathsf{PFC}\text{-}\mathsf{II}}
\newcommand{\MM}{\mathsf{MM}}
\newcommand{\FE}{\mathsf{FE}}
\newcommand{\CB}{\mathsf{CB}}
\newcommand{\Combine}{\mathsf{Combine}}
\newcommand{\VerGen}{\mathsf{VerGen}}

\begin{document}
\def\doi#1{\url{https://doi.org/#1}}

\title{Obfuscation of Pseudo-Deterministic Quantum Circuits}
\author[$\star$]{James Bartusek\thanks{Part of this work was done while visiting NTT Social Informatics Laboratories for an internship.}}
\author[$\ddagger$]{Fuyuki Kitagawa}
\author[$\ddagger$]{Ryo Nishimaki}
\author[$\ddagger$]{Takashi Yamakawa}
\affil[$\star$]{{\small UC Berkeley, USA}\authorcr{\small bartusek.james@gmail.com}}
\affil[$\ddagger$]{{\small NTT Social Informatics Laboratories, Tokyo, Japan}\authorcr{\small \{fuyuki.kitagawa.yh,ryo.nishimaki.zk,takashi.yamakawa.ga\}@hco.ntt.co.jp}}
\renewcommand\Authands{, }
\date{}
\maketitle
\thispagestyle{empty}

\begin{abstract}
    We show how to obfuscate pseudo-deterministic quantum circuits in the classical oracle model, assuming the quantum hardness of learning with errors. Given the classical description of a quantum circuit $Q$, our obfuscator outputs a quantum state $\ket{\widetilde{Q}}$ that can be used to evaluate $Q$ repeatedly on arbitrary inputs.
    
    Instantiating the classical oracle using any candidate post-quantum indistinguishability obfuscator gives us the first candidate construction of indistinguishability obfuscation for all polynomial-size pseudo-deterministic quantum circuits. In particular, our scheme is the first candidate obfuscator for a class of circuits that is powerful enough to implement Shor's algorithm (SICOMP 1997).
    
    Our approach follows Bartusek and Malavolta (ITCS 2022), who obfuscate \emph{null} quantum circuits by obfuscating the verifier of an appropriate classical verification of quantum computation (CVQC) scheme. We go beyond null circuits by constructing a publicly-verifiable CVQC scheme for quantum \emph{partitioning} circuits, which can be used to verify the evaluation procedure of Mahadev's quantum fully-homomorphic encryption scheme (FOCS 2018). We achieve this by upgrading the one-time secure scheme of Bartusek (TCC 2021) to a fully reusable scheme, via a publicly-decodable \emph{Pauli functional commitment}, which we formally define and construct in this work. This commitment scheme, which satisfies a notion of binding against committers that can access the receiver's standard and Hadamard basis decoding functionalities, is constructed by building on techniques of Amos, Georgiou, Kiayias, and Zhandry (STOC 2020) introduced in the context of equivocal but collision-resistant hash functions. 
    
\end{abstract}

\newpage
\thispagestyle{empty}

{
  \hypersetup{linkcolor=Violet}
  \setcounter{tocdepth}{2}
  \tableofcontents
}
\newpage

\pagenumbering{arabic}
\section{Introduction}
A program obfuscator is a ``one-way compiler'' that renders code unintelligible without harming its functionality. This concept dates back to the beginning of modern cryptography \cite{1055638}, and has since attracted much interest as a tool for protecting software against reverse-engineering, intellectual property theft, and piracy. While the theoretical foundations of program obfuscation were laid in 2001 \cite{JACM:BGIRSVY12},\footnote{The preliminary version appeared in CRYPTO 2001.} it was not until 2013 \cite{SIAMCOMP:GGHRSW16}\footnote{The preliminary version appeared in FOCS 2013.} that researchers developed a proposal for obfuscating general-purpose (classical) computation. This first candidate sparked a massive research effort that has both established program obfuscation as a ``central hub'' \cite{SIAMCOMP:SahWat21} of cryptography with countless applications, and has resulted in obfuscation schemes based on well-founded cryptographic assumptions \cite{10.1145/3406325.3451093}.

Meanwhile, the concepts of quantum information and quantum computation have had a profound impact on computer science, with stunning applications such as unconditionally secure key agreement \cite{BenBra84} and efficient integer factorization \cite{10.1137/S0097539795293172}, not to mention the promise of major advances in chemistry and physics. As the field of quantum information science matures, researchers have investigated fundamental questions pertaining to information privacy and information integrity. This has resulted in a remarkable series of feasibility results for securing quantum information and computation, e.g. encryption \cite{892142}, authentication \cite{1181969}, zero-knowledge \cite{doi:10.1137/18M1193530}, secure multi-party computation \cite{10.1145/509907.510000,C:DupNieSal10,EC:DGHMW20}, and delegation of computation \cite{10.5555/2011670.2011674,arxiv:ABEM17,Broadbent_2009,DBLP:journals/nature/ReichardtUV13,FOCS:Mahadev18b,SIAMCOMP:Mahadev22}. However, despite these efforts, the feasibility of \emph{quantum obfuscation} has remained elusive, and the following question has remained largely open.

\begin{quote}\centering
\emph{Is it possible to obfuscate quantum computation?} 
\end{quote}

Prior research has focused its efforts on definitional work \cite{arxiv:AlaFef16}, impossibility results \cite{arxiv:AlaFef16,EC:AnaLaP21,C:ABDS21}, and limited classes of quantum computation \cite{braidobf,10.1007/978-3-030-88238-9_2,DBLP:conf/innovations/BartusekM22}. The best feasibility results we had prior to this work were for obfuscating quantum circuits with logarithmically many non-Clifford gates \cite{10.1007/978-3-030-88238-9_2} and for obfuscating ``null'' quantum circuits that always output zero \cite{DBLP:conf/innovations/BartusekM22}. Neither of these classes comes close to a notion of ``general-purpose'' quantum computation, and thus the feasibility of quantum obfuscation as a tool for quantum software protection has remained wide open. 


\paragraph{Results.} We consider the class of pseudo-deterministic quantum circuits, which are quantum circuits that take a classical input and produce a fixed classical output for each input with overwhelming probability. Essentially, these circuits compute a classical truth table, and can decide any language in (non-promise) BQP. This class captures Shor's algorithm \cite{10.1137/S0097539795293172}, which is arguably the quintessential algorithm for demonstrating the power of quantum computation over classical computation.

Our main result is the following. In the \emph{classical oracle model}, the evaluator (and adversary) are given oracle access to an efficiently computable classical functionality prepared by the obfuscator.

\begin{theorem}
Assuming the quantum hardness of Learning with Errors (QLWE), there exists a VBB obfuscator for any polynomial-size pseudo-deterministic quantum circuit $Q$ in the classical oracle model, where the obfuscated program is a quantum state $\ket{\widetilde{Q}}$.
\end{theorem}




\paragraph{On the classical oracle model.} The classical oracle model idealizes the notion of obfuscation for classical circuits, much like the random oracle model \cite{Bellare95randomoracles} idealizes a cryptographic hash function and the generic group model \cite{10.1007/3-540-69053-0_18} idealizes a cryptographic group. Such idealized primitives are typically not realizable in the real world, and a classical oracle is no exception. Indeed, virtual black-box (VBB) obfuscation comes for free in the classical oracle model, and it is known \cite{JACM:BGIRSVY12} that there exist (contrived) examples of circuits that provably cannot be VBB obfuscated (even with quantum information \cite{EC:AnaLaP21,C:ABDS21}). However, despite these contrived counterexamples, there is by now a fairly long history of establishing the feasibility of novel quantum-cryptographic primitives in the classical oracle model \cite{10.1145/2213977.2213983,arxiv:BenSat16,STOC:AGKZ20,C:ALLZZ21,DBLP:conf/innovations/BartusekM22}, and our result fits into this line of work.  

Moreover, \cite{JACM:BGIRSVY12} also defined a weaker notion of obfuscation called \emph{indistinguishability obfuscation}, which only requires that the obfuscations of two functionally equivalent programs are computationally indistinguishable, and which was subsequently shown by \cite{SIAMCOMP:SahWat21} and many follow-up works to be extremely powerful. Our main result is a construction of obfuscation for pseudo-deterministic quantum circuits from obfuscation of classical circuits (plus QLWE), and in order to prove security, we treat the classical obfuscation as implementing a black-box. However, one can interpret this result as \emph{heuristic} evidence that our construction gives indistinguishability obfuscation for pseudo-deterministic quantum circuits when the classical obfuscator is instantiated with a candidate post-quantum indistinguishability obfuscation scheme \cite{TCC:BGMZ18,C:CheVaiWee18,brakerski_et_al:LIPIcs.ICALP.2022.28,10.1145/3406325.3451070,EC:WeeWic21,TCC:DQVWW21}.

    


One can also appreciate our result in an oraclized world. Here, we show that, assuming QLWE,\footnote{In fact, it may be possible to remove the QLWE assumption entirely from our construction by showing that quantum fully-homomorphic encryption and dual-mode randomized trapdoor claw-free hash functions can be built from classical VBB obfuscation. We leave an exploration of this to future work.} it is possible to simulate access to a ``BQP oracle'' with just a P oracle. More precisely, the BQP oracle we implement can decide \emph{languages} in BQP, as opposed to more general promise problems.\footnote{This is because circuits for deciding promise problems are technically not pseudo-deterministic. There exist inputs that are neither yes or no instances, and thus are not guaranteed to produce a pseudo-deterministic output.}

\paragraph{Building blocks.} To obtain our main result of quantum obfuscation, we construct the following intermediate primitives that may be of independent interest.\\

\noindent\underline{Publicly-decodable Pauli functional commitments.}
\medskip \par We formally define the notion of a \emph{Pauli functional commitment}, which has appeared implicitly in many recent works, e.g. \cite{JACM:BCMVV21,SIAMCOMP:Mahadev22,Vid20-course}. These are bit commitment schemes that, when used in superposition to commit to a qubit, support opening the qubit to a measurement in either the standard or the Hadamard basis. While the only prior construction \cite{JACM:BCMVV21,SIAMCOMP:Mahadev22} of such commitments supports publicly-decodable standard basis measurements, security is completely compromised if the committer obtains access to the receiver's \emph{Hadamard} basis decoding functionality.

In this work, we describe a novel construction of Pauli functional commitments where security holds even if the committer obtains access to \emph{both} the receiver's standard and Hadamard basis decoding functionalities, and we argue security in the classical oracle model. Our construction is inspired by and unifies two lines of work: privately-decodable Pauli functional commitments \cite{JACM:BCMVV21,SIAMCOMP:Mahadev22}, and collision-resistant but equivocal hash functions \cite{FOCS:AmbRosUnr14,STOC:AGKZ20}.\\

\noindent\underline{Publicly-verifiable quantum fully-homomorphic encryption.}
\medskip \par Given the recent progress in constructing quantum homomorphic encryption schemes \cite{C:BroJef15,Dulek_2018,FOCS:Mahadev18b}, a natural question is whether the homomorphic evaluation procedure for these schemes can be \emph{verified} and with what resources.

The first work to address this question was \cite{AC:ADSS17}, who showed how to make the scheme of \cite{Dulek_2018} verifiable. Unfortunately, their verifier requires secret parameters, including the decryption key of the homomorphic encryption scheme. In a recent work, \cite{TCC:Bartusek21} showed how to obtain verifiable quantum FHE based on the scheme of \cite{FOCS:Mahadev18b}, with the following properties. The verifier is completely classical and doesn't require the decryption key of the FHE scheme, though it does require additional secret verification parameters.

In this work, we obtain the first feasibility result for \emph{publicly-verifiable} quantum fully-homomorphic encryption. Our protocol supports the classical verification of pseudo-deterministic quantum computation over the underlying plaintexts, and is proven sound in the classical oracle model. We also remark that the public parameters for our scheme are \emph{quantum} (while the verification is classical), and it is an interesting question for future work to see whether these public parameters (and, in turn, our obfuscated program) can be made completely classical.

\paragraph{Applications.} Program obfuscation has direct applications to software protection, and our results indicate that such protections may be possible to achieve in the context of quantum software. Obfuscated programs intuitively cannot be reverse-engineered, meaning that we can now protect any intellectual property or other secret information contained in the implementation of the quantum program.

In the classical and post-quantum settings, obfuscation has also been identified as a useful tool for digital watermarking \cite{JACM:BGIRSVY12,SIAMCOMP:CHNVW18,EC:KitNis22}, which allows for embedding an unremovable ``mark'' into a program, and acts as a deterrent against software piracy. Quantum information potentially allows for much \emph{stronger} forms of protection against piracy, enabling computation to be encoded into a quantum state that provably cannot be copied \cite{5231194}. However, the scope of such ``copy-protection'' schemes has so far been limited to classical functionalities~\cite{arxiv:ColMajPor20,C:ALLZZ21,C:CLLZ21,TCC:AnaKal21,C:AKLLZ22,AC:KitNis22}. In \cref{subsec:copy-protection}, we sketch how our obfuscation scheme results in a candidate for copy-protection of (unlearnable) \emph{quantum} programs, following the construction of \cite{C:ALLZZ21}. 


Another common application of obfuscation in the classical setting is to advanced forms of encryption, such as functional encryption \cite{SIAMCOMP:GGHRSW16}. In \cref{subsec:FE}, we sketch an application of our construction to functional encryption for \emph{quantum functionalities}. 

That being said, we stress that the main focus of our work is on the construction of quantum obfuscation, and we leave a more in-depth exploration of applications to future work. 

\paragraph{Open problems.} Our work raises many interesting questions on the topic of quantum obfuscation. One immediate question is whether it is possible to obfuscate \emph{all} quantum circuits with classical input and output, extending our result for pseudo-deterministic circuits. That is, can circuits that output an arbitrary distribution over classical strings be obfuscated? As explained in \cref{sec:tech-overview}, we follow the approach of \cite{DBLP:conf/innovations/BartusekM22} who consider obfuscating the verifier of an appropriate classical verification of computation protocol \cite{SIAMCOMP:Mahadev22}. Unfortunately, it is not known how to classically verify general quantum sampling circuits, at least with negligible soundness (the work of \cite{EC:CLLW22} provides a solution with weaker soundness). This appears to be one barrier for extending our approach to all quantum circuits with classical input and output.

One can also wonder about the possibility of obfuscating general quantum operations over quantum registers. That is, while our scheme is able to obfuscate quantum computation, it still only implements a ``classical'' language, albeit one whose truth table may only be (known to be) computable with a quantum circuit. Thus, this leaves open the feasibility (or impossibility) of implementing \emph{quantum} oracles, and we consider this to be a very interesting question to understand in future work.

Finally, we mention two natural open questions regarding our construction itself. First, is it possible to remove the quantum states from our construction and obtain a \emph{classical} obfuscated program? Next, can we improve on the heuristic nature of our security proof, and obtain indistinguishability obfuscation for pseudo-deterministic quantum circuits from the assumption of indistinguishability obfuscation for classical circuits?


\section{Technical Overview}\label{sec:tech-overview}

In this overview, we will describe how to obfuscate any pseudo-deterministic quantum circuit $Q$, where pseudo-deterministic means that for each input $x$ there exists an output $y$ such that $\Pr[Q(x) \to y] = 1-\negl$. That is, we describe a compiler that given the classical description of $Q$, produces an obfuscated program $\ket{\widetilde{Q}}$ that reveals as little as possible about the description of $Q$ while preserving the functionality of $Q$. Throughout this overview, we will treat such circuits as fully deterministic, associating a well-defined bit $y \coloneqq Q(x)$ to each input $x$, which has a negligible effect on our arguments. 

\subsection{Our approach: Verifying quantum partitioning circuits}

\paragraph{Fully-homomorphic encryption.} A natural approach to obfuscation involves the notion of \emph{fully-homomorphic encryption} (FHE), which allows for encoding data $x$ into a ciphertext $\Enc(x)$ so that anyone holding $\Enc(x)$ and a function $f$ can produce a ciphertext $\Enc(f(x))$. Indeed, given an FHE scheme that supports the evaluation of \emph{quantum functionalities} \cite{FOCS:Mahadev18b}, one could release an encryption $\Enc(Q)$ of the description of $Q$. Then, any evaluator with an input $x$ can obtain $\Enc(Q(x))$ by running an appropriate evaluation procedure.

This comes close to a working obfuscation scheme, except that the evaluator obtains $\Enc(Q(x))$ rather than the output $Q(x)$ in the clear. To fix this, we cannot simply release the FHE secret key $\sk$, allowing the evaluator to decrypt $\Enc(Q(x))$ and learn $Q(x)$, because this would \emph{also} allow the evaluator to decrypt $\Enc(Q)$ and learn the description of $Q$. Instead, we could release a carefully ``broken'' secret key that \emph{only} allows decryption of ciphertexts $\Enc(Q(x))$ that encrypt an honestly evaluated output $Q(x)$.

\paragraph{Reducing to classical obfuscation.} But how can we obtain such a carefully broken key? One attempt would be to release an obfuscation of the following program $C$, which has the secret key $\sk$ and the ciphertext $\Enc(Q)$ hard-coded,

\[C[\sk,\Enc(Q)](x,\ct): ~~ \text{if} ~ \Eval[\Enc(Q)](x) \to \ct, ~ \text{output } \Dec(\sk,\ct), ~ \text{and otherwise output } \bot,\] where $\Eval[\Enc(Q)](\cdot)$ is the FHE evaluation circuit that on input $x$ outputs $\Enc(Q(x))$. However, we don't know how to obfuscate $C$ since $\Eval[\Enc(Q)](\cdot)$ is a quantum circuit. 

Instead, building on observations by \cite{DBLP:conf/innovations/BartusekM22}, we could hope to construct an argument system with a \emph{classical} verifier $V$ that satisfies the following properties.

\begin{itemize}
    \item For any $x$, one can compute a ciphertext $\ct$ and a proof $\pi$ such that $V(\Enc(Q),x,\ct,\pi) = 1$.
    \item It is hard to find $(x,\ct,\pi)$ such that $V(\Enc(Q),x,\ct,\pi) = 1$ and $\Dec(\sk,\ct) \neq Q(x)$.
\end{itemize}

If such a system existed, we could instead obfuscate the following \emph{classical} program

\[\widetilde{C}[\sk,\Enc(Q)](x,\ct,\pi): ~~ \text{if} ~~ V(\Enc(Q),x,\ct,\pi) \to 1, ~ \text{output } \Dec(\sk,\ct), ~ \text{and otherwise output } \bot.\]

Crucially, this approach follows the ``verify-then-decrypt'' paradigm, where the output ciphertext is \emph{first} verified to be honest, and only then decrypted using $\sk$. A procedure that first decrypts and then verifies may not be secure since the adversary could submit dishonest ciphertexts to learn information about $\sk$.

\paragraph{Classical verification of quantum computation and its limitations.} Thus, it suffices to construct a classically-verifiable argument system for the class of quantum circuits $\Eval[\Enc(Q)]$ that take \[\Eval[\Enc(Q)](x) \to \Enc(Q(x)),\] where $\Enc$ is a quantum fully-homomorphic encryption (QFHE) scheme and $Q$ is a deterministic quantum circuit. 

As mentioned earlier, \cite{SIAMCOMP:Mahadev22} did construct a protocol for classical verification of quantum computation. Unfortunately, there are two major problems with using \cite{SIAMCOMP:Mahadev22}'s scheme for this application.

\begin{itemize}
    \item \textbf{Sampling circuits.} \cite{SIAMCOMP:Mahadev22}'s scheme only supports verification of (pseudo)-deterministic quantum circuits. However, the evaluation procedure of known QFHE schemes \cite{FOCS:Mahadev18b,C:Brakerski18} is inherently \emph{randomized}, even if the underlying computation is deterministic, meaning that the circuit that we would like to verify actually produces a \emph{sample} $\Enc(Q(x))$ from a classical distribution over ciphertexts.\footnote{While this distribution is only supported on ciphertexts that encrypt the correct output bit $Q(x)$, the random coins used for the output ciphertext will vary.}
    \item \textbf{Public verifiability.} Note that the evaluator will have (obfuscated) access to the verification function, which means that it can repeatedly query the verifier with proofs of its choice. If soundness holds even when verification is \emph{public}, then the evaluator cannot break soundness using access to this oracle. However, \cite{SIAMCOMP:Mahadev22}'s scheme is privately-verifiable, and can be broken given repeated access to the verifier.
\end{itemize}

Towards solving the first problem, \cite{EC:CLLW22} presented a scheme for classical verification of sampling circuits, though only with inverse polynomial soundness error. While interesting on its own, this renders the scheme difficult to use for our application, since a polynomial-time evaluator can eventually break soundness and thus break security of the obfuscation scheme. It appears that improving upon their result to obtain negligible soundness for classical verification of quantum sampling circuits is difficult, and could be considered a major open problem.

\paragraph{Quantum partitioning circuits.} Instead, we relax our goal. We observe that if $Q$ is deterministic, then we don't need the full power of verification of sampling circuits to verify the sampling of $\Enc(Q(x))$. Indeed, we can \emph{partition} the output space of $\Eval[\Enc(Q)](\cdot)$ into ciphertexts $\ct_0$ that decrypt to 0 and ciphertexts $\ct_1$ that decrypt to 1. Thus, each input $x$ outputs a sample from one of these two sets. That is, we can define a classical predicate $P \coloneqq \Dec(\sk,\cdot)$ such that $P(\Eval[\Enc(Q)](\cdot))$ is (pseudo)-deterministic. 

Thus, we say that $Q$ is a \emph{quantum partitioning circuit} if there exists a predicate $P$ such that $P(Q(\cdot))$ is pseudo-deterministic, and we investigate the feasibility of obtaining a classically-verifiable argument system for such partitioning circuits. Crucially for our application, the prover in the argument system cannot depend on $P$ since $P$ will contain the description of the FHE secret key.\footnote{And otherwise, this notion would trivially reduce to classical verification of pseudo-deterministic quantum circuits.} Then, we will need an argument system with (roughly) the following syntax (see \cref{subsec:partitioning-def} for a formal description).

\begin{itemize}
    \item $\Gen(1^\secp,Q) \to \pp$: The parameter generation algorithm outputs public parameters $\pp$. We allow $\pp$ to contain the description of a \emph{classical oracle}, and refer to such a protocol as being \emph{in the oracle model}. 
    \item $\Prove(\pp,Q,x) \to \pi$: The prover algorithm outputs a proof $\pi$.
    \item $\Ver(\pp,Q,x,\pi) \to q \union \{\bot\}$: The verifier checks if the proof is valid, and if so outputs a classical string $q$. 
    \item $\Out(q,P) \to b$: The output algorithm takes $q$ and the description of a predicate $P$ and outputs a bit $b$.
\end{itemize}

For soundness, we require that no computationally bounded prover can produce an $(x,\pi)$ such that $\Ver(\pp,Q,x,\pi) \to q$ and $\Out(q,P) \neq P(Q(x))$. We refer to such a protocol as a \emph{non-interactive publicly-verifiable classical verification of quantum partitioning circuits.} In \cref{sec:obfuscation}, we follow the intuition given above, and show formally how to use this type of argument system along with QFHE and VBB obfuscation of classical circuits (which is used to obfuscate the classical oracle in $\pp$) to obfuscate pseudo-deterministic quantum circuits. 

In the remainder of this overview, we will describe how to construct non-interactive publicly-verifiable classical verification of quantum partitioning circuits in the oracle model.

\subsection{Prior work: One-time soundness}

Building on \cite{SIAMCOMP:Mahadev22,EC:CLLW22}, the prior work of \cite{TCC:Bartusek21} shows how to construct non-interactive \emph{privately-verifiable} classical verification of quantum partitioning circuits,\footnote{In \cite{TCC:Bartusek21}, quantum partitioning circuits were referred to as ``quantum-classical'' circuits.} where soundness breaks down if the prover is given oracle access to the verification functionality. We refer to this security as ``one-time soundness''. We will eventually build on top of this protocol in two steps.\footnote{Breaking this into two steps is only for the purpose of the overview. In \cref{subsec:public-verification}, we perform both steps simultaneously.}
\begin{enumerate}
    \item We will first show how to obtain reusable soundness against provers that can access the verification oracle in a limited ``single instance'' setting. In this setting, there is only one input $x$ that the verification oracle will accept.
    \item We will upgrade this protocol to the fully reusable setting, thus obtaining a \emph{publicly-verifiable} protocol in the oracle model.
\end{enumerate}

In this section, we describe the protocol of \cite{TCC:Bartusek21} in some detail, as our construction will use these internal details. However, before getting into the protocol, we describe a useful abstraction that is novel to this work: a \emph{Pauli functional commitment}. We will then describe \cite{TCC:Bartusek21}'s protocol using the language of Pauli functional commitments, and, later in the overview, show how a new variation on the notion of Pauli functional commitments will be integral to our final construction.

\paragraph{Pauli functional commitments.} Bit commitment schemes traditionally satisfy a notion of binding and a notion of hiding. A \emph{functional} commitment scheme includes an additional notion of functionality, which allows the committer to open its commitment to some \emph{function} of the committed message, up to some limitations imposed by the binding property.



A Pauli functional commitment ($\PFC$) is a traditional (non-interactive) classical bit commitment scheme augmented with a particular quantum functionality property. Note that any classical bit commitment algorithm $\Com(\ck,b) \to (b,u,c),$ where $\ck$ is the commitment key, $u$ is opening information, and $c$ is the commitment string, can be used to commit to a qubit $\alpha_0\ket{0} + \alpha_1\ket{1}$ in superposition. If the commitment scheme is perfectly hiding, then measuring a commitment string $c$ would leave a remaining state of the form $\alpha_0\ket{0}\ket{u_0} + \alpha_1\ket{1}\ket{u_1}$,\footnote{Note that depending on the commitment scheme, the second register may contain a superposition over random coins / opening information.} which preserves the original qubit. A Pauli functional commitment enables the committer to then ``open'' its state to \emph{either} a standard basis measurement \emph{or} a Hadamard basis measurement of its original qubit $\alpha_0\ket{0} + \alpha_1\ket{1}$. More formally, it should satisfy the following syntax.

\begin{itemize}
    \item $\Gen(1^\secp) \to (\ck,\dk)$: $\Gen$ outputs a commitment key $\ck$ and a decoding key $\dk$.\footnote{For now, assume $\ck$ is classical, though later we will consider commitments with quantum commitment keys.}
    \item $\Com(\ck,\regB) \to (\regB,\regU,c)$: $\Com$ takes as input a single-qubit register $\regB$ and produces a classical commitment $c$ along with registers $(\regB,\regU)$, where $\regU$ holds opening information.\footnote{Whenever we say that an algorithm takes as input or outputs a register, we mean that it operates on a quantum state stored on that register.}
    \item $\OpenZ(\regB,\regU) \to u$: The standard basis opening algorithm performs a measurement on registers $(\regB,\regU)$ to produce a classical string $u$.
    \item $\OpenX(\regB,\regU) \to u$: The Hadamard basis opening algorithm performs a measurement on registers $(\regB,\regU)$ to produce a classical string $u$.
    \item $\DecZ(\dk,c,u) \to \{0,1,\bot\}$: The standard basis decoding algorithm takes the decoding key $\dk$, a commitment $c$, an opening $u$, and either decodes a bit 0 or 1, or outputs $\bot$.
    \item $\DecX(\dk,c,u) \to \{0,1,\bot\}$: The Hadamard basis decoding algorithm takes the decoding key $\dk$, a commitment $c$, an opening $u$, and either decodes a bit 0 or 1, or outputs $\bot$.
\end{itemize}

A Pauli functional commitment should satisfy \emph{functionality} as described above and some notion (depending on the application) of \emph{binding} to a classical bit. That is, binding is defined with respect to the bit output by the $\DecZ$ algorithm. A notion of hiding does not need to be explicitly considered - the properties of functionality and binding are already enough to make this primitive both non-trivial and useful.

We note that this notion has appeared implicitly in many previous works, e.g. \cite{JACM:BCMVV21,SIAMCOMP:Mahadev22,Vid20-course}. Indeed, \cite{JACM:BCMVV21,SIAMCOMP:Mahadev22} essentially showed how to construct a Pauli functional commitment that simultaneously satisfies \emph{two} binding properties from the quantum hardness of learning with errors (QLWE). In our own words, these properties are the following. 

\begin{itemize}
    \item \textbf{Dual-mode.} $\Gen(1^\secp,h)$ now takes as input a bit $h$ indicating the ``mode'', where $h=1$ is the regular mode, and $h=0$ is a \emph{perfectly binding} mode. In perfectly binding mode, for every commitment $c$ there is at most one bit $b$ such that there exists an opening $u$ with $\DecZ(\dk,c,u) = b$. This mode allows for the definition of an algorithm $\Invert(\dk,c) \to b$ that outputs the bit $b$ such that there exists $u$ with $\DecZ(\dk,c,u) = b$ (or outputs $\bot$ if such a $b$ does not exist). Importantly, the $\ck$ output on $h=0$ vs $h=1$ must be computationally indistinguishable.
    \item \textbf{Uncertainty.} For any polynomial-time adversary that outputs $(c,b,u_Z,u_X)$, it holds that 
    \begin{align*}&\Pr[\DecZ(\dk,c,u_Z) = b \ \wedge \ \DecX(\dk,c,u_X) = 0] \\ &\approx \Pr[\DecZ(\dk,c,u_Z) = b \wedge \DecX(\dk,c,u_X) = 1].\end{align*}  That is, if an adversary opens successfully to a standard basis measurement of its committed state, the Hadamard basis measurement is maximally uncertain. Note that this can be considered a binding property for the classical bit $b$ since the ability to measure in the Hadamard basis implies the ability to reflect across the Hadamard basis axis, thus influencing the standard basis measurement.
\end{itemize}

More precisely, prior work has shown how to construct a $\PFC$ satisfying the above binding properties from (what we call) a \emph{dual-mode randomized trapdoor claw-free hash function} with an \emph{adaptive hard-code bit} property. We refer to this primitive as a ``Type I'' $\PFC$ or $\PFCI$, in order to differentiate it from a ``Type II'' $\PFC$ that we will construct in this work. We also note that in the body of this work, we build our protocols directly from the underlying claw-free hash function, so that we can appeal to theorems from prior work.\footnote{However, we believe it could be interesting to re-prove prior results using the notion of $\PFC$, and we leave an exploration of this possibility to future work. That is, does a $\PFC$ that satisfies the dual-mode and uncertainty binding properties generically imply classical verification of quantum computation?} Thus the primitive of $\PFCI$ does not appear explicitly in the body. However, in the remainder of this overview, we find it more convenient to explain these protocols using the primitive of $\PFCI$. 

\paragraph{Verification of quantum partitioning circuits with one-time soundness.} Now, we describe a privately-verifiable scheme for classical verification of quantum partitioning circuits that follows from prior work \cite{SIAMCOMP:Mahadev22,EC:CLLW22,TCC:Bartusek21}.

The starting point is a particular way to prepare a history state $\ket{\psi_{Q,x}}$ of the computation $Q(x)$, due to \cite{EC:CLLW22}. Given $\ket{\psi_{Q,x}}$, the verifier can either measure certain registers in the standard basis to obtain an approximate sample $q \gets Q(x)$, or measure a random local Hamiltonian term (which involves just standard basis and Hadamard basis measurements). In \cite{TCC:Bartusek21}, the prover is instructed to prepare multiple copies of the history state, and the verifier chooses some subset for \emph{sampling} (obtaining an output sample) and the other subset for \emph{verifying} (measuring a local Hamiltonian term). If verification passes, the verifier collects the output samples $\{q_t\}_t$ and outputs the bit $b \coloneqq \mathsf{Maj}\left(\{P(q_t)\}_t\right)$, which should be equal to $P(Q(x))$ with overwhelming probability. 

Combining this approach with \cite{SIAMCOMP:Mahadev22}'s measurement protocol, applying parallel repetition, and finally applying Fiat-Shamir, we obtain the protocol described in \cref{fig:one-time}.

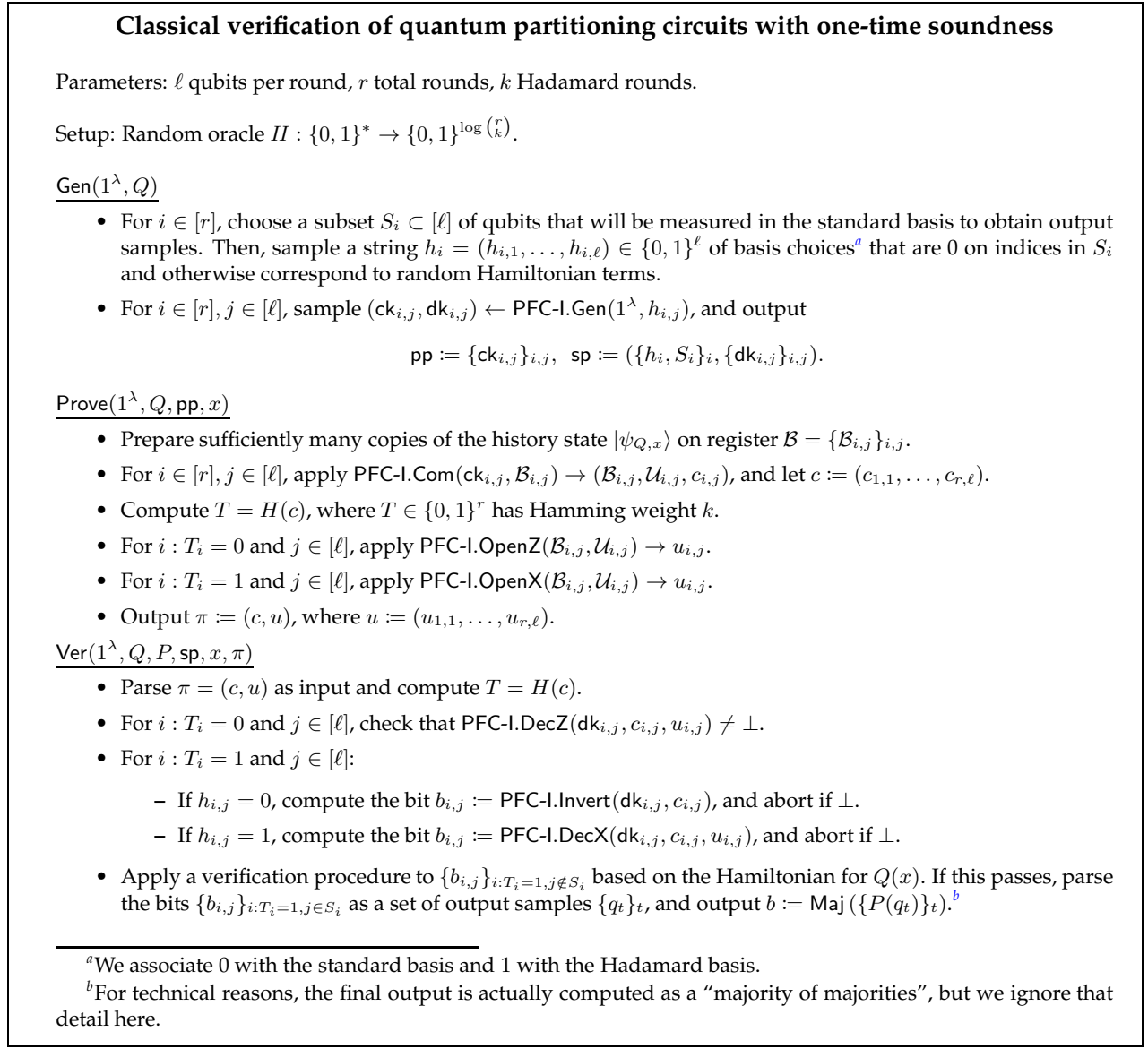
\protocol{Classical verification of quantum partitioning circuits with one-time soundness}{A non-interactive privately-verifiable protocol for classical verification of quantum partitioning circuits, due to \cite{SIAMCOMP:Mahadev22,EC:CLLW22,TCC:Bartusek21}. The circuit $Q$ and predicate $P$ are such that $P(Q(\cdot))$ is a pseudo-deterministic circuit.  
}{fig:one-time}{
Parameters: $\ell$ qubits per round, $r$ total rounds, $k$ Hadamard rounds.\\ 

Setup: Random oracle $H: \{0,1\}^* \to \{0,1\}^{\log\binom{r}{k}}$.\\


\underline{$\Gen(1^\secp,Q)$}
\begin{itemize}
    \item For $i \in [r]$, choose a subset $S_i \subset [\ell]$ of qubits that will be measured in the standard basis to obtain output samples. Then, sample a string $h_i = (h_{i,1},\dots,h_{i,\ell}) \in \{0,1\}^{\ell}$ of basis choices\footnote{We associate 0 with the standard basis and 1 with the Hadamard basis.} that are 0 on indices in $S_i$ and otherwise correspond to random Hamiltonian terms.
    \item For $i \in [r], j \in [\ell]$, sample $(\ck_{i,j},\dk_{i,j}) \gets \PFCI.\Gen(1^\secp,h_{i,j})$, and output
    \[\pp \coloneqq \{\ck_{i,j}\}_{i,j}, ~~ \sparam \coloneqq (\{h_i,S_i\}_i,\{\dk_{i,j}\}_{i,j}).\]
\end{itemize}
\underline{$\Prove(1^\secp,Q,\pp,x)$}
\begin{itemize}
    \item Prepare sufficiently many copies of the history state $\ket{\psi_{Q,x}}$ on register $\regB = \{\regB_{i,j}\}_{i,j}$.
    \item For $i \in [r], j \in [\ell]$, apply $\PFCI.\Com(\ck_{i,j},\regB_{i,j}) \to (\regB_{i,j},\regU_{i,j},c_{i,j})$, and let $c \coloneqq (c_{1,1},\dots,c_{r, \ell})$.
    \item Compute $T = H(c)$, where $T \in \{0,1\}^r$ has Hamming weight $k$.
    \item For $i : T_i = 0$ and $j \in [\ell]$, apply $\PFCI.\OpenZ(\regB_{i,j},\regU_{i,j}) \to u_{i,j}$.
    \item For $i : T_i = 1$ and $j \in [\ell]$, apply $\PFCI.\OpenX(\regB_{i,j},\regU_{i,j}) \to u_{i,j}$.
    \item Output $\pi \coloneqq (c,u)$, where $u \coloneqq (u_{1,1},\dots,u_{r, \ell})$.
\end{itemize}
\underline{$\Ver(1^\secp,Q,P,\sparam,x,\pi)$}
\begin{itemize}
    \item Parse $\pi = (c,u)$ as input and compute $T = H(c)$.
    \item For $i : T_i = 0$ and $j \in [\ell]$, check that $\PFCI.\DecZ(\dk_{i,j},c_{i,j},u_{i,j}) \neq \bot$.
    \item For $i : T_i = 1$ and $j \in [\ell]$:
    \begin{itemize}
        \item If $h_{i,j} = 0$, compute the bit $b_{i,j} \coloneqq \PFCI.\Invert(\dk_{i,j},c_{i,j})$, and abort if $\bot$.
        \item If $h_{i,j} = 1$, compute the bit $b_{i,j} \coloneqq \PFCI.\DecX(\dk_{i,j},c_{i,j},u_{i,j})$, and abort if $\bot$.
    \end{itemize}
    \item Apply a verification procedure to $\{b_{i,j}\}_{i: T_i = 1, j \notin S_i}$ based on the Hamiltonian for $Q(x)$. If this passes, parse the bits $\{b_{i,j}\}_{i: T_i = 1, j \in S_i}$ as a set of output samples $\{q_t\}_t$, and output $b \coloneqq \mathsf{Maj}\left(\{P(q_t)\}_t\right)$.\footnote{For technical reasons, the final output is actually computed as a ``majority of majorities'', but we ignore that detail here.}
\end{itemize}
}

In more detail, \cref{fig:one-time} consists of a number $r$ of parallel rounds, where $k$ of them are denoted ``Hadamard'' rounds, and the rest are denoted ``test'' rounds. Which rounds are Hadamard rounds are determined by a random oracle $H$ applied to the prover's Pauli functional commitments $c$. 

Each Hadamard round essentially runs a copy of the protocol described above, where the verifier obtains a number of output samples. We let $\ell$ denote the number of qubits per round, which is the number of history states per round times the number of qubits per history state. The standard basis measurements are obtained by inverting the commitments themselves (since these commitments are generated in mode $h = 0$), and the Hadamard basis measurements are obtained via the $\OpenX$ procedure. On the other hand, in the test rounds, the prover opens all of their commitments using the $\OpenZ$ procedure, and the verifier simply checks that $\DecZ$ does not reject these openings. We also note that the public and secret parameters $(\pp,\sparam)$ are generated independently of the input $x$, which was shown to be possible by an observation of \cite{TCC:ACGH20}.\footnote{Technically, $\Gen$ just needs to know the size of $Q$.}

The one-time soundness of this protocol was proven in \cite{TCC:Bartusek21}, and relies on the soundness of the underlying measurement protocol due to \cite{SIAMCOMP:Mahadev22}. While the proof in \cite{SIAMCOMP:Mahadev22} actually required an additional property of the claw-free hash function beyond dual-mode and adaptive hard-core bit, the recent work of \cite{DBLP:conf/crypto/BartusekKLMMVVY22} showed that these two properties, which correspond to the dual-mode and uncertainty properties of the $\PFC$, suffice for proving soundness.

\paragraph{Challenges with reusability.} Now, our goal is to obtain soundness even against provers that have (superposition) oracle access to the verification algorithm. We denote this algorithm $\Ver[\sparam](\cdot,\cdot)$, which has the secret parameters $\sparam$ hard-coded (and implicitly $1^\secp,Q,$ and $P$), expects $(x,\pi)$ as input, and outputs either a bit $b$ or $\bot$. 

Unfortunately, there is a simple attack on soundness in this setting. The main issue is that the secret parameters $\sparam$ hard-code the measurement bases $h = (h_1,\dots,h_r)$, and soundness of the underlying information-theoretic protocol would be completely compromised if the prover could figure out $h$. Note that in the Hadamard rounds, the strings $u_{i,j}$ corresponding to $h_{i,j} = 0$ are completely ignored by the verifier, while the strings $u_{i,j}$ corresponding to $h_{i,j} = 1$ factor into the verifier's response. This discrepancy provides a way for the prover to learn the bits of $h_{i,j}$ by querying the verifier multiple times, ultimately breaking soundness of the protocol (see \cite{DBLP:conf/innovations/BartusekM22} for a more detailed discussion of this issue).

\paragraph{Can signature tokens help?} Before coming to our solution, we discuss one promising but flawed attempt at upgrading to reusable soundness via the primitive of \emph{signature tokens} \cite{arxiv:BenSat16}. A signature token consists of a quantum signing $\ket{\sk}$ that can be used to sign a \emph{single} arbitrary message $x$, and then becomes useless. 

So suppose we included $\ket{\sk}$ in the public parameters, and ask that the prover sign its proof $\pi$ before querying $\Ver[\sparam]$. That is, $\Ver[\sparam]$ will now take as input $(x,\pi,\sigma)$, and only respond if $\sigma$ is a valid signature on $\pi$. Intuitively, if the prover tries to start collecting information from multiple malformed proofs in order to learn enough bits of $h$ to break soundness, they should fail to produce the multiple signatures required to learn this information. 

Unfortunately, this intuition is false. First, since the prover has \emph{superposition} access to the verifier, they never have to actually output a classical signature $\sigma$. Moreover, in known signature token schemes \cite{arxiv:BenSat16}, the public parameters can be used to implement a projection $\ketbra{\sk}{\sk}$ onto the original signing key. Thus, even though a prover may ``damage'' its state $\ket{\sk}$ by querying $\Ver[\sparam]$ in superposition in order to learn a single bit of information about $h$, they could then project back onto $\ket{\sk}$ via amplitude amplification. Thus, they could launch the same attacks as before, ultimately learning enough about $h$ to break soundness.
 
\subsection{Reusable soundness for a single instance}

Classically, the following is a common route for boosting one-time soundness to reusable soundness for, say, an NP argument system. Note that any \emph{fixed} instance $x$, either $x$ is a yes instance, so we don't have to worry about the prover breaking soundness with respect to $x$, or $x$ is a no instance, so by the one-time soundness of the protocol, the prover should never be able to make the verification oracle accept, rendering it useless. Thus, we can obtain reusable soundness if each instance $x$ was associated with its own pair of public and secret parameters $(\pp_x,\sparam_x)$. One method for achieving this is to fix the actual public parameters as an obfuscation of a program that takes $x$ as input and samples parameters $(\pp_x,\sparam_x)$ using randomness derived from a PRF applied to $x$ (see \cite{DBLP:conf/stoc/BitanskyGLPT15} for an example).


Although we would like to follow this approach, one difficulty is that in our setting the notion of an ``instance'' is unclear. The inputs $x$ to the circuit cannot be classified into yes and no instances, since they all produce some valid outputs. In particular, note that the attacks on reusability outlined above will work even if the prover always queries the verification oracle on the same input $x$, eventually producing a $\pi$ that causes the verifier to output $b \neq P(Q(x))$. A next attempt would be to start with some input $x$, sample $q \gets Q(x)$, and consider the pair $(x,q)$ to be an instance. However, since $Q$ is a sampling circuit, it may be the case that this particular $q$ is only sampled with small, or even negligible, probability on input $x$. Our one-time sound scheme is not equipped to prove a statement of the form, ``$q$ is in the support of the output of $Q(x)$''. Thus, we will need a different approach.

\paragraph{Committing to the history state.} Given an input $x$, we will essentially classify the \emph{history state} of the computation $Q(x)$ into ``yes'' and ``no'' instances. That is, an honestly prepared history state $\ket{\psi_{Q,x}}$ should be classified as a yes instance, while any large enough perturbation to $\ket{\psi_{Q,x}}$ should be classified as a no instance. However, looking ahead, it will be crucial that our instances are classical so that we can generate parameters by applying a PRF to the instance. Thus, what we really need is a \emph{classical commitment} to the history state. Moreover, after the state is committed, we still need it to be available for the prover to use in the one-time sound scheme. Fortunately, the prover only needs to perform standard and Hadamard basis measurements on the state (in addition to some operations that are classically controlled on the state). Thus, we have already discussed the exact primitive that we need - a Pauli functional commitment!

In \cref{fig:single-instance}, we outline a protocol where an instance $(x,\widetilde{c})$, consisting of an input $x$ and a commitment $\widetilde{c}$ to a set of history states $\ket{\psi_{Q,x}}$, is generated and fixed before the protocol begins. We use a Pauli functional commitment denoted $\PFCII$ to commit to the history states (since we will eventually require $\PFCII$ to satisfy different properties than $\PFCI$). 

We remark that correctness of this protocol relies on a couple of specific properties: (1) $\PFCI.\Com$ and $\PFCII.\Com$ are both \emph{classically controlled} on the register $\regB$, so they commute with each other, and (2) $\PFCI.\OpenZ$ (resp. $\PFCI.\OpenX$) simply measures the register $\regB$ in the standard (resp. Hadamard) basis\footnote{Though it could be performing an arbitrary operation to the $\regU$ register.} so the first bit of the string $u$ can be computed instead by applying $\PFCII.\Com$ to $\regB$ followed by $\PFCII.\OpenZ$ and $\PFCII.\DecZ$ (resp. $\PFCII.\OpenX$ and $\PFCII.\DecX$). 

Now, our goal will be to obtain reusable soundness for any fixed instance $(x,\widetilde{c})$. That is, we give the prover oracle access to $\Ver[\sparam,\widetilde{\dk},(x,\widetilde{c})](\cdot)$ where $\widetilde{\dk}$ and $(x,\widetilde{c})$ are now hard-coded and the only input is a proof $\pi$, and require that the prover cannot make the verifier output $b \neq P(Q(x))$.

\protocol{A protocol with reusable soundness for a single ``instance''}{A protocol for classical verification of quantum partitioning circuits that is reusably sound for each fixed instance $(x,\widetilde{c})$.}{fig:single-instance}{
Parameters: $\ell$ qubits per round, $r$ total rounds, $k$ Hadamard rounds.\\

Setup: Random oracle $H: \{0,1\}^* \to \{0,1\}^{\log\binom{r}{k}}$.\\

\underline{Instance generation}
\begin{itemize}
    \item \textcolor{red}{For $i \in [r], j \in [\ell]$, the verifier samples $(\widetilde{\ck}_{i,j},\widetilde{\dk}_{i,j}) \gets \PFCII.\Gen(1^\secp)$, outputs $\widetilde{\ck} \coloneqq \{\widetilde{\ck}_{i,j}\}_{i,j}$, and keeps $\widetilde{\dk} \coloneqq \{\widetilde{\dk}_{i,j}\}_{i,j}$ private.}
    \item Given an input $x$, the prover prepares sufficiently many copies of the history state $\ket{\psi_{Q,x}}$ on register $\regB = \{\regB_{i,j}\}_{i,j}$. 
    \item \textcolor{red}{For $i \in [r]$, $j \in [\ell]$, the prover applies $\PFCII.\Com(\widetilde{\ck}_{i,j},\regB_{i,j}) \to (\regB_{i,j},\widetilde{\regU}_{i,j},\widetilde{c}_{i,j})$. Then, it sets $\widetilde{c} \coloneqq (\widetilde{c}_{1,1},\dots,\widetilde{c}_{r, \ell})$ and outputs the instance $(x,\widetilde{c}).$}
\end{itemize}

\underline{$\Gen(1^\secp,Q)$}
\begin{itemize}
    \item The verifier samples $\pp = \{\ck_{i,j}\}_{i,j}, \sparam = (\{h_i,S_i\}_i,\{\dk_{i,j}\}_{i,j})$ as in \cref{fig:one-time}.
\end{itemize}

\underline{$\Prove(1^\secp,Q,\pp,x)$}
\begin{itemize}
    \item For $i \in [r], j \in [\ell]$, apply $\PFCI.\Com(\ck_{i,j},\regB_{i,j}) \to (\regB_{i,j},\regU_{i,j},c_{i,j})$, and let $c \coloneqq (c_{1,1},\dots,c_{r, \ell})$.
    \item Compute $T = H(c)$, where $T \in \{0,1\}^r$ has Hamming weight $k$.
    \item For $i : T_i = 0$ and $j \in [\ell]$, \textcolor{red}{apply $\PFCII.\OpenZ(\regB_{i,j},\widetilde{\regU}_{i,j}) \to \widetilde{u}_{i,j}$} followed by $\PFCI.\OpenZ(\regB_{i,j},\regU_{i,j}) \to u_{i,j}$. Let $u'_{i,j}$ be $u_{i,j}$ with the first bit removed.
    \item For $i : T_i = 1$ and $j \in [\ell]$, \textcolor{red}{apply $\PFCII.\OpenX(\regB_{i,j},\widetilde{\regU}_{i,j}) \to \widetilde{u}_{i,j}$} followed by $\PFCI.\OpenX(\regB_{i,j},\regU_{i,j}) \to u_{i,j}$. Let $u'_{i,j}$ be $u_{i,j}$ with the first bit removed.
    \item Output $\pi \coloneqq (c,\textcolor{red}{\widetilde{u}},u)$, where $\widetilde{u} \coloneqq (\widetilde{u}_{1,1},\dots,\widetilde{u}_{r, \ell})$ and $u \coloneqq (u'_{1,1},\dots,u'_{r, \ell})$.
\end{itemize}

\underline{$\Ver(1^\secp,Q,P,\sparam,\widetilde{\dk},(x,\widetilde{c}),\pi)$}
\begin{itemize}
    \item Parse $\pi = (c,\textcolor{red}{\widetilde{u}},u)$ and compute $T = H(c)$.
    \item For $i : T_i = 0$ and $j \in [\ell]$, \textcolor{red}{compute $b'_{i,j} \coloneqq \PFCII.\DecZ(\widetilde{\dk}_{i,j},\widetilde{c}_{i,j},\widetilde{u}_{i,j})$} and check that  $\PFCI.\DecZ(\dk_{i,j},c_{i,j},(\textcolor{red}{b'_{i,j}},u'_{i,j})) \neq \bot$.
    \item For $i : T_i = 1$ and $j \in [\ell]$:
    \begin{itemize}
        \item If $h_{i,j} = 0$, compute the bit $b_{i,j} \coloneqq \PFCI.\Invert(\dk_{i,j},c_{i,j})$, and abort if $\bot$.
        \item If $h_{i,j} = 1$, \textcolor{red}{compute $b'_{i,j} \coloneqq \PFCII.\DecX(\widetilde{\dk}_{i,j},\widetilde{c}_{i,j},\widetilde{u}_{i,j})$,} followed by the bit $b_{i,j} \coloneqq \PFCI.\DecX(\dk_{i,j},c_{i,j},(\textcolor{red}{b'_{i,j}},u'_{i,j}))$, and abort if $\bot$.
    \end{itemize}
    \item Apply a verification procedure to $\{b_{i,j}\}_{i: T_i = 1, j \notin S_i}$ based on the Hamiltonian for $Q(x)$. If this passes, parse the bits $\{b_{i,j}\}_{i: T_i = 1, j \in S_i}$ as a set of output samples $\{q_t\}_t$, and output $b \coloneqq \mathsf{Maj}\left(\{P(q_t)\}_t\right)$.
\end{itemize}
}

\paragraph{Binding.} Following the classical intuition, we would like to split $(x,\widetilde{c})$ into yes and no instances:
\begin{enumerate}
    \item ``Yes'' instance: $\widetilde{c}$ can only be opened in a way that would cause the verifier to output $b = P(Q(x))$ (or $\bot$). In this case, the prover could potentially learn the secret parameters $\sparam$ via repeated queries, but would not be able to break soundness.
    \item ``No'' instance: $\widetilde{c}$ can only be opened in a way that would cause the verifier to output $b \neq P(Q(x))$ (or $\bot$). In this case, by one-time soundness of the underlying protocol, the prover should never be able to make the verifier output anything other than $\bot$.
\end{enumerate}

Now, a crucial difference from the classical case is that a prover might launch a \emph{superposition} of both strategies, so we can't exactly classify each $(x,\widetilde{c})$ as either a yes or a no instance. However, in this case we will hope to rely on some notion of binding from the $\PFCII$ commitment scheme in order to guarantee that the prover cannot meaningfully ``mix'' these two strategies.

As discussed above, Pauli functional commitments satisfy a notion of binding to \emph{classical bits} rather than to quantum states, so we will need to capture these two options using classical openings. For the first option, the parallel repetition theorem of \cite{TCC:ACGH20,TCC:Bartusek21} can be used to show that if the verifier accepts, then \emph{many}, say 4/5, of their output samples $q_t$ from indices $\{S_i\}_{i : T_i = 1}$ must be such that $P(q_t) = P(Q(x))$. For the second option, it is clear that the verifier will only output $b \neq P(Q(x))$ if at least half of these output samples are such that $P(q_t) \neq P(Q(x))$. Thus, it suffices to show that the prover can't mix the following strategies.

\begin{enumerate}
    \item Open $\widetilde{c}$ on the positions $\{S_i\}_{i : T_i = 1}$ to samples $q_t$ such that \emph{a large fraction} (say 4/5) of them are ``honest'': $P(q_t) = P(Q(x))$.
    \item Open $\widetilde{c}$ on the positions $\{S_i\}_{i: T_i = 1}$ to samples $q_t$ such that \emph{a significant fraction} (say 1/2) of them are ``dishonest'': $P(q_t) \neq P(Q(x))$.
\end{enumerate}

Since the $\{S_i\}_i$ positions are all standard basis positions, and no string can satisfy both requirements, arguing that these strategies can't mix should now reduce to some binding property for the classical strings opened on the $\{S_i\}_{i : T_i = 1}$ positions. However, note that in \cref{fig:single-instance}, none of these positions are even opened by $\PFCII.\OpenZ$ (that is, opened in the standard basis)! Indeed, \emph{only} the test round positions are opened in the standard basis. 

Thus, we need to relate the strings opened on $\{S_i\}_{i : T_i = 1}$ to the strings opened on $\{S_i\}_{i : T_i = 0}$. Now, we note that $T$ is chosen via a random oracle applied to $c$, and $c$ already determines the only possible openings for the standard basis positions since the $\PFCI$ parameters are sampled in perfectly binding mode on these positions. Thus, it is possible to argue that the adversary can't significantly change their distribution of opened strings on test round vs. Hadamard round positions. So it suffices to show that the following strategies can't mix:

\begin{enumerate}
    \item Open $\widetilde{c}$ on the positions $\{S_i\}_{i : T_i = 0}$ to samples $q_t$ such that \emph{a large fraction} (say 3/4) of them are ``honest'': $P(q_t) = P(Q(x))$.
    \item Open $\widetilde{c}$ on the positions $\{S_i\}_{i: T_i = 0}$ to samples $q_t$ such that \emph{a significant fraction} (say 1/3) of them are ``dishonest'': $P(q_t) \neq P(Q(x))$.
\end{enumerate}

Thus, we will only need a ``vanilla'' notion of string binding for $\PFCII$, which can be reduced (see \cref{subsec:com-def} for more discussion) to a vanilla notion of single-bit binding for a quantum commitment to a classical bit. That is, given a decoding key $\widetilde{\dk}$, a commitment $\widetilde{c}$, and a bit $b$, let 
\[\Pi_{\widetilde{\dk},\widetilde{c},b} \coloneqq \sum_{\widetilde{u}:\DecZ(\widetilde{\dk},\widetilde{c},\widetilde{u}) = b}\ketbra{\widetilde{u}}{\widetilde{u}}\] be the projection onto strings $\widetilde{u}$ that open to $b$. Then for any two-part adversary $(\sC,\sU)$, where $\sC$ is the committer, and $\sU$ is the ``opener''\footnote{More precisely, $\sU$ is an algorithm that tries to break binding by rotating a state that is supported on valid openings to $b$ to a state that is supported on valid openings to $1-b$. We refer to this part of the adversary as the opener.} (modeled as a unitary), it holds that for any $b \in \{0,1\}$,

\[\E_{(\widetilde{\ck},\widetilde{\dk}) \gets \Gen(1^\secp)}\left[\big\|\Pi_{\widetilde{\dk},\widetilde{c},1-b} \sU \Pi_{\widetilde{\dk},\widetilde{c},b} \ket{\psi}\big\| : (\ket{\psi},\widetilde{c}) \gets \sC(\widetilde{\ck})\right] = \negl(\secp).\]

A couple of remarks:
\begin{itemize}

    \item Looking at \cref{fig:single-instance}, we see that this binding property should hold \emph{even} if the opener has oracle access to $\DecZ(\widetilde{\dk},\widetilde{c},\cdot)$. In fact, in the known construction of $\PFCI$ described above \cite{JACM:BCMVV21,SIAMCOMP:Mahadev22}, $\DecZ$ decoding can be \emph{public}. Moreover, this definition of binding is weaker than both the dual-mode and uncertainty properties, and thus our requirements for $\PFCII$ can \emph{so far} be satisfied by the known construction of $\PFCI$.
    \item Note that we only require binding on the standard basis positions, that is, $(i,j)$ such that $h_{i,j}$ = 0. Looking at \cref{fig:single-instance}, we see that the prover does \emph{not} have access to $\DecX(\widetilde{\dk}_{i,j},\widetilde{c}_{i,j},\cdot)$ on these positions. This is important, because the ability to perform a Hadamard basis measurement on the committed qubit implies the ability to reflect it across the $X$ (Hadamard basis) axis, thus changing its standard basis measurement. Thus, it seems difficult to design a Pauli functional commitment scheme that remains binding when the opener has access to $\DecX$.
\end{itemize}

\paragraph{Proving soundness for a single instance.} Next, we briefly discuss how soundness for a single instance can be proven based on this binding property of $\PFCII$. We start with an adversary that is assumed to be breaking soundness after a number of queries to the verification oracle. That is, they output a proof $\pi^*$ that causes the verifier to accept and output $b \neq P(Q(x))$).  We know that a significant fraction of the samples $q_t$ from positions $\{S_i\}_{i : T_i = 0}$ in $\pi^*$ must be such that $Q(q_t) \neq P(Q(x))$. Then, we replace each of the adversary's $\Ver[\sparam,\widetilde{\dk},(x,\widetilde{c})]$ queries one by one to being answered with $\bot$. While the adversary may query $\Ver[\sparam,\widetilde{\dk},(x,\widetilde{c})]$ on accepting $\pi$, we know that for such $\pi$, a large fraction of the samples $q_t$ from positions $\{S_i\}_{i : T_i = 0}$ must be such that $Q(q_t) = P(Q(x))$. Thus, by the binding of $\PFCII$, the fact that we are changing the oracle's response to such $\pi$ should have a negligible effect on the probability that the adversary continues to output $\pi^*$, since $\pi$ and $\pi^*$ contain openings to different strings and thus reside in parts of the adversary's state that have negligible overlap. After replacing all of these queries with $\bot$, we see that our adversary is actually breaking soundness of the underlying one-time sound protocol, since they no longer learn anything from their queries to $\Ver[\sparam,\widetilde{\dk},(x,\widetilde{c})]$, which completes the proof. For more details, see the discussion before the ``soundness'' part of the proof of \cref{thm:PV-soundness}.

\subsection{Public verifiability in the oracle model}

Next, we show how to obtain full-fledged public-verifiability in the oracle model. As a first attempt, we follow the classical approach, and include in the public parameters the $\PFCII$ parameters $\{\widetilde{\ck}_{i,j}\}_{i,j}$ along with a classical oracle that implements the following program $\mathsf{OGen}[k]$, which has a PRF key $k$ hard-coded.\\

\noindent\underline{$\mathsf{OGen}[k]$:}
\begin{itemize}
    \item Take an $x$ and a commitment $\widetilde{c}$ as input, and compute $s \coloneqq \PRF_k((x,\widetilde{c}))$.
    \item Compute $(\pp,\sparam) \coloneqq \Gen(1^\secp;s)$ from \cref{fig:one-time} using random coins $s$, and output $\pp$.
\end{itemize}

Unfortunately, this attempt does not result in a sound scheme. To see why, note that the adversary can query the verification oracle on multiple $(x,\widetilde{c})$, thus using it to implement the oracle $\PFCII.\DecX(\widetilde{\dk}_{i,j},\cdot,\cdot)$ for \emph{any} index $(i,j)$ of its choice. Indeed, for each index $(i,j)$, the adversary just has to find some $(x,\widetilde{c})$ that generates parameters with $h_{i,j} = 1$. As mentioned above, if the opener has access to $\PFCII.\DecX(\widetilde{\dk}_{i,j},\cdot,\cdot)$, it is not clear how to obtain any binding property for the bit on index $(i,j)$. Thus, an adversary could break soundness on a particular instance $(x,\widetilde{c})$ by querying its oracles on \emph{other} instances $(x',\widetilde{c}')$ in order to obtain access to any $\PFCII.\DecX(\widetilde{\dk}_{i,j},\cdot,\cdot)$ of its choice. 

\paragraph{Using signature tokens.} To solve this issue, we use signature tokens to make sure that the adversary's strategy on multiple distinct $(x,\widetilde{c})$ cannot ``mix''. That is, we include the signing key $\ket{\sk}$ for a signature token scheme in the public parameters, and alter $\mathsf{OGen}[k]$ as follows, where $\vk$ is the verification key for the signature token scheme.\\  

\noindent\underline{$\mathsf{OGen}[k,\textcolor{red}{\vk}]$:}
\begin{itemize}
    \item Take an $x$, a commitment $\widetilde{c}$, \textcolor{red}{and a signature $\sigma$} as input.
    \item \textcolor{red}{If $\sigma$ is a valid signature of $(x,\widetilde{c})$ under $\vk$,} compute $s \coloneqq \PRF_k((x,\widetilde{c},\textcolor{red}{\sigma}))$, \textcolor{red}{and otherwise abort.}
    \item Compute $(\pp,\sparam) \coloneqq \Gen(1^\secp;s)$ from \cref{fig:one-time} using random coins $s$, and output $\pp$.
\end{itemize}

Moreover, the verification oracle $\Ver[\vk,k]$, which now hard-codes $k$ rather than some fixed secret parameters $\sparam$, will also require a valid signature $\sigma$ on any $(x,\widetilde{c})$ that it takes as input. Intuitively, once the adversary learns the public parameters $\pp_{x,\widetilde{c},\sigma}$ corresponding to some instance $(x,\widetilde{c})$ and signature $\sigma$, it can \emph{only} access the oracles $\PFCII.\DecX(\dk_{i,j},\cdot,\cdot)$ on the specific indices $(i,j)$ such that $h_{i,j} = 1$  for the $h$ hard-coded in parameters $\pp_{x,\widetilde{c},\sigma}$. Note that this actually requires the signature token scheme to be \emph{strongly unforgeable}. That is, the adversary shouldn't even be able to produce a different signature $\sigma'$ on the same message $(x,\widetilde{c})$, since then $(x,\widetilde{c},\sigma')$ could be used to generate a fresh set of parameters with different $h$. While this notion was not proven explicitly in \cite{arxiv:BenSat16}, we note that it follows easily from their proof strategy.

To formalize this intuition, we treat the PRF as a random oracle $H$ and make use of the measure and re-program technique of \cite{C:DFMS19,C:DonFehMaj20}. If the adversary is breaking soundness, it must output a proof $\pi$ with respect to some $(x,\widetilde{c},\sigma)$. Thus, we can ``pre-measure'' one of the adversary's queries to $H$ to obtain $(x,\widetilde{c},\sigma)$, and then re-program $H((x,\widetilde{c},\sigma)) \to s$ to fresh randomness $s$, which defines fresh parameters $(\pp_{x,\widetilde{c},\sigma},\sparam_{x,\widetilde{c},\sigma})$. After this measurement, by the strong unforgeability of the signature token, the adversary won't be able to query the verification oracle on any $(x',\widetilde{c}',\sigma') \neq (x,\widetilde{c},\sigma)$, so they will only be able to access $\DecX(\widetilde{\dk}_{i,j},\cdot,\cdot)$ for $(i,j)$ such that $h_{i,j} = 1$ as defined by $\pp_{x,\widetilde{c},\sigma}$. Then, security should reduce to the single instance setting discussed above.

It is useful to note a crucial difference from the more direct but flawed approach to using signature tokens discussed earlier in the overview. There, we could never hope to use the security of the signature token, because we couldn't ``force'' the adversary to ever measure a signature (and indeed there was an attack on the attempted scheme). Here, since we are using the signature as part of the input to a random oracle, we can make use of measure-and-reprogram to first ``force'' a measurement of a signature during the security proof, and \emph{then} use signature token security.

\paragraph{The need for public decodability.} However, we have so far omitted a crucial detail. Note that \emph{before} the measurement of $(x,\widetilde{c},\sigma)$, the adversary \emph{can} access any $\DecX$ oracle of its choice. Indeed, we can't hope to prevent this, as the adversary has full access to both $\mathsf{OGen}[k,\vk]$ and $\Ver[k,\vk]$, and this measurement anyway only happens during an intermediate hybrid in the proof. 

In the reduction to the binding of $\PFCII$, this first part of the adversary corresponds to the \emph{commit} stage. Thus, we will need a Pauli functional commitment scheme where the \emph{committer} has access to both the $\DecZ$ and $\DecX$ oracles, while the \emph{opener} (necessarily) only has access to $\DecZ$.

We refer to such a commitment scheme as a \emph{Pauli functional commitment with public decodability}. Somewhat more formally, we will require the following binding property, where $\DecZ[\dk]$ (resp. $\DecX[\dk]$) is the oracle implementing the classical functionality $\DecZ(\dk,\cdot,\cdot)$ (resp. $\DecX(\dk,\cdot,\cdot)$). For any polynomial-query adversary $(\sC,\sU)$,

\[\Pr_{(\widetilde{\ck},\widetilde{\dk}) \gets \Gen(1^\secp)}\left[\big\|\Pi_{\widetilde{\dk},\widetilde{c},1-b} \sU^{\DecZ[\widetilde{\dk}]} \Pi_{\widetilde{\dk},\widetilde{c},b} \ket{\psi}\big\| = 1/\poly(\secp) : (\ket{\psi},\widetilde{c}) \gets \sC^{\DecZ[\widetilde{\dk}],\DecX[\widetilde{\dk}]}(\widetilde{\ck})\right] = \negl(\secp).\]

Unfortunately, the known construction of Pauli functional commitments \cite{JACM:BCMVV21,SIAMCOMP:Mahadev22} does not satisfy this property, which we explain in the following section. Thus, in the remainder of this overview, we demonstrate a novel approach to constructing Pauli functional commitments, and describe a construction with public decodability in the oracle model. Once we have this commitment, our construction of non-interactive publicly-verifiable classical verification of quantum partitioning circuits is complete, which also completes our construction of obfuscation for pseudo-deterministic quantum circuits.



\subsection{Pauli functional commitments with public decodability}

First, we review why the Pauli functional commitment based on claw-free hash functions \cite{JACM:BCMVV21,SIAMCOMP:Mahadev22} does not satisfy binding with public decodability. To commit to a state $\ket{\psi} = \alpha_0\ket{0} + \alpha_1\ket{1}$, the committer evaluates and measures an (approximately) two-to-one hash function $f$ in superposition to end up with a commitment $c$ and a left-over state $\alpha_0\ket{0}\ket{x_0} + \alpha_1\ket{1}\ket{x_1}$, where $x_0,x_1$ are $n$-bit strings such that $x_0$ starts with 0 and $x_1$ starts with 1. If they do this honestly, it will hold that $f(x_0) = f(x_1) = c$. Moreover, the receiver has a trapdoor for $f$ and can thus compute both $x_0$ and $x_1$ from $c$.

Now, a standard basis opening to the bit $b$ is the string $x_b$. To open $\ket{\psi}$ in the Hadamard basis, the committer measures each qubit of their left-over state in the Hadamard basis, obtaining a bit $b'$ and a string $d$. It follows that $b \coloneqq b' + d \cdot (x_0 + x_1)$\footnote{Here, and throughout this section, all arithmetic will be over $\bbF_2$.} is a decoding of the Hadamard basis measurement of $\ket{\psi}$. 
Thus, if we define $S \coloneqq \{0,x_0 + x_1\}$ to be a one-dimensional subspace of $\bbF_2^n$, access to the $\DecX$ oracle provides the committer with a membership oracle for the subspace $S^\bot$. Since $S$ is just one dimension, it is straightforward to use this oracle to learn a description of $S$, which is $x_0 + x_1$. But if the committer $\sC$ computes the string $x_0+x_1$ and passes it along with $\alpha_0\ket{0}\ket{x_0} + \alpha_1\ket{1}\ket{x_1}$ to $\sU$, the opener can first measure their state in the standard basis to obtain $(b,x_b)$, and then use $x_0 + x_1$ to compute $(1-b,x_{1-b})$, obtaining a valid opening for \emph{both} bits in the standard basis. This completely breaks any notion of binding for the commitment scheme.

\paragraph{Using a larger subspace.} To solve this issue, we follow this template but increase the dimension of $S$, thus decreasing the dimension of $S^\bot$. That is, suppose that the left-over state after a commitment to $\ket{\psi} = \alpha_0\ket{0} + \alpha_1\ket{1}$ was instead \[\alpha_0\ket{0}\ket{A_0} + \alpha_1\ket{1}\ket{A_1},\] where $A = S+v$ is a coset of a random $n/2$-dimensional subspace $S$,\footnote{Assume that $A$ and $S$ are ``balanced'', meaning that exactly half of their vectors start with 0.} $A_0$ is the affine subspace of vectors in $A$ that start with 0, and $A_1$ is the affine subspace of vectors in $A$ that start with 1. Here, we are using the notation \[\ket{A} \coloneqq \frac{1}{\sqrt{|A|}}\sum_{s \in A}\ket{s}\] for any affine subspace $A$.

It can be shown that if this state is measured in the Hadamard basis to produce $b',d$, then $b \coloneqq b' \oplus r_{d,S}$ is a decoding of the Hadamard basis measurement of $\ket{\psi}$, where we define the bit $r_{d,S} = 0$ if $d \in S^\bot$ and $r_{d,S} = 1$ if $d + (1,0,\dots,0) \in S^\bot$. Thus, the $\DecX$ oracle can be implemented just given a membership checking oracle for $S^\bot$. Moreover, now that $S^\bot$ has $n/2$ dimensions, and $S$ is random, it is no longer clear that an adversary can use oracle access to $S^\bot$ to learn a description of $S$.

\paragraph{Completing the construction.} Now, two main questions remain: (1) How do we define a commitment key $\ck$ that enables the committer to apply the map $\ket{b} \to \ket{b}\ket{A_b}$? (2) What is the actual \emph{commitment string} $c$? We will first address question (1).

Our commitment key will consist of a quantum state and a classical oracle. The $\Gen$ algorithm will sample a random $n/2$-dimensional affine subspace $A = S+v$, set $\dk = A$, and release the quantum state $\ket{A}$, which is a uniform superposition over all vectors in $A$. Note that $\ket{A} = \frac{1}{\sqrt{2}}\ket{A_0} + \frac{1}{\sqrt{2}}\ket{A_1}$, which can be seen as the ``$\ket{+}$'' state in the two-dimensional space spanned by $\ket{A_0}$ and $\ket{A_1}$. Thus, for any $b \in \{0,1\}$, we need to allow the committer to rotate the $\ket{+}$ state to the ``$\ket{b}$'' state $\ket{A_b}$. It is easy to project onto vectors that start with either 0 or 1, but we will have to implement a reflection across the $X$-axis of this space if this projection results in $\ket{A_{1-b}}$. While it is clear that this can be done given a quantum oracle implementing the projection $\ketbra{A}{A}$, it was observed by \cite{STOC:AGKZ20} that a \emph{classical} oracle for membership in $S^\bot$ suffices! Thus, as a first attempt, we will set the commitment key $\ck$ to consist of $\ket{A}$ and an oracle $O[S^\bot]$ for membership in $S^\bot$.

This brings us to our second question. So far, we have shown that a committer, given $\ck$, can perform the map

\[\alpha_0\ket{0} + \alpha_1\ket{1} \to \alpha_0\ket{0}\ket{A_0} + \alpha_1\ket{1}\ket{A_1},\]

and give this final state to the opener. However, since the opener also has access to $\ck$ and thus to $O[S^\bot]$, there is no sense in which the original state is committed, since the opener could continue to use $O[S^\bot]$ to rotate arbitrarily around the space spanned by $\ket{A_0}$ and $\ket{A_1}$.

To fix this, we use a signature token. We include the signing key $\ket{\sk}$ for a single-bit signature token scheme in $\ck$, and alter the oracle $O[S^\bot]$ so that it only responds given a valid signature on 0. The actual commitment string $c$ will then be a signature on 1. Thus, while the \emph{committer} is free to rotate around $\mathsf{span}\{\ket{A_0},\ket{A_1}\}$ using access to $S^\bot$, as soon as it outputs a valid classical commitment string $c$, the membership oracle for $S^\bot$ will become inaccessible and the \emph{opener} will intuitively be unable to make further changes to the state.

\paragraph{The proof of binding.} Now, it remains to formalize this intuition, and prove that this scheme satisfies binding with public decodability. After appealing to the security of the signature token scheme, we can reduce this to showing that for any polynomial-query adversary $(\sC,\sU)$,

\[\Pr\left[\big\|\Pi_{A_1} \sU^{O[A]} \Pi_{A_0}\ket{\psi} \big\| \geq 1/\poly(\secp) : \ket{\psi} \gets \sC^{O[A],O[S^\bot]}(\ket{A})\right] = \negl,\] 

where the probability is over a random choice of $n/2$-dimensional affine subspace $A = S+v$, and $\Pi_{A_b}$ is the projection onto vectors $s \in A_b$. Note that $\sC$ and $\sU$ have access to $O[A]$, the membership checking oracle for the affine subspace $A$ since this is needed to implement $\DecZ$, and $\sC$ has access to $O[S^\bot]$ because it is needed to implement both $\ck$ and $\DecX$.

To show this, we will follow \cite{10.1145/2213977.2213983}'s blueprint for proving security in the classical oracle model, and proceed via the following steps.

\begin{enumerate}
    \item Show that we can instead sample $A$ from a public ambient space of dimension $3n/4$, and remove $\sU$'s access to the $O[A]$ oracle.
    \item Perform a worst-case to average-case reduction over the sampling of $A$. 
    \item Have the committer apply amplitude amplification onto $\Pi_{A_0}$. At this point, we can reduce the problem to showing that for small enough $\epsilon$, there cannot exist a query-bounded $\sC$ and a unitary $\sU$ such that for \emph{all} $n/2$-dimensional affine subspaces $A$ of $\bbF_2^{3n/4}$, \[\ket{\psi_A} \in \mathsf{Im}(\Pi_{A_0}) ~~ \text{and} ~~ \big\| \Pi_{A_1}\sU \ket{\psi_A}\big\| \geq \epsilon,\] where $\ket{\psi_A} \gets \sC^{O[A],O[S^\bot]}(\ket{A})$.
    \item Apply the ``inner-product adversary method'' of \cite{10.1145/2213977.2213983}. That is, we (i) define a relation $\cR$ on pairs of affine subspaces $(A,B)$ such that $\braket{A | B} = 1/2$ for all $(A,B) \in \cR$, (ii) argue that for any collection of states $\{\ket{\psi_A}\}_A$ that satisfy the above conditions, \[\E_{(A,B) \gets \cR}[|\braket{\psi_A | \psi_B}|] \leq 1/2 - \delta\] for some large enough $\delta$, and (iii) conclude that if $\sC$ can decrease the expected inner product over $\cR$ by $\delta$, it must be making ``too many'' oracle queries, yielding a contradiction. 
\end{enumerate}

However, arguing part (ii) of this final step turns out to be significantly more challenging than analogous claims in previous work (e.g. \cite{10.1145/2213977.2213983,arxiv:BenSat16,STOC:AGKZ20}). Indeed, the condition is neither that $\ket{\psi_A}$ is some \emph{fixed} state (as in  \cite{10.1145/2213977.2213983}), or that measuring $\ket{\psi_A}$ in the standard basis yields a classical string in some well-defined set (as in \cite{arxiv:BenSat16,STOC:AGKZ20}). Rather, the condition involves reasoning about the overlap between two projectors, where one is defined via an \emph{arbitrary} rotation $\sU$. Moreover, we only have the guarantee that $\ket{\psi_A}$ is $\epsilon$-close to $\mathsf{Im}(\sU^\dagger\Pi_{A_1}\sU)$, and this value cannot be amplified to 1 (depending on $\sU$, the images of $\Pi_{A_0}$ and $\sU^\dagger\Pi_{A_1}\sU$ may not intersect at all). 

In \cref{appendix:inner-product}, we show that for our definition of $\cR$, $\delta > \epsilon^{13}$, which is enough for us to reach a contradiction and complete the proof. We proceed by contradiction, and eventually reduce to a Welch bound \cite{1055219}, which upper bounds the number of vectors of a given minimum distance that can be packed into a low-dimensional Hilbert space. We defer a further overview and details of this proof to \cref{appendix:inner-product}. This completes our proof of binding with public decodability.

\section{Preliminaries}

Let $\secp$ denote the security parameter. We write $\negl(\cdot)$ to denote any \emph{negligible} function, which is a function $f$ such that for every constant $c \in \mathbb{N}$ there exists $N \in \mathbb{N}$ such that for all $n > N$, $f(n) < n^{-c}$. We write $\nonnegl(\cdot)$ to denote any function $f$ that is not negligible. That is, there exists a constant $c$ such that for infinitely many $n$, $f(n) \geq n^{-c}$. Finally, we write $\poly(\cdot)$ to denote any polynomial function $f$. That is, there exists a constant $c$ such that for all $n \in \bbN$, $f(n) \leq n^{-c}$. For two probability distributions $D_0,D_1$ with classical support $S$, let \[\mathsf{TV}\left(D_0,D_1\right) \coloneqq \sum_{x \in S}|D_0(x) - D_1(x)|\] denote the total variation distance. For a set $S$, we let $x \gets S$ denote sampling a uniformly random element $x$ from $S$. For a classical randomized algorithm $y \gets C(x)$, we let $y \coloneqq C(x;r)$ denote running $C$ with random coins $r$.

\subsection{Quantum information}

An $n$-qubit register $\regX$ is a named Hilbert space $\bbC^{2^n}$. A pure quantum state on register $\regX$ is a unit vector $\ket{\psi}^{\regX} \in \bbC^{2^n}$. A mixed state on register $\regX$ is described by a density matrix $\rho^{\regX} \in \bbC^{2^n \times 2^n}$, which is a positive semi-definite Hermitian operator with trace 1. 

A \emph{quantum operation} $F$ is a completely-positive trace-preserving (CPTP) map from a register $\regX$ to a register $\regY$, which in general may have different dimensions. That is, on input a density matrix $\rho^{\regX}$, the operation $F$ produces $F(\rho^{\regX}) = \tau^{\regY}$ a mixed state on register $\regY$. A \emph{unitary} $U: \regX \to \regX$ is a special case of a quantum operation that satisfies $U^\dagger U = U U^\dagger = \bbI^{\regX}$, where $\bbI^{\regX}$ is the identity matrix on register $\regX$. A \emph{projector} $\Pi$ is a Hermitian operator such that $\Pi^2 = \Pi$, and a \emph{projective measurement} is a collection of projectors $\{\Pi_i\}_i$ such that $\sum_i \Pi_i = \bbI$. Throughout this work, we will often write an expression like $\Pi\ket{\psi}$, where $\ket{\psi}$ has been defined on some multiple registers, say $\regX$, $\regY$, and $\regZ$, and $\Pi$ has only been defined on a subset of these registers, say $\regY$. In this case, we technically mean $(\bbI^\regX \otimes \Pi \otimes \bbI^\regZ)\ket{\psi}$, but we drop the identity matrices for notational convenience. 

A family of quantum circuits is in general a sequence of quantum operations $\{C_\secp\}_{\secp \in \bbN}$, parameterized by the security parameter. We say that the family is \emph{quantum polynomial time} (QPT) if $C_\secp$ can be implemented with a $\poly(\secp)$-size circuit. A family of \emph{oracle-aided} quantum circuits $\{C^F_\secp\}_{\secp \in \bbN}$ have access to an oracle $F: \{0,1\}^* \to \{0,1\}^*$ that implements some classical map. That is, $C$ can apply a unitary that maps $\ket{x}\ket{y} \to \ket{x}\ket{y \oplus F(x)}$. Finally, we will sometimes also consider families of \emph{unitaries} $\{U_\secp\}_{\secp \in \bbN}$ and families of \emph{oracle-aided unitaries} $\{U^F_\secp\}_{\secp \in \bbN}$, where each operation between oracle queries is a unitary.

Let $\Tr$ denote the trace operator. For registers $\regX,\regY$, the \emph{partial trace} $\Tr^{\regY}$ is the unique operation from $\regX,\regY$ to $\regX$ such that for all $(\rho,\tau)^{\regX,\regY}$, $\Tr^{\regY}(\rho,\tau) = \Tr(\tau)\rho$. The \emph{trace distance} between states $\rho,\tau$, denoted $\TD(\rho,\tau)$ is defined as \[\TD(\rho,\tau) \coloneqq \frac{1}{2}\Tr\left(\sqrt{(\rho-\tau)^\dagger(\rho-\tau)}\right).\] The trace distance between two states $\rho$ and $\tau$ is an upper bound on the probability that any (unbounded) algorithm can distinguish $\rho$ and $\tau$. 

\begin{lemma}[Gentle measurement \cite{DBLP:journals/tit/Winter99}]\label{lemma:gentle-measurement}
Let $\rho$ be a quantum state and let $(\Pi,\bbI-\Pi)$ be a projective measurement such that $\Tr(\Pi\rho) \geq 1-\delta$. Let \[\rho' = \frac{\Pi\rho\Pi}{\Tr(\Pi\rho)}\] be the state after applying $(\Pi,\bbI-\Pi)$ to $\rho$ and post-selecting on obtaining the first outcome. Then, $\TD(\rho,\rho') \leq 2\sqrt{\delta}$.
\end{lemma}

We will also often make use of the following simple claim.

\begin{claim}\label{claim:max-prob}
Consider a register $\regR$ on $n$ qubits and a distribution $\cF$ over classical functions $f : \{0,1\}^n \to \{0,1\}$. For any such $f$, let $\Pi_f$ be the projection onto $x$ such that $f(x) = 1$. Then for any $\ket{\psi}$ on register $\regR$, \[\E_{f \gets \cF}\left[\big\|\Pi_f \ket{\psi} \big\|^2\right] \leq \max_x \left\{\Pr_{f \gets \cF}[f(x) = 1]\right\}.\]
\end{claim}

\begin{proof}
For any $\ket{\psi} \coloneqq \sum_x \alpha_x \ket{x}$, write 
\[\E_{f \gets \cF}\left[\big\| \Pi_f \ket{\psi}\big\|^2\right] = \E_{f \gets \cF}\left[\sum_{x : f(x) = 1}|\alpha_x|^2\right] = \sum_x \Pr_{f \gets \cF}[f(x)=1] \cdot |\alpha_x|^2 \leq \max_x \left\{\Pr_{f \gets \cF}[f(x) = 1]\right\},\] where the last inequality holds because $\{|\alpha_x|^2\}_x$ is a probability distribution.
\end{proof}

Finally, we define the notion of a pseudo-deterministic quantum ciruit.

\begin{definition}[Pseudo-deterministic quantum circuit]
\label{def: deterministic}
A family of psuedo-deterministic quantum circuits $\{Q_\secp\}_{\secp \in \mathbb{N}}$ is defined as follows. The circuit $Q_\secp$ takes as input a classical string $x \in \{0,1\}^{n(\secp)}$ and outputs a bit $b \gets Q_\secp(x)$. The circuit is pseudo-deterministic if for every sequence of classical inputs $\{x_\secp\}_{\secp \in \mathbb{N}}$, there exists a sequence of outputs $\{b_\secp\}_{\secp \in \mathbb{N}}$ such that \[\Pr[Q_\secp(x_\secp) \to b_\secp] = 1-\negl(\secp).\] We will often leave the dependence on $\secp$ implicit, and just refer to pseudo-deterministic circuits $Q$ with input $x$. In a slight abuse of notation, we will denote by $Q(x)$ the bit $b$ such that $\Pr[Q(x) \to b] = 1-\negl(\secp)$.
\end{definition}

\subsection{Obfuscation}

\begin{definition}[Virtual black-box obfuscation]\label{def:ideal-obfuscation}
  A virtual black-box (VBB) obfuscator for a family of pseudo-deterministic quantum (resp. classical) circuits is a pair of QPT algorithms $(\Obf,\Eval)$ with the following syntax.
  \begin{itemize}
      \item $\Obf(1^\secp, Q) \to \widetilde{Q}$: $\Obf$ takes as input the security parameter $1^\secp$ and the description of a quantum (resp. classical) circuit $Q$, and outputs a (potentially quantum) obfuscated circuit $\widetilde{Q}$.
      \item $\Eval(\widetilde{Q},x) \to b$: $\Eval$ takes as input an obfuscated circuit $\widetilde{Q}$ and an input $x$, and outputs a bit $b \in \{0,1\}$.
  \end{itemize}
  A VBB obfuscator should satisfy the following properties for any pseudo-deterministic (resp. classical) family of circuits $Q = \{Q_\secp\}_{\secp \in \bbN}$ with input length $n = n(\secp)$.
  \begin{itemize}
      \item \textbf{Correctness}: It holds with probability $1-\negl(\secp)$ over $\widetilde{Q} \gets \Obf(1^\secp,Q)$ that for all $x \in \{0,1\}^n$, $\Pr[\Eval(\widetilde{Q},x) \to Q(x)] = 1-\negl(\secp)$.
      \item \textbf{Security}: For any QPT adversary $\{\sA_\secp\}_{\secp \in \bbN}$, there exists a QPT simulator $\{\sS_\secp\}_{\secp \in \bbN}$ such that \[\bigg|\Pr\left[1 \gets \sA_\secp\left(\Obf(1^\secp,Q)\right)\right] - \Pr\left[1 \gets \sS_\secp^{O[Q]}\right]\bigg| = \negl(\secp),\] where $O[Q]$ is the oracle that computes the map $x \to Q(x)$.
  \end{itemize}
\end{definition}

\begin{definition}[Indistinguishability obfuscation]\label{def:iO}
An indistinguishability obfuscator (iO) for a family of pseudo-deterministic (resp. classical) circuits is a pair of QPT algorithms $(\Obf,\Eval)$ that has the same syntax and correctness properties as a VBB obfuscator and satisfies the following security property. For any QPT adversary $\{\sA_\secp\}_{\secp \in \bbN}$ and pair of functionally equivalent families of pseudo-deterministic (resp. classical) circuits $Q_0 = \{Q_{0,\secp}\}_{\secp \in \bbN}, Q_1 = \{Q_{1,\secp}\}_{\secp \in \bbN}$,

\[\left| \Pr\left[1 \gets \sA_\secp\left(\Obf(1^\secp,Q_0)\right)\right] - \Pr\left[1 \gets \sA_\secp\left(\Obf(1^\secp,Q_1)\right)\right]\right| = \negl(\secp).\]
\end{definition}

\subsection{Dual-mode randomized trapdoor claw-free hash functions}


\begin{definition}\label{def:clawfree}
  Let $\{X_\secp\}_{\secp \in \bbN}$ and $\{Y_\secp\}_{\secp \in \bbN}$ be families of finite sets. Below, we will leave the dependence of these sets on $\secp$ implicit. A \emph{dual-mode randomized trapdoor claw-free hash function} is described by a tuple of algorithms $(\Gen,\Eval,\Invert,\Ch,\IsValid)$ with the following syntax.
  \begin{itemize}
      \item $\Gen(1^\lambda,h) \to (\pk,\sk)$ is a randomized classical algorithm that takes as input a security parameter $1^\secp$ and a bit $h\in\{0,1\}$ (where $h=0$ indicates \emph{injective mode} and $h=1$ indicates \emph{2-to-1 mode}), and outputs a public key $\pk$ and a secret key $\sk$. The public key $\pk$ implicitly defines a function $f_\pk: \{0,1\} \times X \to \cD_Y$, where $\cD_Y$ is the set of probability distributions over $Y$.
 
      \item $\Eval(\pk,b) \to \ket{\psi_{\pk,b}}$ is a QPT algorithm that takes as input a public key $\pk$ and a bit $b$, and outputs a fixed pure state $\ket{\psi_{\pk,b}}^{\regX,\regY}$ on two registers $\regX$ and $\regY$, where $\regX$ is spanned by the elements of $X$ and $\regY$ is spanned by the elements of $Y$. We then define \[\Eval[\pk] \coloneqq \ketbra{0}{0}^\regB \otimes \Eval(\pk,0) + \ketbra{1}{1}^\regB \otimes \Eval(\pk,1),\] which is a map from the single qubit register $\regB$ to registers $(\regB,\regX,\regY)$.
      \item $\Invert(h,\sk,y)$ is a deterministic classical algorithm that takes as input $h\in\{0,1\}$, a secret key $\sk$, and an element $y \in Y$. If $h=0$, it outputs a pair $(b,x)\in\{0,1\}\times X$ or $\bot$. If $h=1$, it outputs two pairs $(0,x_0)$ and $(1,x_1)$ with $x_0,x_1 \in X$, or $\bot$. 
      \item $\Ch(\pk,b,x,y) \to \{\top,\bot\}$ is a deterministic classical algorithm that takes as input a public key $\pk$, a bit $b\in\{0,1\}$, an element $x \in X$, and an element $y \in Y$, and outputs either $\top$ or $\bot$.
      \item $\IsValid(x_0, x_1, d) \to \{\top,\bot\}$ is a deterministic classical algorithm that takes as input two elements $x_0, x_1 \in X$ and a string $d$, and outputs either $\top,\bot$, characterizing membership in a set that we call 
      \[\Valid_{x_0, x_1} \coloneqq \{d : \IsValid(x_0,x_1,d) = 1\}.\]

  \end{itemize}
  We require that the following properties are satisfied.
  \begin{enumerate}
      \item \textbf{\em Correctness:} 
      \begin{enumerate}
          \item 
      For all $(\pk,\sk) \in \Gen(1^\secp,0)$: For every $b\in\{0,1\}$, every $x\in X$, and every $y\in{\sf Supp}(f_{\pk}(b,x))$, 
      \[ \Invert(0,\sk,y) = (b,x).\]

      \item For all $(\pk,\sk) \in \Gen(1^\secp,1)$: For
      every $b\in\{0,1\}$, every  $x\in X$, and every $y\in{\sf Supp}(f_{\pk}(b,x))$, 
      $$ \Invert(1,\sk,y) = ((0,x_0),(1,x_1))$$ such that $x_b=x$, $y\in{\sf Supp}(f_{\pk}(0,x_0))$, and $y\in{\sf Supp}(f_{\pk}(1,x_1))$.
      
      \item For all $(\pk,\sk) \in \Gen(1^\secp,0) \cup \Gen(1^\secp,1)$, every $b\in\{0,1\}$ and every $x\in X$, it holds that $\Ch(\pk,(b,x),y) = 1$
      if and only if $y\in{\sf Supp}(f_{\pk}(b,x))$.
      
      \item For all $(\pk,\sk) \in \Gen(1^\secp,0) \cup \Gen(1^\secp,1)$ and every $b \in \{0,1\}$, it holds that \[\TD\left(\ket{\psi_{\pk,b}}^{\regX,\regY},\frac{1}{\sqrt{|X|}}\sum_{x \in X,y \in Y}\sqrt{(f_{\pk}(b,x))(y)}\ket{x}^\regX\ket{y}^\regY\right) = \negl(\secp),\] where $\ket{\psi_{\pk,b}} \gets \Eval(\pk,b)$. 
      
      
      \item For all $(\pk,\sk)\in\Gen(1^\secp,1)$ and every pair of elements $x_0,x_1 \in X$, the density of $\Valid_{x_0, x_1}$ is $1-\negl(\secp)$.      \end{enumerate}
      
      \item \textbf{\em Key indistinguishability:} For every QPT adversary $\{\sA_\secp\}_{\secp \in \bbN}$,
      \[ \Big| \Pr\left[1 \gets \sA_\secp(\pk) : (\pk,\sk) \gets \Gen(1^\secp,0)\right] - \Pr\left[ 1 \gets \sA_\secp(\pk) : (\pk,\sk) \gets \Gen(1^\secp,1) \right] \Big| = \negl(\secp).\]
      
      \item \textbf{\em Adaptive hardcore bit:} There is an efficiently computable and efficiently invertible injection $J: X \to \{0,1\}^w$ such that for every QPT adversary $\{\sA_\secp\}_{\secp \in \bbN}$,
      \begin{align*}
      & \Bigg| \Pr\left[\begin{array}{l}\Ch(\pk,b,x,y) = 1 ~~ \wedge \\ d \in \Valid_{x_0,x_1} ~~ \wedge \\ d \cdot (J(x_0) \oplus J(x_1)) = 0  \end{array}: \begin{array}{r}(\pk,\sk)\leftarrow \Gen(1^\secp,1) \\ (y,b,x,d) \gets {\sA_\secp}(\pk) \\ ((0,x_0),(1,x_1)) \coloneqq \Invert(1,\sk,y)\end{array}\right] \\
      & -\Pr\left[\begin{array}{l}\Ch(\pk,b,x,y) = 1 ~~ \wedge \\ d \in \Valid_{x_0,x_1} ~~ \wedge\\ d \cdot (J(x_0) \oplus J(x_1)) = 1  \end{array}: \begin{array}{r}(\pk,\sk)\leftarrow \Gen(1^\secp,1) \\ (y,b,x,d) \gets {\sA_\secp}(\pk) \\ ((0,x_0),(1,x_1)) \coloneqq \Invert(1,\sk,y)\end{array}\right]\Bigg| = \negl(\secp).
      \end{align*}

  \end{enumerate}
\end{definition}

The works of \cite{JACM:BCMVV21,SIAMCOMP:Mahadev22} showed that, assuming QLWE, there exists a dual-mode randomized trapdoor claw-free hash function.

\subsection{Quantum fully-homomorphic encryption}\label{subsec:QFHE}

We define quantum fully-homomorphic encryption (QFHE) with classical keys and classical encryption of classical messages. One could also define encryption for quantum states and decryption for quantum ciphertexts, but we will not need that in this work.

\begin{definition}[Quantum fully-homomorphic encryption]\label{def:qfhe}
A quantum fully-homomorphic encryption scheme $(\Gen,\Enc,\Eval,\Dec)$ consists of the following efficient algorithms.
\begin{itemize}
    \item $\Gen(1^\lambda,D) \to (\pk,\sk)$: On input the security parameter $1^\secp$ and a circuit depth $D$, the key generation algorithm returns a public key $\pk$ and a secret key $\sk$.
    \item $\Enc(\pk, x) \to \ct$: On input the public key $\pk$ and a classical plaintext $x$, the encryption algorithm returns a classical ciphertext $\ct$.
    \item $\Eval(Q, \ct) \to \widetilde{\ct}$: On input a quantum circuit $Q$ and a ciphertext $\ct$, the quantum evaluation algorithm returns an evaluated  ciphertext $\widetilde{\ct}$.
    \item $\Dec(\sk, \ct) \to x$: On input the secret key $\sk$ and a classical ciphertext $\ct$, the decryption algorithm returns a message $x$.
\end{itemize}
\end{definition}

The scheme should satisfy the standard notion of semantic security. 

\begin{definition}[Semantic security]\label{def:semantic-security}
A QFHE scheme $(\Gen,\Enc,\Eval,\Dec)$ is secure if for any QPT adversary $\{\sA_\secp\}_{\secp \in \bbN}$ and circuit depth $D$, 

\[\bigg| \Pr\left[\sA_\secp(\ct) = 1 : \begin{array}{r}(\pk,\sk) \gets \Gen(1^\secp,D) \\ \ct \gets \Enc(\pk,0)\end{array}\right] - \Pr\left[\sA_\secp(\ct) = 1 : \begin{array}{r}(\pk,\sk) \gets \Gen(1^\secp,D) \\ \ct \gets \Enc(\pk,1)\end{array}\right]\bigg| = \negl(\secp).\]
\end{definition}

We will also require the following notion of correctness for evaluation of pseudo-deterministic quantum circuits.

\begin{definition}[Evaluation Correctness]\label{def:eval correct}
A QFHE scheme $(\Gen,\Enc,\Eval,\Dec)$ is correct if for any polynomial $D(\secp)$, family of pseudo-deterministic quantum circuits $\{Q_\secp\}_{\secp \in \bbN}$ of depth $D(\secp)$, inputs $\{x_\secp\}_{\secp \in \bbN}$, security parameter $\secp$, $(\pk,\sk) \in \Gen(1^\secp,D(\secp))$, and $\ct \in \Enc(\pk,x)$,
\[
\Pr[\Dec(\sk, \Eval(Q_\secp, \ct)) = Q_\secp(x_\secp)] = 1-\negl(\secp).
\]
\end{definition}

The works of Mahadev~\cite{FOCS:Mahadev18b} and Brakerski~\cite{C:Brakerski18} show that such a QFHE scheme can be constructed from QLWE. 

\subsection{Measure and re-program}

\begin{importedtheorem}[Measure and re-program \cite{C:DFMS19,C:DonFehMaj20}]\footnote{This theorem was stated more generally in \cite{C:DFMS19,C:DonFehMaj20} to consider the drop in expectation for each specific $a^* \in A$, and also to consider a more general class of quantum predicates. }\label{thm:measure-and-reprogram}
Let $A,B$ be finite non-empty sets, and let $q \in \bbN$. Let $\sA$ be an oracle-aided quantum circuit that makes $q$ queries to a uniformly random function $H: A \to B$ and then outputs classical strings $(a,z)$ where $a \in A$. There exists a two-stage quantum circuit $\Sim[\sA]$ such that for any predicate $V$, it holds that 
\begin{align*}\Pr\left[V(a,b,z) = 1: \begin{array}{r} (a,\state) \gets \Sim[\sA] \\ b \gets B \\ z \gets \Sim[\sA](b,\state) \end{array}\right] \geq \frac{\Pr\left[V(a,H(a),z) = 1 : (a,z) \gets \sA^H\right]}{(2q+1)^2}.\end{align*}

Moreover, $\Sim[\sA]$ operates as follows.

\begin{itemize}
    \item Sample $H: A \to B$ as a $2q$-wise independent function and $(i,d) \gets (\{0,\dots,q-1\} \times \{0,1\}) \cup \{(q,0)\}$.
    \item Run $\sA$ until it has made $i$ oracle queries, answering each query using $H$. 
    \item When $\sA$ is about to make its $(i+1)$'th oracle query, measure its query registers in the standard basis to obtain $a$. In the special case that $(i,d) = (q,0)$, the simulator measures (part of) the final output register of $\sA$ to obtain $a$.
    \item The simulator receives $b \gets B$.
    \item If $d = 0$, answer $\sA$'s $(i+1)$'th query using $H$, and if $d=1$, answer $\sA$'s $(i+1)$'th query using $H[a \to b]$, which is the function $H$ except that $H(a)$ is re-programmed to $b$.
    \item Run $\sA$ until it has made all $q$ oracle queries. For queries $i+2$ through $q$, answer using $H[a \to b]$.
    \item Measure $\sA$'s output $z$.
\end{itemize}

Note that the running time of $\Sim[\sA]$ is at most $\poly(q,\log|A|,\log|B|)$ times the running time of $\sA$.

\end{importedtheorem}

\subsection{Signature tokens}\label{subsec:sig-tokens}

A signature token scheme consists of algorithms $(\Gen,\Sign,\Verify)$ with the following syntax. 

\begin{itemize}
    \item $\Gen(1^\secp) \to (\vk,\ket{\sk})$: The $\Gen$ algorithm takes as input the security parameter $1^\secp$ and outputs a classical verification key $\vk$ and a quantum signing key $\ket{\sk}$.
    \item $\Sign(b,\ket{\sk}) \to \sigma$: The $\Sign$ algorithm takes as input a bit $b \in \{0,1\}$ and the signing key $\ket{\sk}$, and outputs a signature $\sigma$.
    \item $\Verify(\vk,b,\sigma) \to \{\top,\bot\}$: The $\Verify$ algorithm takes as input a verification key $\vk$, a bit $b$, and a signature $\sigma$, and outputs $\top$ or $\bot$.
\end{itemize}

A signature token should satisfy the following definition of correctness.

\begin{definition}\label{def:token-correctness}
A signature token scheme $(\Gen,\Sign,\Verify)$ is \emph{correct} if for any $b \in \{0,1\}$,
\[\Pr\left[\Verify(\vk,b,\sigma) = \top : \begin{array}{r}(\vk,\ket{\sk}) \gets \Gen(1^\secp) \\ \sigma \gets \Sign(b,\ket{\sk})\end{array}\right] = 1-\negl(\secp).\]
\end{definition}

Next, we define notions of unforgeability. In this paper, it suffices to consider security in the \emph{oracle model}, where the adversarial signer has oracle access to the verification function, rather than to the description of the verification key $\vk$ itself.

\begin{definition}\label{def:unforgeability}
A signature token scheme $(\Gen,\Sign,\Verify)$ satisfies \emph{unforgeability} if for any oracle-aided adversary $\{\sA_\secp\}_{\secp \in \bbN}$ that makes at most $\poly(\secp)$ oracle queries,

\[\Pr\left[\begin{array}{l} \Verify(\vk,0,\sigma_0) = \top ~~ \wedge \\ \Verify(\vk,1,\sigma_1) = \top\end{array} : \begin{array}{r} (\vk,\ket{\sk}) \gets \Gen(1^\secp) \\ (\sigma_0,\sigma_1) \gets \sA_\secp^{\Verify[\vk]}(\ket{\sk}) \end{array}\right] = \negl(\secp),\] where $\Verify[\vk]$ is the functionality $\Verify(\vk,\cdot,\cdot)$.
\end{definition}

\begin{importedtheorem}[\cite{arxiv:BenSat16}]
There exists a signature token scheme in the oracle model that satisfies unforgeability.
\end{importedtheorem}


We will also require a signature token with the property of \emph{strong unforgeability}, defined as follows. 

\begin{definition}\label{def:strong-unforgeability}
A signature token scheme $(\KeyGen,\Sign,\Verify)$ satisfies \emph{strong unforgeability} if for any oracle-aided adversary $\{\sA_\secp\}_{\secp \in \bbN}$ that makes at most $\poly(\secp)$ oracle queries,

\[\Pr\left[\begin{array}{l} (b_0,\sigma_0) \neq (b_1,\sigma_1) ~~ \wedge \\ \Verify(\vk,b_0,\sigma_0) = \top ~~ \wedge \\ \Verify(\vk,b_1,\sigma_1) = \top\end{array} : \begin{array}{r} (\vk,\ket{\sk}) \gets \Gen(1^\secp) \\ (b_0,\sigma_0,b_1,\sigma_1) \gets \sA_\secp^{\Verify[\vk]}(\ket{\sk}) \end{array}\right] = \negl(\secp),\] where $\Verify[\vk]$ is the functionality $\Verify(\vk,\cdot,\cdot)$.
\end{definition}

\begin{claim}
There exists a signature token scheme in the oracle model that satisfies strong unforgeability.
\end{claim}

\begin{proof}
This follows by a slight tweak to arguments in \cite{arxiv:BenSat16}. We first note that by a union bound, it suffices to show that each of the following three cases happens with negligible probability: (1) $\sA_\secp$ outputs $\sigma_0,\sigma_1$ such that $\sigma_0$ is a valid signature of 0 and $\sigma_1$ is a valid signature of 1, (2) $\sA_\secp$ outputs $\sigma_0 \neq \sigma_0'$ that are both valid signatures of 0, and (3) $\sA_\secp$ outputs $\sigma_1 \neq \sigma_1'$ that are both valid signatures of 1. The first case is already proven by \cite{arxiv:BenSat16}. 

The second case can be shown by following the proofs in \cite{arxiv:BenSat16} except for one difference: for a subspace $A < \bbF_2^n$, the ``target set'' $\Lambda(A)$ (defined on page 25 of \cite{arxiv:BenSat16}) is instead defined to consist of pairs of vectors $(a,b)$ such that $a \neq b \in A \setminus \{0^n\}$. The only change in the proof then comes in \cite[Lemma 19]{arxiv:BenSat16}, where we need to show that 
\[\max_{A \in S(n), (a,b) \in \Lambda(A)} \Pr_{B \gets \cR_A}[(a,b) \in \Lambda(B)] \leq \frac{1}{4},\] where $S(n)$ is the set of subspaces of $\bbF_2^n$ of dimension $n/2$, and for any $A \in S(n)$, $\cR_A$ is the set of $B \in S(n)$ such that $\mathsf{dim}(A \cap B) = n/2 -1$. This follows by first noting that any distinct non-zero $a,b \in A$ specify a two-dimensional subspace $\{0,a,b,a+b\}$. Then, following the proof of \cite[Lemma 19]{arxiv:BenSat16}, and defining \[G(m,k) \coloneqq \prod_{i = 0}^{k-1}\frac{2^{m-i}-1}{2^{k-i}-1}\] to be the number of subspaces of $\bbF_2^k$ of dimension $m$, we have that this expression is at most 

\[\frac{G(n/2-2,n/2-3)}{G(n/2,n/2-1)} = \frac{2^{n/2 - 1}-1}{2^{n/2}-1} \cdot \frac{2^{n/2 - 2}-1}{2^{n/2-1}-1} \leq \frac{1}{4}.\]


Finally, the third case can be proven in the same way as the second, by defining $\Lambda(A)$ as the set of $(a,b)$ such that $a \neq b \in A^\bot \setminus \{0^n\}$. 
\end{proof}

\begin{remark}
It is straightforward to extend any single-bit signature token scheme (which is described above) to a multi-bit scheme for polynomial-size messages, by signing each bit with a different invocation of the single-bit scheme.
\end{remark}
\section{Pauli Functional Commitments}\label{sec:SSC}

\subsection{Definition}\label{subsec:com-def}

A Pauli functional commitment resembles a standard bit commitment scheme with a classical receiver. However, when used to commit to a qubit $\ket{\psi} = \alpha_0\ket{0} + \alpha_1\ket{1}$ in superposition, it supports the ability to open to either a \emph{standard} or \emph{Hadamard} basis measurement of $\ket{\psi}$. A Pauli functional commitment should also satisfy some notion of binding to a classical bit.

The syntax of a Pauli functional commitment is given below. We present the syntax in the \emph{oracle model}, where the committer obtains access to an efficient classical oracle $\cCK$ as part of its commitment key. Such a scheme can be heuristically instantiated in the plain model by using a post-quantum indistinguishability obfuscator to obfuscate this oracle. We also specify that the remainder of the commitment key is a quantum state $\ket{\ck}$, but note that this is not inherent to the definition of a Pauli functional commitment.
 

\begin{definition}[Pauli functional commitment: Syntax]\label{def:PFC}
A Pauli functional commitment consists of six algorithms $(\Gen,\allowbreak\Com,\allowbreak\OpenZ,\allowbreak\OpenX,\allowbreak\DecZ,\allowbreak\DecX)$ with the following syntax.

\begin{itemize}
    \item $\Gen(1^\secp) \to (\dk,\ket{\ck},\cCK)$ is a QPT algorithm that takes as input the security parameter $1^\secp$ and outputs a classical decoding key $\dk$ and a quantum commitment key $(\ket{\ck},\cCK)$, where $\ket{\ck}$ is a quantum state on register $\regK$, and $\cCK$ is the description of a classical deterministic polynomial-time functionality $\cCK: \{0,1\}^* \to \{0,1\}^*$.
    \item $\Com_b^{\cCK}(\ket{\ck}) \to (\regU,c)$ is a QPT algorithm that is parameterized by a bit $b$ and has oracle access to $\cCK$. It applies a map from register $\regK$ (initially holding the commitment key $\ket{\ck}$) to registers $(\regU,\regC)$ and then measures $\regC$ in the standard basis to obtain a classical string $c \in \{0,1\}^*$ and a left-over state on register $\regU$. We then write \[\Com^{\cCK} \coloneqq \ketbra{0}{0} \otimes \Com_0^{\cCK} + \ketbra{1}{1} \otimes \Com_1^{\cCK}\] to refer to the map that applies the  $\Com_b^{\cCK}$ map classically controlled on a single-qubit register $\regB$ to produce a state on registers $(\regB,\regU,\regC)$, and then measures $\regC$ in the standard basis to obtain a classical string $c$ along with a left-over quantum state on registers $(\regB,\regU)$.
    \item $\OpenZ(\regB,\regU) \to u$ is a QPT measurement on registers $(\regB,\regU)$ that outputs a classical string $u$.
    \item $\OpenX(\regB,\regU) \to u$ is a QPT measurement on registers $(\regB,\regU)$ that outputs a classical string $u$.
    \item $\DecZ(\dk,c,u) \to \{0,1,\bot\}$ is a classical deterministic polynomial-time algorithm that takes as input the decoding key $\dk$, a string $c$, and a string $u$, and outputs either a bit $b$ or a $\bot$ symbol.
    \item $\DecX(\dk,c,u) \to \{0,1,\bot\}$ is a classical deterministic polynomial-time algorithm that takes as input a the decoding key $\dk$, a string $c$, and a string $u$, and outputs either a bit $b$ or a $\bot$ symbol.
\end{itemize}

\end{definition} 



\begin{definition}[Pauli functional Commitment: Correctness]\label{def:poc-correctness}
A Pauli functional commitment $(\Gen,\allowbreak\Com,\allowbreak\OpenZ,\allowbreak\OpenX,\allowbreak\DecZ,\allowbreak\DecX)$ is \emph{correct} if for any single-qubit (potentially mixed) state on register $\regB$, it holds that 
\[\mathsf{TV}\left(\mathsf{Z}(\regB), \mathsf{PFCZ}(1^\secp,\regB)\right) = \negl(\secp), ~~ \text{and} ~~ \mathsf{TV}\left(\mathsf{X}(\regB), \mathsf{PFCX}(1^\secp,\regB)\right) = \negl(\secp),\]

where the distributions are defined as follows.
\begin{itemize}
    \item $\mathsf{Z}(\regB)$ measures $\regB$ in the standard basis.
    \item $\mathsf{X}(\regB)$ measures $\regB$ in the Hadamard basis.
    \item $\mathsf{PFCZ}(1^\secp,\regB)$ samples $(\dk,\ket{\ck},\cCK) \gets \Gen(1^\secp), (\regB,\regU,c) \gets \Com^{\cCK}(\regB,\ket{\ck}), u \gets \OpenZ(\regB,\regU)$, and outputs $\DecZ(\dk,c,u)$.
    \item $\mathsf{PFCX}(1^\secp,\regB)$ samples $(\dk,\ket{\ck},\cCK) \gets \Gen(1^\secp), (\regB,\regU,c) \gets \Com^{\cCK}(\regB,\ket{\ck}), u \gets \OpenX(\regB,\regU)$, and outputs $\DecX(\dk,c,u)$.
\end{itemize}

\end{definition}

A Pauli functional commitment that satisfies \emph{binding with public decodability} allows the adversarial Committer to have oracle access to the receiver's decoding functionalities $\DecZ(\dk,\cdot,\cdot)$ and $\DecX(\dk,\cdot,\cdot)$. However, we crucially do not give the adversarial \emph{Opener} access to $\DecX(\dk,\cdot,\cdot)$.

\begin{definition}[Pauli functional commitment: Single-bit binding with public decodability]\label{def:poc-single-bit-binding}
A Pauli functional commitment $(\Gen,\Com,\allowbreak\OpenZ,\allowbreak\OpenX,\allowbreak\DecZ,\allowbreak\DecX)$ satisfies \emph{single-bit binding with public decodability} if the following holds. Given $\dk,c,$ and $b \in \{0,1\}$, let 

\[\Pi_{\dk,c,b} \coloneqq \sum_{u : \DecZ(\dk,c,u) = b}\ketbra{u}{u}.\]

Consider any adversary $\{(\sC_\secp, \sU_\secp)\}_{\secp \in \bbN}$, where each $\sC_\secp$ is an oracle-aided quantum operation, each $\sU_\secp$ is an oracle-aided unitary, and each $(\sC_\secp,\sU_\secp)$ make at most $\poly(\secp)$ oracle queries. Then for any $b \in \{0,1\}$,

\[\E\left[\bigg\|\Pi_{\dk,c,1-b}\sU_\secp^{\cCK,\DecZ[\dk]}\Pi_{\dk,c,b}\ket{\psi}\bigg\| : (\ket{\psi},c) \gets \sC_\secp^{\cCK,\DecZ[\dk],\DecX[\dk]}(\ket{\ck})\right] = \negl(\secp),\] where the expectation is over $\dk,\ket{\ck},\cCK \gets \Gen(1^\secp)$. Here, $\DecZ[\dk]$ is the oracle implementing the classical functionality $\DecZ(\dk,\cdot,\cdot)$ and $\DecX[\dk]$ is the oracle implementing the classical functionality $\DecX(\dk,\cdot,\cdot)$.

\end{definition}



Next, we extend the above single-bit binding property to a notion of \emph{string binding}.


\begin{definition}[Pauli functional commitment: String binding with public decodability]\label{def:predicate-binding} A Pauli functional commitment $(\Gen,\allowbreak\Com,\allowbreak\OpenZ,\allowbreak\OpenX,\allowbreak\DecZ,\allowbreak\DecX)$ satisfies \emph{string binding with public decodability} if the following holds for any polynomial $m = m(\secp)$ and two disjoint sets $W_0,W_1 \subset \{0,1\}^m$ of $m$-bit strings. Given a set of $m$ verification keys $\bdk = (\dk_1,\dots,\dk_m)$, $m$ strings $\bc = (c_1,\dots,c_m)$, and $b \in \{0,1\}$, define \[\Pi_{\bdk,\bc,W_b} \coloneqq \sum_{w \in W_b} \left(\bigotimes_{i \in [m]}\Pi_{\dk_i,c_i,w_i} \right).\] 
Consider any adversary $\{(\sC_\secp, \sU_\secp)\}_{\secp \in \bbN}$, where each $\sC_\secp$ is an oracle-aided quantum operation, each $\sU_\secp$ is an oracle-aided unitary, and each $(\sC_\secp,\sU_\secp)$ make at most $\poly(\secp)$ oracle queries. Then,

\begin{align*}\E\left[\bigg\|\Pi_{\bdk,\bc,W_1}\sU_\secp^{\bCK,\DecZ[\bdk]}\Pi_{\bdk,\bc,W_0}\ket{\psi}\bigg\|: (\ket{\psi},\bc) \gets \sC_\secp^{\bCK,\DecZ[\bdk],\DecX[\bdk]}(\ket{\mathbf{\bck}})\right] = \negl(\secp), \end{align*}


where the expectation is over $\{\dk_i,\ket{\ck_i},\cCK_i \gets \Gen(1^\secp)\}_{i \in [m]}$. Here, $\ket{\bck} = (\ket{\ck_1},\dots,\ket{\ck_m})$, $\bCK$ is the collection of oracles $\cCK_1,\dots,\cCK_m$, $\DecZ[\bdk]$ is the collection of oracles $\DecZ[\dk_1],\dots,\DecZ[\dk_m]$, and $\DecX[\bdk]$ is the collection of oracles $\DecX[\dk_1],\dots,\DecX[\dk_m]$. 

\end{definition}




We prove the following lemma in \cref{sec:appendix-string}.

\begin{lemma}
Any Pauli functional commitment that satisfies \emph{single-bit binding with public decodability} also satisfies \emph{string binding with public decodability}.
\end{lemma}

\subsection{Construction}\label{subsec:SSC-functionality}

Before describing our construction, we introduce some notation.


\begin{itemize}
    \item A subspace $S < \bbF_2^n$ is \emph{balanced} if half of its vectors start with 0 and the other half start with 1. Note that $S$ is balanced if and only if at least one of its basis vectors starts with 1. Thus, a random large enough (say $n/2$-dimensional) subspace is balanced with probability $1-\negl(n)$. By default, we will only consider balanced subspaces in what follows.
    \item For an affine subspace $A = S+v$ of $\bbF_2^n$, we write 
    \[\ket{S+v} \coloneqq \frac{1}{\sqrt{|S|}}\sum_{s \in S}\ket{s+v}.\]
    \item Given an affine subspace $S+v$, let $(S+v)_0$ be the set of vectors in $S+v$ that start with 0 and let $(S+v)_1$ be the set of vectors in $S+v$ that start with 1.

\end{itemize}

We describe our construction of a Pauli functional commitment in \cref{fig:PFC-construction}. 

\protocol{Pauli Functional Commitment}{A Pauli functional commitment that satisfies \emph{binding with public decodability}.}{fig:PFC-construction}{

Parameters: Polynomial $n = n(\secp) \geq \secp$.\\
Ingredients: Signature token scheme $(\Tok.\Gen,\Tok.\Sign,\Tok.\Verify)$ (\cref{subsec:sig-tokens}).\\

\begin{itemize}
    \item $\Gen(1^\secp)$: Sample a uniformly random $n/2$-dimensional balanced affine subspace $S + v$ of $\bbF_2^n$ and sample $(\vk,\ket{\sk}) \gets \Tok.\Gen(1^\secp)$. Set \[\dk \coloneqq (S,v,\vk), \ \ \ \ \ket{\ck} \coloneqq (\ket{S+v},\ket{\sk}).\] Define $\cCK$ to take as input $(\sigma,s)$ for $s \in \{0,1\}^n$ and output $\bot$ if $\Tok.\Verify(\vk,0,\sigma) = \bot$, and otherwise output 0 if $s \in S^\bot$ or 1 if $s \notin S^\bot$.
    \item $\Com_b^{\cCK}(\ket{\ck})$: 
    \begin{itemize}
        \item Parse $\ket{\ck} = (\ket{S+v}^{\regK_0},\ket{\sk}^{\regK_1})$.
        \item Coherently apply $\Tok.\Sign(1^\secp,0,\cdot)$ from the $\regK_1$ register to a fresh register $\regG$, which will now hold a superposition over signatures $\sigma$ on the bit 0.
        \item Measure the first qubit of register $\regK_0$ in the standard basis. If the result is $b$, the state on register $\regK_0$ has collapsed to $\ket{(S+v)_b}$, and we continue. Otherwise, perform a rotation from $\ket{(S+v)_{1-b}}$ to $\ket{(S+v)_b}$ by applying the operation $(H^{\otimes n})^{\regK_0}\mathsf{Ph}^{\cCK(\cdot,\cdot)}(H^{\otimes n})^{\regK_0}$ to registers $(\regK_0,\regG)$, where $\mathsf{Ph}^{\cCK(\cdot,\cdot)}$ is the map $\ket{s}^{\regK_0}\ket{\sigma}^\regG \to (-1)^{\cCK(\sigma,s)}\ket{s}^{\regK_0}\ket{\sigma}^\regG$.  
        \item Next, reverse the $\Tok.\Sign(1^\secp,0,\cdot)$ operation on $(\regK_1,\regG)$ to recover $\ket{\sk}$ on register $\regK_1$. 
        \item Finally, sample and output $c \gets \Tok.\Sign(1^\secp,1,\ket{\sk})$, along with the final state on register $\regU \coloneqq \regK_0$.
    \end{itemize}
    \item $\OpenZ(\regB,\regU)$: Measure all registers in the standard basis.
    \item $\OpenX(\regB,\regU)$: Measure all registers in the Hadamard basis.
    \item $\DecZ(\dk,c,u)$: 
    \begin{itemize}
        \item Parse $\dk = (S,v,\vk)$ and $u = (b,s)$, where $b \in \{0,1\}$ and $s \in \{0,1\}^n$.
        \item Check that $\Tok.\Verify(\vk,1,c) = \top$, and if not output $\bot$.
        \item If $s \in (S + v)_b$, output $b$, and otherwise output $\bot$. 
    \end{itemize}
    \item $\DecX(\dk,c,u)$: 
    \begin{itemize}
        \item Parse $\dk = (S,v,\vk)$ and $u = (b',s)$, where $b' \in \{0,1\}$ and $s \in \{0,1\}^n$.
        \item Check that $\Tok.\Verify(\vk,1,c) = \top$, and if not output $\bot$.
        \item If $s \in S^\bot$, then define $r \coloneqq 0$. If $s \oplus (1,0,\dots,0) \in S^\bot$, then define $r \coloneqq 1$. Otherwise, abort and output $\bot$. That is, $r$ is set to 0 if $s \in S^\bot$ and to 1 if $s \in (S_0)^\bot \setminus S^\bot$. Then, output $b \coloneqq b' \oplus r$. 
    \end{itemize}

\end{itemize}
}





\begin{theorem}
The Pauli functional commitment described in \cref{fig:PFC-construction} satisfies \emph{correctness} (\cref{def:poc-correctness}).
\end{theorem}

\begin{proof}
We will show correctness assuming that the signature token scheme $\Tok$ is perfectly correct. In reality, it may be statistically correct, but in this case we can still conclude that \cref{fig:PFC-construction} satisfies correctness, which allows for a negligible statistical distance.

We will first show that the map applied by $\Com_b^{\cCK}$ in the case that the measurement of the first qubit of $\regK_0$ is $1-b$ successfully takes $\ket{(S+v)_{1-b}} \to \ket{(S+v)_b}$. Since we are assuming perfect correctness from $\Tok$, it suffices to show that for any balanced affine subspace $\ket{S+v}$, \[H^{\otimes n}\mathsf{Ph}^{O[S^\bot]}H^{\otimes n}\ket{(S+v)_{1-b}} \to \ket{(S+v)_b},\] where $\mathsf{Ph}^{O[S^\bot]}$ is the map $\ket{s} \to (-1)^{O[S^\bot](s)}\ket{s}$, and $O[S^\bot]$ is the oracle that outputs 0 if $s \in S^\bot$ and 1 if $s \notin S^\bot$. This was actually shown in \cite{STOC:AGKZ20}, but we repeat it here for completeness. 

We will use the facts that $S_1 = S_0 + w$ for some $w$, and that $(S+v)_0 = S_0 + v_0$ and $(S+v)_1 = S_0 + v_1$ for some $v_0,v_1$ such that $v_0 + v_1 = w$. Also note that for any $s \in S^\bot$, $s \cdot w = 0$, and for any $s \in ({S_0})^\bot \setminus S^\bot$, $s \cdot w = 1$.

\begin{align*}
    H^{\otimes n}&\mathsf{Ph}^{O[S^\bot]}H^{\otimes n}\ket{(S+v)_{1-b}} \\ 
    &= H^{\otimes n} \mathsf{Ph}^{O[S^\bot]}H^{\otimes n}\frac{1}{\sqrt{2^{n/2 - 1}}}\left(\sum_{s \in S_0} \ket{s + v_{1-b}}\right) \\ 
    &= H^{\otimes n}\mathsf{Ph}^{O[S^\bot]}\frac{1}{\sqrt{2^{n/2+1}}}\left(\sum_{s \in S_0^\bot}(-1)^{s \cdot v_{1-b}}\ket{s}\right) \\ 
    &= H^{\otimes n}\mathsf{Ph}^{O[S^\bot]}\frac{1}{\sqrt{2^{n/2+1}}}\left(\sum_{s \in S^\bot} (-1)^{s \cdot w + s \cdot v_{b}}\ket{s} + \sum_{s \in S_0^\bot \setminus S^\bot} (-1)^{s \cdot w + s \cdot v_{b}} \ket{s}\right) \\ 
    &= H^{\otimes n}\mathsf{Ph}^{O[S^\bot]}\frac{1}{\sqrt{2^{n/2+1}}}\left(\sum_{s \in S^\bot} (-1)^{s \cdot v_{b}}\ket{s} + \sum_{s \in S_0^\bot \setminus S^\bot} (-1)^{1 + s \cdot v_{b}} \ket{s}\right) \\ 
    &= H^{\otimes n} \frac{1}{\sqrt{2^{n/2+1}}}\left(\sum_{s \in S^\bot} (-1)^{s \cdot v_b}\ket{s} + \sum_{s \in S_0^\bot \setminus S^\bot} (-1)^{s \cdot v_b} \ket{s}\right) \\ 
    &= H^{\otimes n}\frac{1}{\sqrt{2^{n/2+1}}}\left(\sum_{s \in S_0^\bot}(-1)^{s \cdot v_b}\ket{s}\right) \\ 
    &= \ket{(S+v)_b}.
\end{align*}

Thus, applying $\Com^{\cCK}$ to a pure state $\ket{\psi} = \alpha_0\ket{0} + \alpha_1\ket{1}$ and commitment key $\ket{\ck}$ produces (up to negligible trace distance) the state \[\ket{\psi_\Com} = \alpha_0\ket{0}\ket{(S+v)_0} + \alpha_1\ket{1}\ket{(S+v)_1},\] and a signature $c$ on the bit 1.

We continue by arguing that measuring and decoding $\ket{\psi_\Com}$ in the standard (resp. Hadamard) basis produces the same distribution as directly measuring $\ket{\psi}$ in the standard (resp. Hadamard) basis. As a mixed state is a probability distribution over pure states, this will complete the proof of correctness.

First, it is immediate that measuring $\ket{\psi_\Com}$ in the standard basis produces a bit $b$ with probability $|\alpha_b|^2$ along with a vector $s$ such that $s \in (S+v)_b$. 

Next, note that applying Hadamard to each qubit of $\ket{\psi_\Com}$ except the first results in the state 

\[\alpha_0 \ket{0} \left(\sum_{s \in S_0^\bot}(-1)^{s \cdot v_0}\ket{s}\right) + \alpha_1 \ket{1} \left(\sum_{s \in S_0^\bot}(-1)^{s \cdot v_1}\ket{s}\right),\] and thus, measuring each of these qubits (except the first) in the Hadamard basis produces a vector $s$ and a single-qubit state \[(-1)^{s \cdot v_0}\alpha_0\ket{0} + (-1)^{s \cdot v_1} \alpha_1\ket{1} = \alpha_0 \ket{0} + (-1)^{s \cdot w}\alpha_1 \ket{1}.\] So, measuring this qubit in the Hadamard basis is equivalent to measuring $\ket{\psi}$ in the Hadamard basis and masking the result with $s \cdot w$. Recalling that $s \cdot w = 0$ if $s \in S^\bot$ and $s \cdot w = 1$ if $s \in (S_0)^\bot \setminus S^\bot$ completes the proof of correctness.
\end{proof}

\subsection{Binding}\label{subsec:binding-proof}


This section is dedicated to proving the following theorem.

\begin{theorem}\label{thm:binding}
Assuming that $\Tok$ satisfies unforgeability (\cref{def:unforgeability}), the Pauli functional commitment described in \cref{fig:PFC-construction} with $n \geq 130\secp$ satisfies \emph{single-bit binding with public decodability} (\cref{def:poc-single-bit-binding}).
\end{theorem}

The proof of this theorem will be identical for each choice of $b \in \{0,1\}$ in the statement of \cref{def:poc-single-bit-binding}. So, consider any adversary $(\sC,\sU)$ attacking the publicly-decodable single-bit binding game for $b=0$, where we drop the indexing by $\secp$ for notational convenience. We first show that it suffices to prove the following claim, in which $\sU$ no longer has oracle access to $\cCK$. 


\begin{claim}\label{claim:no-com-oracle}
For any $(\sC,\sU)$ where $\sC$ and $\sU$ each make $\poly(\secp)$ many oracle queries, it holds that
\[\Pr_{\dk,\ket{\ck},\cCK \gets \Gen(1^\secp)}\left[\bigg\|\Pi_{\dk,c,1}\sU^{\DecZ[\dk]}\Pi_{\dk,c,0}\ket{\psi}\bigg\|^2 \geq \frac{1}{2^\secp} : (\ket{\psi},c) \gets \sC^{\cCK,\DecZ[\dk],\DecX[\dk]}(\ket{\ck})\right] = \negl(\secp).\]
\end{claim}

\begin{lemma}
\cref{claim:no-com-oracle} implies \cref{thm:binding}.
\end{lemma}

\begin{proof}

First, we note that to prove \cref{thm:binding}, it suffices to show that for any any $(\sC,\sU)$ with $\poly(\secp)$ many oracle queries and any $\epsilon(\secp) = 1/\poly(\secp)$, it holds that

\[\Pr_{\dk,\ket{\ck},\cCK \gets \Gen(1^\secp)}\left[\bigg\|\Pi_{\dk,c,1}\sU^{\cCK,\DecZ[\dk]}\Pi_{\dk,c,0}\ket{\psi}\bigg\|^2 \geq \epsilon(\secp) : (\ket{\psi},c) \gets \sC^{\cCK,\DecZ[\dk],\DecX[\dk]}(\ket{\ck})\right] = \negl(\secp).\]

To show that \cref{claim:no-com-oracle} implies the above statement, we define the oracle $O_\bot$ to always map $(\sigma,s) \to \bot$, and then argue that 

\[\E_{\substack{\dk,\ket{\ck},\cCK \gets \Gen(1^\secp) \\ (\ket{\psi},c) \gets \sC^{\cCK,\DecZ[\dk],\DecX[\dk]}(\ket{\ck})}}\left[\bigg\| \Pi_{\dk,c,1}\sU^{\cCK,\DecZ[\dk]}\Pi_{\dk,c,0}\ket{\psi} \bigg\|^2 - \bigg\| \Pi_{\dk,c,1}\sU^{O_\bot,\DecZ[\dk]}\Pi_{\dk,c,0}\ket{\psi} \bigg\|^2 \right] = \negl(\secp).\] This follows from a standard hybrid argument, by reduction to the unforgeability of the signature token scheme. That is, consider replacing each $\cCK$ oracle query with a $O_\bot$ oracle query one by one, starting with the last query. That is, we define hybrid $\cH_0$ to be \[\E_{\substack{\dk,\ket{\ck},\cCK \gets \Gen(1^\secp) \\ (\ket{\psi},c) \gets \sC^{\cCK,\DecZ[\dk],\DecX[\dk]}(\ket{\ck})}}\left[\bigg\| \Pi_{\dk,c,1}\sU^{\cCK,\DecZ[\dk]}\Pi_{\dk,c,0}\ket{\psi} \bigg\|^2 \right],\] and in hybrid $\cH_i$, we switch the $i$'th from the last query from being answered by $\cCK$ to being answered by $O_\bot$. Now, fix any $i$, and consider measuring the query register of $\sU$'s $i$'th from last query to obtain classical strings $(\sigma,s)$. Then since $\Pi_{\dk,c,0}$ is the zero projector when $c$ is not a valid signature on 1, and $\cCK$ outputs $\bot$ whenever $\sigma$ is not a valid signature on 0, we have that

\[\E[\cH_{i-1} - \cH_i] \leq \Pr[\Tok(\vk,1,c) = 1 \wedge \Tok(\vk,0,\sigma) = 1] = \negl(\secp),\] by the unforgeability of the signature token scheme. Since there are $\poly(\secp)$ many hybrids, this completes the hybrid argument.

Finally, it follows by Markov that

\begin{align*}\Pr_{\substack{\dk,\ket{\ck},\cCK \gets \Gen(1^\secp) \\ (\ket{\psi},c) \gets \sC^{\cCK,\DecZ[\dk],\DecX[\dk]}(\ket{\ck})}}&\left[\bigg\| \Pi_{\dk,c,1}\sU^{\cCK,\DecZ[\dk]}\Pi_{\dk,c,0}\ket{\psi} \bigg\|^2 - \bigg\| \Pi_{\dk,c,1}\sU^{O_\bot,\DecZ[\dk]}\Pi_{\dk,c,0}\ket{\psi} \bigg\|^2 \geq \epsilon(\secp)-\frac{1}{2^\secp}\right]\\ &\leq \frac{\negl(\secp)}{\epsilon(\secp)-1/2^\secp} = \negl(\secp),
\end{align*} which completes the proof.

\end{proof}

Now, we introduce some more notation.

\begin{itemize}
    \item Let $\cA_{k,n}$ be the set of balanced $k$-dimensional affine subspaces of $\bbF_2^n$.
    \item For an affine subspace $A = S+v$, let $O[A] : \bbF_2^n \to \{0,1\}$ be the classical functionality that outputs 1 on input $s$ iff $s \in S+v$, and let $O[A^\bot]: \bbF_2^n \to \{0,1\}$ be the classical functionality that outputs 1 on input $s$ iff $s \in S^\bot$.
    \item For an affine subspace $A = S+v$ and a bit $b \in \{0,1\}$, define the projector \[\Pi[A_b] \coloneqq \sum_{s \in (S+v)_b}\ketbra{s}{s}.\]
\end{itemize}

We will use this notation to re-define the game in \cref{claim:no-com-oracle}, and show that it suffices to prove the following claim.


\begin{claim}\label{claim:A-oracle}
For any two unitaries $(\sU_\Com,\sU_\Open)$, where $\sU_\Com$ and $\sU_\Open$ each make $\poly(\secp)$ many oracle queries, it holds that

\[\Pr_{A \gets \cA_{n/2,n}}\left[\bigg\|\Pi[A_1]\sU_\Open^{O[A]}\Pi[A_0]\ket{\psi}\bigg\|^2 \geq \frac{1}{2^\secp} : \ket{\psi} \coloneqq \sU_\Com^{O[A],O[A^\bot]}(\ket{A})\right] = \negl(\secp).\]
\end{claim}

\begin{lemma}
\cref{claim:A-oracle} implies \cref{claim:no-com-oracle}.
\end{lemma}

\begin{proof}
First, we note that re-defining $\Pi_{\dk,c,b}$ in the statement of \cref{claim:no-com-oracle} to ignore $c$ and only check for membership in the affine subspace $(S+v)_b$ only potentially increases the squared norm of the resulting vector. This means that we can ignore the string $c$ output by $\sC$. Then, we can give the committer $\vk$ in the clear, and observe that it is now straightforward for the committer to simulate its $\DecZ[\dk]$ oracle with $O[A]$, where $A$ is the affine subspace defined by $\dk$, and also to simulate its $\DecX[\dk]$ oracle with $O[A^\bot]$. Finally, we can purify any operation $\sC$ to consider a unitary $\sU_\Com$ that outputs $\ket{\psi}$.
\end{proof}

Our next step is to remove $\sU_\Open$'s oracle access to $O[A]$. We will show that it suffices to prove the following. 




\begin{claim}\label{claim:remove-oracle}
For any two unitaries $(\sU_\Com,\sU_\Open)$, where $\sU_\Com$ makes $\poly(\secp)$ many oracle queries, it holds that

\[\Pr_{A \gets \cA_{n/2,3n/4}}\left[\bigg\|\Pi[A_1]\sU_\Open\Pi[A_0]\ket{\psi}\bigg\|^2 \geq \frac{1}{2^{\secp+1}} : \ket{\psi} \coloneqq \sU_\Com^{O[A],O[A^\bot]}(\ket{A})\right] = \negl(\secp).\]
\end{claim}

Notice that we are now sampling affine subspaces of a $3n/4$-dimensional space.

\begin{lemma}
\cref{claim:remove-oracle} implies \cref{claim:A-oracle}.
\end{lemma}

\begin{proof}

Given an $n/2$-dimensional affine subspace $A$, let $T \gets \mathsf{Super}(3n/4,A)$ denote sampling a uniformly random $(3n/4)$-dimensional subspace $T$ such that $A \subset T$. Then, define $O[T \setminus \{0^n\}]$ to be the oracle that checks for membership in the set $T \setminus \{0^n\}$.

Now, we will show via a standard hybrid argument that 

\[\E_{\substack{A \gets \cA_{n/2,} \\ T \gets \mathsf{Super}(3n/4,A) \\ \ket{\psi} \coloneqq \sU_\Com^{O[A],O[A^\bot]}(\ket{A})}}\left[\bigg\| \Pi[A_1]\sU_\Open^{O[A]}\Pi[A_0]\ket{\psi}\bigg\|^2 - \bigg\| \Pi[A_1]\sU_\Open^{O[T \setminus \{0^n\}]}\Pi[A_0]\ket{\psi}\bigg\|^2\right] \leq \frac{\poly(\secp)}{2^{n/4}}.\]

Consider replacing each $O[A]$ oracle query with a $O[T \setminus \{0^n\}]$ oracle query one by one, starting with the last query. That is, we define hybrid $\cH_0$ to be \[\E_{\substack{A \gets \cA_{n/2,n} \\ T \gets \mathsf{Super}(3n/4,A) \\ \ket{\psi} \coloneqq \sU_\Com^{O[A],O[A^\bot]}(\ket{A})}}\left[\bigg\| \Pi[A_1]\sU_\Open^{O[A]}\Pi[A_0]\ket{\psi}\bigg\|^2\right],\] and in hybrid $\cH_i$, we switch the $i$'th from the last query from being answered by $O[A]$ to being answered by $O[T \setminus \{0^n\}]$. By \cref{claim:max-prob}, we have that \[\E[\cH_{i-1} - \cH_i] \leq \max_{s}\Pr_{T}[s \in (T \setminus \{0^n\}) \setminus S ] \leq \frac{1}{2^{n/4}}.\]

Since there are $\poly(\secp)$ many hybrids, this completes the hybrid argument. Now, it follows by Markov that 

\begin{align*}\Pr_{\substack{A \gets \cA_{n/2,n} \\ T \gets \mathsf{Super}(3n/4,A) \\ \ket{\psi} \coloneqq \sU_\Com^{O[A],O[A^\bot]}(\ket{A})}}&\left[\bigg\| \Pi[A_1]\sU_\Open^{O[A]}\Pi[A_0]\ket{\psi}\bigg\|^2 - \bigg\| \Pi[A_1]\sU_\Open^{O[T \setminus \{0^n\}]}\Pi[A_0]\ket{\psi}\bigg\|^2 \geq \frac{1}{2^{\secp}}-\frac{1}{2^{\secp+1}}\right]\\ &\leq \frac{\poly(\secp)2^{\secp+1}}{2^{n/4}} = \negl(\secp),\end{align*} since $n > 5\secp$. This completes the proof, since we can imagine fixing $T$ as a public ambient space of dimension $3n/4$ and sampling $A$ as a random affine subspace of $T$.

\end{proof}

Next, we perform a worst-case to average-case reduction over the sampling of $A$ and thus show that it suffices to prove the following.

\begin{claim}\label{claim:worst-case}
There do not exist two unitaries $(\sU_\Com,\sU_\Open)$, where $\sU_\Com$ makes $\poly(\secp)$ many oracle queries, such that for all $A \in \cA_{n/2,3n/4}$ it holds that \[\bigg\|\Pi[A_1]\sU_\Open\Pi[A_0]\ket{\psi_A}\bigg\|^2 \geq \frac{1}{2^{2\secp}},\] where $\ket{\psi_A} \coloneqq \sU_\Com^{O[A],O[A^\bot]}(\ket{A})$.
\end{claim}

\begin{lemma}
\cref{claim:worst-case} implies \cref{claim:remove-oracle}.
\end{lemma}

\begin{proof}

Suppose that there exists $(\sU_\Com,\sU_\Open)$ that violates \cref{claim:remove-oracle}. We define an adversary $(\widetilde{\sC},\widetilde{\sU}_\Open)$ as follows.

\begin{itemize}
    \item $\widetilde{\sC}$ takes $\ket{A}$ as input and samples a uniformly random change of basis $B$ of $\bbF_2^{3n/4}$. Define the unitary $\sU_B$ acting on $3n/4$ qubits to map $\ket{s} \to \ket{B(s)}$.
    \item Run $\sU_\Com$ on $\ket{B(A)}$. Answer each of $\sU_\Com$'s oracle queries with $\sU_B O[A] \sU_B^\dagger$ or $\sU_B O[A^\bot] \sU_B^\dagger$, where $\sU_B$ acts on the query register.
    \item Let $\ket{\psi}$ be $\sU_\Com$'s output, and output $\ket{\widetilde{\psi}} \coloneqq (\sU_B^\dagger\ket{\psi},B)$, where register $\regB$ holds $B$, which is a classical description of the change of basis.
    \item $\widetilde{\sU}_\Open$ is defined to be $\sU_{\mathsf{CoB}^{-1}} \sU_\Open \sU_{\mathsf{CoB}}$, where \[\sU_{\mathsf{CoB}} \coloneqq \frac{1}{\#B}\sum_{\bB}\sU_B \otimes \ket{B}\bra{B}^\regB, \ \ \ \text{and} \ \ \ \sU_{\mathsf{CoB}^{-1}} \coloneqq \frac{1}{\#B}\sum_{B}\sU^\dagger_B \otimes \ket{B}\bra{B}^\regB,\] where $\#B$ is the total number of change of bases $B$.
\end{itemize}

Then it holds that for any $A \in \cA_{n/2,3n/4}$,

\begin{align*}
&\Pr\left[\bigg\|\Pi[A_1]\widetilde{\sU}_\Open\Pi[A_0]\ket{\widetilde{\psi}}\bigg\|^2 \geq \frac{1}{2^{\secp+1}} : \ket{\widetilde{\psi}} \gets \widetilde{\sC}^{O[A],O[A^\bot]}(\ket{A})\right] \\
&= \Pr_{B(A) \gets \cA_{n/2,3n/4}}\left[\bigg\|\Pi[{B(A)}_1]\sU_\Open\Pi[{B(A)}_0]\ket{\psi}\bigg\|^2 \geq \frac{1}{2^{\secp+1}} : \ket{\psi} \gets \sU_\Com^{O[B(A)],O[{B(A)}^\bot]}(\ket{B(A)})\right]\\ &= \nonnegl(\secp),\end{align*}

where the final equality follows because we are assuming that $(\sU_\Com,\sU_\Open)$ violates \cref{claim:remove-oracle}, and for any fixed balanced $A$ and uniformly random $B$, it holds that $B(A)$ is a uniformly random balanced affine subspace except with $\negl(n)$ probability. Now, define $\ket{\widetilde{\psi}_B}$ to be the output of $\widetilde{\sC}$ conditioned on sampling $B$. Then define $\widetilde{\sU}_\Com$ to be a purification of $\sC$. It holds that for any fixed $A \in \cA_{n/2,3n/4}$ and $\ket{\widetilde{\psi}} \coloneqq \widetilde{\sU}_\Com^{O[A],O[A^\bot]}(\ket{A})$,

\begin{align*}
    \bigg\| \Pi[A_1]\widetilde{\sU}_\Open \Pi[A_0]\ket{\widetilde{\psi}}\bigg\|^2 = \frac{1}{\#B}\sum_{B}\bigg\| \Pi[A_1]\widetilde{\sU}_\Open\Pi[A_0]\ket{\widetilde{\psi}_B}\bigg\|^2 \geq \nonnegl(\secp) \cdot \frac{1}{2^{\secp+1}} \geq \frac{1}{2^{2\secp}},
\end{align*}

which completes the proof.
\end{proof}

Next, we perform amplitude amplification onto $\Pi[A_0]$, showing that it suffices to prove the following claim.

\begin{claim}\label{claim:amp}
There do not exist two unitaries $(\sU_\Com, \sU_\Open)$, where $\sU_\Com$ makes at most $2^{2\secp}$ oracle queries, such that for all $A \in \cA_{n/2,3n/4}$ and $\ket{\psi_A} \coloneqq \sU_\Com^{O[A],O[A^\bot]}(\ket{A})$, there exists a state $\ket{\psi_A'}$ such that

\[\big\| \ket{\psi_A} - \ket{\psi_A'} \big\| \leq \frac{1}{2^{15\secp}}, ~~ \ket{\psi'_A} \in \mathsf{Im}(\Pi[A_0]), ~~ \text{and} ~~ \big\| \Pi[A_1]U_\Open\ket{\psi_A'}\big\| \geq \frac{1}{2^{\secp}}.\]

\end{claim}

\begin{lemma}
\cref{claim:amp} implies \cref{claim:worst-case}.
\end{lemma}

\begin{proof}

For any binary projective measurement $(\Pi, \bbI - \Pi)$, we define $\sU_\Pi$ to be a unitary that maps $\ket{\phi} \to - \ket{\phi}$ for any $\ket{\phi} \in \mathsf{Im}(\Pi)$ and acts as the identity on all $\ket{\phi}$ orthogonal to $\Pi$. We use the following imported theorem.

\begin{importedtheorem}[Fixed-point amplitude amplification, \cite{STOC:GSLW19} Theorem 27]\label{thm:quantum-search}
There exists an oracle-aided unitary $\mathsf{Amplify}$ that is parameterized by $(\alpha,\beta)$, and has the following properties. Let $\ket{\psi}$ and $\ket{\psi_G}$ be normalized states and $\Pi$ be a projector such that $\Pi\ket{\psi} = \gamma\ket{\psi_G}$, where $\gamma \geq \alpha$. Then $\ket{\widetilde{\psi}_G} \coloneqq \mathsf{Amplify}_{\alpha,\beta}^{\sU_{\ket{\psi}\bra{\psi}},\sU_\Pi}(\ket{\psi})$ is such that $\|\ket{\psi_G} - \ket{\widetilde{\psi}_G}\| \leq \beta$, and $\mathsf{Amplify}_{\alpha,\beta}^{\sU_{\ket{\psi}\bra{\psi}},\sU_\Pi}(\ket{\psi})$ makes $O(\log(1/\beta)/\alpha)$ oracle queries.
\end{importedtheorem}

Now, suppose that $(\sU_\Com,\sU_\Open)$ violates \cref{claim:worst-case}. Set $\alpha = 1/2^\secp$, $\beta = 1/2^{15\secp}$, and define \[\widetilde{\sU}_\Com(\ket{A}) \coloneqq \Amplify_{\alpha,\beta}^{\sU_{\ket{\psi_A}\bra{\psi_A}},\sU_{\Pi[A_0]}}(\ket{\psi_A}),\] where $\ket{\psi_A} \coloneqq \sU_\Com(\ket{A})$.

We first argue that $\widetilde{\sU}_\Com$ can be implemented with just oracle access to $O[A]$ and $O[A^\bot]$. Clearly, the projector $\Pi[A_0]$ can be implemented with $O[A]$, so it remains to show how to implement the projector $\ketbra{\psi_A}{\psi_A}$. Note that \[\ketbra{\psi_A}{\psi_A} = \sU_\Com \ketbra{A}{A} \sU_\Com^\dagger,\] so it suffices to show how to implement $\ketbra{A}{A}$.

Recalling that $A = S+v$, we claim that \[\ketbra{A}{A} = H^{\otimes n}\Pi[S^\bot]H^{\otimes n}\Pi[S+v].\] The proof is essentially shown in \cite[Lemma 21]{10.1145/2213977.2213983} (in the case where $A$ is a subspace), and we repeat it here for completeness. It is clear that $H^{\otimes n}\Pi[S^\bot]H^{\otimes n}\Pi[S+v]\ket{A} = \ket{A}$, so it remains to show that for any $\ket{\psi}$ such that $\braket{\psi | A} = 0$, $H^{\otimes n}\Pi[S^\bot]H^{\otimes n}\Pi[S+v]\ket{\psi} = 0$. Write $\ket{\psi} = \sum_{s \in \{0,1\}^n}c_s \ket{s}$, where $\sum_{s \in S+v}c_s = 0$. Then 

\begin{align*}
    H^{\otimes n}\Pi[S^\bot]H^{\otimes n}\Pi[S+v]\ket{\psi} &= H^{\otimes n}\Pi[S^\bot]H^{\otimes n}\sum_{s \in S+v} c_s \ket{s} \\
    &= \frac{1}{2^{n/2}}H^{\otimes n}\Pi[S^\bot]\sum_{t \in \{0,1\}^n}\sum_{s \in S+v} (-1)^{s \cdot t}c_s\ket{t} \\
    &= \frac{1}{2^{n/2}}H^{\otimes n}\sum_{t \in S^\bot}\sum_{s \in S+v} (-1)^{s \cdot t}c_s\ket{t} \\
    &= \frac{1}{2^{n/2}}H^{\otimes n}\sum_{t \in S^\bot}\left(\sum_{s \in S+v} c_s\right)\ket{t}  = 0.
\end{align*}

Thus, $\widetilde{\sU}_\Com$ can be implemented with just oracle access to $O[A]$ and $O[A^\bot]$. Moreover, it makes at most $O(\log(1/\beta)\alpha) \cdot \poly(\secp) \leq O(\secp 2^{\secp}) \cdot \poly(\secp) \leq 2^{2\secp}$ queries to $O[A]$ and $O[A^\bot]$. 

Now, define \[\ket{\psi_A'} \coloneqq \frac{\Pi[A_0]\ket{\psi_A}}{\| \Pi[A_0]\ket{\psi_A}\|},\] so $\ket{\psi_A'} \in \mathsf{Im}(\Pi[A_0])$ by definition. By the fact that $(\sU_\Com,\sU_\Open)$ violates \cref{claim:worst-case}, we know that \[\big\| \Pi[A_1]\sU_\Open\ket{\psi_A'}\big\|^2 \geq \big\| \Pi[A_1]\sU_\Open\Pi[A_0]\ket{\psi_A}\big\|^2 \geq \frac{1}{2^{2\secp}} \implies \big\| \Pi[A_1]\sU_\Open\ket{\psi_A'}\big\| \geq \frac{1}{2^\secp}.\] Finally, by the definition of $\ket{\psi_A'}$,

\[\|\Pi[A_1]\sU_\Open\Pi[A_0]\ket{\psi_A}\|^2 \geq \frac{1}{2^{2\secp}} \implies \Pi[A_0]\ket{\psi_A} = \gamma\ket{\psi_A'} ~ \text{for} ~ \gamma \geq \frac{1}{2^{\secp}},\]
so the guarantee of \cref{thm:quantum-search} implies that \[\big\| \widetilde{\sU}_\Com(\ket{A}) - \ket{\psi_A'}\big\| \leq \frac{1}{2^{15p}}.\] Thus, $(\widetilde{\sU}_\Com,\sU_\Open)$ violates \cref{claim:amp}, which completes the proof.

\end{proof}

Finally, we prove \cref{claim:amp}, which, as we have shown, suffices to prove \cref{thm:binding}.

\begin{proof}(of \cref{claim:amp}) We will use the following imported theorem.

\begin{importedtheorem}[\cite{10.1145/2213977.2213983}]\label{thm:adversary-method}
Let $\cO$ be a set of classical functionalities $F : \{0,1\}^* \to \{0,1\}$. Let $\cR$ be a symmetric binary relation between functionalities where for every $F \in \cO$, $(F,F) \notin \cR$, and for every $F \in \cO$, there exists $G \in \cO$ such that $(F,G) \in \cR$. Moreover, for any $F \in \cO$ and $x$ such that $F(x) = 0$, suppose that \[\Pr_{G \gets \cR_F}[G(x) = 1] \leq \delta,\] where $\cR_F$ is the set of $G$ such that $(F,G) \in \cR$. Now, consider any oracle-aided unitary $\sU^F(\ket{\psi_F})$ that has oracle access to some $F \in \cO$, is initialized with some state $\ket{\psi_F}$ that may depend on $F$, makes $T$ queries, and outputs a state $\ket{\widetilde{\psi}_F}$. Then if $|\braket{\psi_F | \psi_G}| \geq c$ for all $(F,G) \in \cR$ and $\E_{(F,G) \gets \cR}[|\braket{\widetilde{\psi}_F|\widetilde{\psi}_G}|] \leq d$, then $T = \Omega\left(\frac{c-d}{\sqrt{\delta}}\right)$.
\end{importedtheorem}

Now, suppose there exists $\sU_\Com, \sU_\Open$ that violates \cref{claim:amp}. Recall that $\sU_\Com$ has access to the oracles $O[A]$ and $O[A^\bot]$, defined by the $n/2$-dimensional balanced affine subspace $A = S+v$ of $\bbF_2^{3n/4}$. We define a single functionality $F_A$ that takes as input $(b,s)$ and if $b=0$ outputs whether $s \in S + v$, and if $b=1$ outputs whether $s \in S^\bot$.


Then, we define a binary symmetric relation on functionalities $F_A, F_B$ as follows. Letting $A = S_A+v_A$ and $B = S_B+v_B$, we define $(F_A,F_B) \in \cR$ if and only if $\mathsf{dim}(A_0 \cap B_0) = n/2-2$ and $\mathsf{dim}(A_1 \cap B_1) = n/2-2$. Note that for any $(F_A,F_B) \in \cR$, $\mathsf{dim}(A \cap B) = n/2-1$.

Given $\cR$ defined this way, we see that for any fixed $F_A$ and $(b,s)$ such that $F_A(b,s) = 0$,

\begin{align*}
    \Pr_{F_B \gets \cR_{F_A}}&[F_B(b,s) = 1] \\ &\leq \max\left\{\frac{|B \setminus A|}{|\bbF_2^{3n/4} \setminus A \setminus \{0^{3n/4}\}|},\frac{|S_B^\bot \setminus S_A^\bot|}{|\bbF_2^{3n/4} \setminus S_A^\bot|}\right\} \\ &\leq \max\left\{\frac{2^{n/2-1}}{2^{3n/4} - 2^{n/2} - 1}, \frac{2^{n/4-1}}{2^{3n/4}-2^{n/4}}\right\}\\ &\leq \frac{1}{2^{n/4}}.
\end{align*}

Next, we note that $\sU_\Com^{F_A}$ is initialized with the state $\ket{A}$, and, for any $(A,B)$ such that $(F_A,F_B) \in \cR$, it holds that $|\braket{A|B}| = 1/2$. Our goal is then to bound \[\E_{(F_A,F_B) \gets \cR}\left[|\braket{\psi_A | \psi_B}|\right],\] where $\ket{\psi_A} = \sU_\Com^{F_A}(\ket{A})$. Since $(\sU_\Com, \sU_\Open)$ violates \cref{claim:amp}, we can write each $\ket{\psi_A}$ as $\ket{\psi_A'} + \ket{\psi_A^\err}$, where 

\[\big\| \ket{\psi_A^\err} \big\| \leq \frac{1}{2^{15\secp}}, ~~ \ket{\psi'_A} \in \mathsf{Im}(\Pi[A_0]), ~~ \text{and} ~~ \big\| \Pi[A_1]\sU_\Open\ket{\psi_A'}\big\| \geq \frac{1}{2^\secp}.\]

Thus, we have that \[\E_{(F_A,F_B) \gets \cR}[|\braket{\psi_A | \psi_B}] \leq \E_{(F_A,F_B) \gets \cR}[|\braket{\psi_A' | \psi_B'}|] + \frac{3}{2^{15\secp}}.\]

Now, we appeal to the following theorem, which is proven in \cref{appendix:inner-product}.

\begin{theorem}
Let $n,m,d \in \bbN, \epsilon \in (0,1/8)$ be such that $d \geq 2$ and $n-d+1 > 10\log(1/\epsilon)+6$. Let $\sU^{\regX,\regY}$ be any $(2^{n+m})$-dimensional unitary, where register $\regX$ is $2^n$ dimensions and register $\regY$ is $2^m$ dimensions. Let $\cA$ be the set of $d$-dimensional balanced affine subspaces $A = (A_0,A_1)$ of $\bbF_2^n$, where $A_0$ is the affine subspace of vectors in $A$ that start with 0 and $A_1$ is the affine subspace of vectors in $A$ that start with 1. For any $A = (A_0,A_1)$, let 

\[\Pi_{A_0} \coloneqq \sum_{v \in A_0}\ketbra{v}{v}^{\regX} \otimes \bbI^{\regY}, ~~~ \Pi_{A_1} \coloneqq \sU^\dagger\left(\sum_{v \in A_1}\ketbra{v}{v}^{\regX} \otimes \bbI^{\regY}\right)\sU.\]

Let $\cR$ be the set of pairs $(A,B)$ of $d$-dimensional affine subspaces of $\bbF_2^n$ such that $\dim(A_0 \cap B_0) = d-2$ and $\dim(A_1 \cap B_1) = d-2$. Then for any set of states $\{\ket{\psi_A}\}_A$ such that for all $A \in \cA$, $\ket{\psi_A} \in \mathsf{Im}(\Pi_{A_0})$, and $\|\Pi_{A_1}\ket{\psi_A}\| \geq \epsilon$, 
\[\E_{(A,B) \gets \cR}[|\braket{\psi_A | \psi_B}|] < \frac{1}{2}-\epsilon^{13}.\]
\end{theorem}

Setting $\epsilon = 1/2^\secp$, and noting that $3n/4 - n/2 + 1 > 11\secp > 10\log(2^\secp)+6$, this theorem implies that \[\E_{(F_A,F_B) \gets \cR}[|\braket{\psi_A' | \psi_B'}|] \leq \frac{1}{2} - \frac{1}{2^{13\secp}},\] and thus we conclude that \[\E_{(F_A,F_B) \gets \cR}[|\braket{\psi_A | \psi_B}] \leq \frac{1}{2} - \frac{1}{2^{14\secp}}.\]

Thus, by \cref{thm:adversary-method}, $\sU_\Com$ must be making \[\Omega\left(\frac{2^{n/8}}{2^{14\secp}}\right) = \Omega\left(2^{130\secp/8 - 14\secp}\right) > 2^{2\secp}\] oracle queries, recalling that $n \geq 130\secp$. However, $\sU_\Com$ was assumed to be making at most $2^{2\secp}$ queries, so this is a contradiction, completing the proof.

\end{proof}
\section{Verification of Quantum Partitioning Circuits}

\subsection{Definition}\label{subsec:partitioning-def}

A protocol for publicly-verifiable non-interactive classical verification of quantum partitioning circuits consists of the following procedures. We write the syntax in the \emph{oracle model}, where the prover obtains access to a classical oracle as part of its public key. We also specify a quantum proving key $\ket{\pk}$, but note that one could also consider the case where the proving key $\pk$ is classical.

\begin{itemize}
    \item $\Gen(1^\secp,Q) \to (\vk, \ket{\pk}, \cPK)$: The $\Gen$ algorithm takes as input the security parameter $1^\secp$ and the description of a quantum circuit $Q : \{0,1\}^{n'} \to \{0,1\}^n$, and outputs a classical verification key $\vk$ and a quantum proving key $(\ket{\pk},\cPK)$, which consists of a quantum state $\ket{\pk}$ and the description of a classical deterministic polynomial-time functionality $\cPK: \{0,1\}^* \to \{0,1\}^*$.
    \item $\Prove^{\cPK}(\ket{\pk},Q,x) \to \pi$: The $\Prove$ algorithm has oracle access to $\cPK$, takes as input the quantum proving key $\ket{\pk}$, a circuit $Q$, and an input $x \in \{0,1\}^{n'}$, and outputs a proof $\pi$.
    \item $\Ver(\vk,x,\pi) \to \{(q_1,\dots,q_m)\} \union \{\bot\}$: The classical $\Ver$ algorithm takes as input the verification key $\vk$, an input $x$, and a proof $\pi$, and either outputs a sequence of samples $(q_1,\dots,q_m)$ or $\bot$.
    \item $\Combine(b_1,\dots,b_m) \to b$: The $\Combine$ algorithm takes as input a sequence of bits $(b_1,\dots,b_m)$ and outputs a bit $b$.
\end{itemize}

The proof should satisfy the following notions of completeness and soundness.

\begin{definition}[Publicly-verifiable non-interactive classical verification of quantum partitioning circuits: Completeness]\label{def:PV-completeness}
A protocol for publicly-verifiable non-interactive classical verification of quantum partitioning circuits is \emph{complete} if for any family $\{Q_\secp,P_\secp\}_{\secp \in \bbN}$ such that $\{P_\secp \circ Q_\secp\}_{\secp \in \bbN}$ is pseudo-deterministic, and any sequence of inputs $\{x_\secp\}_{\secp \in \bbN}$, it holds that (where we leave indexing by $\secp$ implicit)


\[\Pr\left[\begin{array}{l} \Ver(\vk,x,\pi) = (q_1,\dots,q_m) \ \wedge \\ \Combine(P(q_1),\dots,P(q_m)) = P(Q(x)) \end{array}: \begin{array}{r} (\vk,\ket{\pk},\cPK) \gets \Gen(1^\secp,Q) \\ \pi \gets \Prove^{\cPK}(\ket{\pk},Q,x) \end{array}\right] = 1-\negl(\secp).\]

\end{definition}

We define soundness in the oracle model, where the adversarial prover gets access to an oracle for the functionality $\Ver(\vk,\cdot,\cdot)$. 

\begin{definition}[Publicly-verifiable non-interactive classical verification of quantum partitioning circuits: Soundness]\label{def:PV-soundness}
A protocol for publicly-verifiable non-interactive classical verification of quantum partitioning circuits is \emph{sound} if for any family $\{Q_\secp,P_\secp\}_{\secp \in \bbN}$ such that $\{P_\secp \circ Q_\secp\}_{\secp \in \bbN}$ is pseudo-deterministic, and any QPT adversarial prover $\{\sA_\secp\}_{\secp \in \bbN}$, it holds that (where we leave indexing by $\secp$ implicit)

\[\Pr\left[\begin{array}{l} \Ver(\vk,x,\pi) = (q_1,\dots,q_m) \ \wedge \\ \Combine(P(q_1),\dots,P(q_m)) = 1-P(Q(x)) \end{array}: \begin{array}{r} (\vk,\ket{\pk},\cPK) \gets \Gen(1^\secp,Q) \\ (x,\pi) \gets \sA^{\cPK,\Ver[\vk]}(\ket{\pk})  \end{array} \right] = \negl(\secp),\]

where $\Ver[\vk]$ is the classical functionality $\Ver(\vk,\cdot,\cdot) : (x,\pi) \to \{(q_1,\dots,q_m)\} \union \{\bot\}$.

\end{definition}

\subsection{$\text{QPIP}_1$ verification}\label{subsec:QPIP}

First, we recall an information-theoretic protocol for verifying quantum partitioning circuits using only single-qubit standard and Hadamard basis measurements.\footnote{Quantum interactive protocols where the verifier only requires the ability to measure single qubits have been referred to as $\text{QPIP}_1$ protocols.} This protocol is a $\secp$-wise parallel repetition of the quantum sampling verification protocol from \cite{EC:CLLW22}, and was described in \cite{TCC:Bartusek21}. Most of the underlying details of the protocol will not be important to us, but we provide a high-level description. 

The prover prepares multiple copies of a history state of the computation $Q(x)$, which is in general a sampling circuit. Each history state is prepared in a special way \cite{EC:CLLW22} to satisfy the following properties: (i) a sample approximately from the output distribution may be obtained by measuring certain registers of the state in the \emph{standard basis}, which can be achieved by adding enough dummy identity gates to ensure that the output state is a large fraction of the history state, and (ii) the history state is the \emph{unique} ground state of the Hamiltonian, and all orthogonal states have much higher energy, ensuring that the verifier can test the validity of the entire computation by testing the energy of the history state. 

Then, the verifier samples certain copies for \emph{verifying} and other copies for \emph{sampling}. In the verify copies, it samples a random Hamiltonian term, and measures in the corresponding standard and Hadamard bases, while in the sample copies, the verifier measures the output register in the standard basis. If the verifier accepts the results from measuring the verify copies, it outputs the collection of samples obtained from the sample copies. It was shown by \cite{TCC:Bartusek21} that if $Q$ is a partitioning circuit with predicate $P$, then one can set parameters so that conditioned on verification passing, it holds with \emph{overwhelming probability} that \emph{at least half} of the output samples $q_t$ are such that $P(q_t) = P(Q(x))$. We describe the formal syntax of this protocol in \cref{fig:syntax-QPIP1}, where the prover state $\ket{\psi}$ consists of sufficiently many copies of the history state, and the verifier's string $h$ of measurement bases consists of (mostly) indices used for verification as well as some indices used for sampling outputs, which we denote by $S$. By an observation of \cite{TCC:ACGH20}, the sampling of $h$ can be performed independently of the input $x$, which is reflected in the syntax of \cref{fig:syntax-QPIP1} (technically, it only needs the size $|Q|$ rather than $Q$ itself).

Next, we introduce some notation, and then state the correctness and soundness guarantees of this protocol that follow from prior work.


\begin{definition}
Define $\Maj$ to be the predicate that takes as input a set of bits $\{b_i\}_i$ and outputs the most frequently occurring bit $b$. In the event of a tie, we arbitrarily set the output to 0.
\end{definition}

\begin{definition}
For a string $x \in \{0,1\}^n$ and a subset $S \subseteq [n]$, define $x[S]$ to be the string consisting of bits $\{x_i\}_{i \in S}$.
\end{definition}

\begin{definition}
Given an $h \in \{0,1\}^n$ and an $n$-qubit state $\ket{\psi}$, let $M(h,\ket{\psi})$ denote the distribution over $n$-bit strings that results from measuring each qubit $i$ of $\ket{\psi}$ in basis $h_i$, where the bit $h_i = 0$ indicates standard basis and $h_i = 1$ indicates Hadamard basis.
\end{definition}

\begin{importedtheorem}[\cite{EC:CLLW22,TCC:Bartusek21}]\label{impthm:QV}
The protocol $\Pi^\QV$ (\cref{fig:syntax-QPIP1}) that satisfies the following properties.

\begin{itemize}
    \item \textbf{Completeness.} For any family $\{Q_\secp,P_\secp\}_{\secp \in \bbN}$ such that $\{P_\secp \circ Q_\secp,\}_{\secp \in \bbN}$ is pseudo-deterministic, and any sequence of inputs $\{x_\secp\}_{\secp \in \bbN}$,
    \[\Pr\left[\sV_\Ver^\QV(Q,x,h,m) = \top \wedge \Maj(\{P(q_t)\}_t) = P(Q(x)): \begin{array}{r}  \ket{\psi} \gets \sP^{\QV}(1^\secp,Q,x) \\ (h,S) \gets \sV_\Gen^{\QV}(1^\secp,Q) \\ m \gets M(h,\ket{\psi}) \\ \{q_t\}_{t \in [\secp]} \coloneqq m[S] \end{array}\right] = 1-\negl(\secp).\]
    \item \textbf{Soundness.} For any family $\{Q_\secp,P_\secp\}_{\secp \in \bbN}$ such that $\{P_\secp \circ Q_\secp\}_{\secp \in \bbN}$ is pseudo-deterministic, any sequence of inputs $\{x_\secp\}_{\secp \in \bbN}$, and any sequence of states $\{\ket{\psi^*_\secp}\}_{\secp \in \bbN}$,
    \[\Pr\left[\sV_\Ver^\QV(Q,x,h,m) = \top \wedge \Maj(\{P(q_t)\}_t) = 1-P(Q(x)) : \begin{array}{r}  (h,S) \gets \sV_\Gen^{\QV}(1^\secp,Q) \\ m \gets M(h,\ket{\psi^*}) \\ \{q_t\}_{t \in [\secp]} \coloneqq m[S] \end{array}\right] = \negl(\secp).\]
\end{itemize}
\end{importedtheorem}

\protocol{$\text{QPIP}_1$ protocol $\Pi^\QV = \left(\sP^\QV,\sV_\Gen^\QV,\sV_\Ver^\QV\right)$}{Syntax for a $\text{QPIP}_1$ protocol that verifies the output of a quantum partitioning circuit $Q$.}{fig:syntax-QPIP1}{

Parameters: Number of bits $n$ output by $Q$, and number of qubits $\ell = \ell(\secp)$ in the prover's state.\\

\textbf{Prover's computation}

\begin{itemize}
    \item $\sP^\QV(1^\secp,Q,x) \to \ket{\psi}:$ on input the security parameter $1^\secp$, the description of a quantum circuit $Q$, and an input $x$, the prover prepares a state $\ket{\psi}$ on $\ell$ qubits, and sends it to the verifier.
\end{itemize}

\textbf{Verifier's computation}

\begin{itemize}
    \item $\sV_\Gen^\QV(1^\secp,Q) \to (h,S):$ on input the security parameter $1^\secp$ and the description of a quantum circuit $Q$, the verifier's $\Gen$ algorithm samples a string $h \in \{0,1\}^\ell$ and a subset $S \subset [\ell]$ of size $n \cdot \secp$ with the property that for all $i \in S$, $h_i = 0$.
    \item Next, the verifier measures $m \gets M(h,\ket{\psi})$ to obtain a string of measurement results $m \in \{0,1\}^\ell$. 
    \item $\sV_\Ver^\QV(Q,x,h,m) \to \{\top,\bot\}:$ on input a circuit $Q$, input $x$, string of bases $h$, and measurement results $m$, the verifier's $\Ver$ algorithm outputs $\top$ or $\bot$.
    \item If $\top$, the verifier outputs the string $m[S]$ which is parsed as $\{q_t\}_{t \in [\secp]}$ where each $q_t \in \{0,1\}^{n}$ and otherwise the verifier outputs $\bot$.
\end{itemize}
}

\subsection{Classical verification}\label{subsec:classical-verification}

Next, we compile the above information-theoretic protocol into a classically-verifiable but computationally-sound protocol, using Mahadev's measurement protocol \cite{SIAMCOMP:Mahadev22}. The measurement protocol itself is a four-message protocol with a single bit challenge from the verifier. Then, we apply  parallel repetition and Fiat-Shamir, following \cite{TCC:ACGH20,TCC:ChiChuYam20,TCC:Bartusek21}, which results in a two-message negligibly-sound protocol in the quantum random oracle model. 

The resulting protocol $\Pi^\CV = (\sP_\Prep^\CV,\sV^\CV_\Gen,\sP_\Prove^\CV,\sP_\Meas^\CV,\sV^\CV_\Ver)$ makes use of a dual-mode randomized trapdoor claw-free hash function $(\TCF.\Gen, \allowbreak\TCF.\Eval, \allowbreak \TCF.\Invert,\allowbreak \TCF.\Ch,\allowbreak \TCF.\IsValid)$ (\cref{def:clawfree}), and is described in \cref{fig:CVQC}. We choose to explicitly split the second prover's algorithm into two parts $\sP_\Prove^\CV$ and $\sP_\Meas^\CV$ for ease of notation when we build on top of this protocol in the next section.

\protocol{Classically-verifiable protocol $\Pi^\CV = \left(\sP_\Prep^\CV,\sV^\CV_\Gen,\sP_\Prove^\CV,\sP_\Meas^\CV,\sV^\CV_\Ver\right)$}{Two-message protocol for verifying quantum partitioning circuits with a classical verifier.}{fig:CVQC}{
\textbf{Parameters}: Number of qubits per round $\ell \coloneqq \ell(\secp)$, number of parallel rounds $r \coloneqq r(\secp)$, number of Hadamard rounds $k \coloneqq k(\secp)$, and random oracle $H: \{0,1\}^* \to \{0,1\}^{\log\binom{r}{k}}$.

\begin{itemize}
    \item $\sP_\Prep^\CV(1^\secp,Q,x) \to (\regB_1,\dots,\regB_r)$: For each $i \in [r]$, prepare the state $\ket{\psi_i} \coloneqq \sP^\QV(1^\secp,Q,x)$ on register $\regB_i = (\regB_{i,1},\dots,\regB_{i,\ell})$, which we write as \[\ket{\psi_i} \coloneqq \sum_{v \in \{0,1\}^\ell}\alpha_v\ket{v}^{\regB_i}.\]
    \item $\sV_\Gen^\CV(1^\secp,Q) \to (\pp,\sparam)$: For each $i \in [r]$, sample $(h_i,S_i) \gets \sV^\QV_\Gen(1^\secp,Q)$ where $h_i = (h_{i,1},\dots,h_{i,\ell})$, and sample $\{(\pk_{i,j},\sk_{i,j}) \gets \TCF.\Gen(1^\secp,h_{i,j})\}_{j \in [\ell]}$. Then, set \[\pp \coloneqq \{\{\pk_{i,j}\}_{j \in [\ell]}\}_{i \in [r]}, \sparam \coloneqq \{h_i,S_i,\{\sk_{i,j}\}_{j \in [\ell]}\}_{i \in [r]}.\]
    \item $\sP_\Prove^\CV(\regB_1,\dots,\regB_r,\pp) \to (\regB_1,\dots,\regB_r,\{y_{i,j},z_{i,j}\}_{i \in [r],j \in [\ell]})$:
    \begin{itemize}
        \item Do the following for each $i \in [r]$: For each $j \in [\ell]$, apply $\TCF.\Eval[\pk_{i,j}](\regB_{i,j}) \to (\regB_{i,j},\regZ_{i,j},\regY_{i,j})$, resulting in the state
            \[\sum_{v \in \{0,1\}^\ell}\alpha_v\ket{v}^{\regB_i}\ket{\psi_{\pk_{i,1},v_1}}^{\regZ_{i,1},\regY_{i,1}},\dots,\ket{\psi_{\pk_{i,\ell},v_\ell}}^{\regZ_{i,\ell},\regY_{i,\ell}},\]
            and measure registers $\regY_{i,1},\dots,\regY_{i,\ell}$ in the standard basis to obtain strings $y_{i,1},\dots,y_{i,\ell}$.
        \item Compute $T \coloneqq H(y_{1,1},\dots,y_{r,\ell})$, where $T \in \{0,1\}^r$ with Hamming weight $k$.
        \item For each $i: T_i = 0$, measure $\regZ_{i,1},\dots,\regZ_{i,\ell}$ in the standard basis to obtain strings $z_{i,1},\dots,z_{i,\ell}$.
        \item For each $i: T_i = 1$, apply $J(\cdot)$ coherently to each register $\regZ_{i,1},\dots,\regZ_{i,\ell}$ and then measure in the Hadamard basis to obtain strings $z_{i,1},\dots,z_{i,\ell}$.
    \end{itemize}
    \item $\sP^\CV_\Meas(\regB_1,\dots,\regB_r) \to \{b_{i,j}\}_{i \in [r], j \in [\ell]}$: Measure registers $\{\regB_{i,j}\}_{i: T_i = 0, j \in [\ell]}$ in the standard basis to obtain bits $\{b_{i,j}\}_{i: T_i = 0, j \in [\ell]}$ and measure registers $\{\regB_{i,j}\}_{i: T_i = 1, j \in [\ell]}$ in the Hadamard basis to obtain bits $\{b_{i,j}\}_{i: T_i = 1, j \in [\ell]}$. 
    \item $\sV_\Ver^\CV(Q,x,\sparam,\pi) \to \{\{\{q_{i,t}\}_{t \in [\secp]}\}_{i : T_i = 1}\} \cup \{\bot\}$:
    \begin{itemize}
        \item Parse $\pi \coloneqq \{b_{i,j},y_{i,j},z_{i,j}\}_{i \in [r], j \in [\ell]}$ and compute $T \coloneqq H(y_{1,1},\dots,y_{r,\ell})$.
        \item For each $i: T_i = 0$ and $j \in [\ell]$, compute $\TCF.\Ch(\pk_{i,j},b_{i,j},z_{i,j},y_{i,j})$. If any are $\bot$, then output $\bot$.
        \item For each $i: T_i = 1$, do the following.
        \begin{itemize}
            \item For each $j \in [\ell]$: If $h_{i,j} = 0$, compute $\TCF.\Invert(0,\sk_{i,j},y_{i,j})$, output $\bot$ if the output is $\bot$, and otherwise parse the output as $(m_{i,j},x_{i,j})$. If $h_{i,j} = 1$, compute $\TCF.\Invert(1,\sk_{i,j},y_{i,j})$, output $\bot$ if the output is $\bot$, and otherwise parse the output as $(0,x_{i,j,0}),(1,x_{i,j,1})$. Then, check $\TCF.\IsValid(x_{i,j,0},x_{i,j,1},z_{i,j})$ and output $\bot$ if the result is $\bot$. Finally, set $m_{i,j} \coloneqq b_{i,j} \oplus z_{i,j} \cdot (J(x_{i,j,0}) \oplus J(x_{i,j,1}))$.
            \item Let $m_i  = (m_{i,1},\dots,m_{i,\ell})$, compute $\sV_\Ver^\QV(Q,x,h_i,m_i)$, output $\bot$ if the result is $\bot$, and otherwise set $\{q_{i,t}\}_{t \in [\secp]} \coloneqq m_i[S_i]$.
        \end{itemize}
        \item Output $\{\{q_{i,t}\}_{t \in [\secp]}\}_{i:T_i = 1}$.
    \end{itemize}
\end{itemize}
}

We introduce some notation needed for describing the security properties of this protocol. 

\begin{itemize}
    \item Fix a security parameter $\secp$, circuit $Q$, input $x$, and parameters $(\pp,\sparam) \in \sV_\Gen^{\CV}(1^\secp,Q)$.
    \item Based on $\sparam = \{h_i,S_i,\{\sk_{i,j}\}_{j \in [\ell]}\}_{i \in [r]}$, we define the set $S \coloneqq \{S_i\}_{i \in [r]}$. For any proof $\pi = \{b_{i,j},y_{i,j},z_{i,j}\}_{i,j}$ generated by $\sP^\CV$, we let $w \coloneqq \TestOut[\sparam](\pi)$ be a string $w \in \{0,1\}^{|S|}$ defined as follows. Let $T \coloneqq H(y_{1,1},\dots,y_{r,\ell})$. The string $w$ consists of $r$ sub-strings $w_1,\dots,w_r$, where for each $i : T_i = 0$, $w_i$ consists of the bits $\{b_{i,j}\}_{j \in S_i}$, and for each $i : T_i = 1$, $w_i = 0^{|S_i|}$. 
    \item For any predicate $P$ and bit $b \in \{0,1\}$, we define the set $D_\inside[P,b] \subset \{0,1\}^{|S|}$ to consist of $w \coloneqq (w_1,\dots,w_r)$ with the following property. There are at least $3/4$ fraction of $w_i$ such that, parsing $w_i$ as $(w_{i,1},\dots,w_{i,\secp})$, it holds that  $\Maj\left(\left\{P(w_{i,t})\right\}_{t \in [\secp]}\right) = b$. 
    \item For any predicate $P$ and bit $b \in \{0,1\}$, we define the set $D_\out[P,b] \subset \{0,1\}^{|S|}$ to consist of $w \coloneqq (w_1,\dots,w_r)$ with the following property. There are at least 1/3 fraction of $w_i$ such that, parsing $w_i$ as $(w_{i,1},\dots,w_{i,\secp})$, it holds that $\Maj\left(\left\{P(w_{i,t})\right\}_{t \in [\secp]}\right) = 1-b$.

\end{itemize}

Note that for any predicate $P$ and $b \in \{0,1\}$, $D_\inside[P,b]$ and $D_\out[P,b]$ are disjoint sets of strings.

Now, we state four properties that $\Pi^\CV$ satisfies. The proof of \cref{lemma:CV-completeness} follows immediately from the completeness of $\Pi^\QV$ (\cref{impthm:QV}) and the correctness of the dual-mode randomized trapdoor claw-free hash function (\cref{def:clawfree}). The proofs of the remaining three lemmas mostly follow from the prior work of \cite{TCC:Bartusek21}, and we show this formally in \cref{sec:appendix-CVQC}.

\begin{definition}\label{def:maj-maj}
Let $\MM_\secp$ be the predicate that takes as input a set of bits $\{\{b_{i,t}\}_{t \in [\secp]}\}_i$, and outputs the bit \[\MM_\secp(\{\{b_{i,t}\}_{t \in [\secp]}\}_i) \coloneqq \Maj\left(\left\{\Maj\left(\{b_{i,t}\}_{t \in [\secp]}\right)\right\}_i\right).\]
\end{definition}

\begin{lemma}[Completeness]\label{lemma:CV-completeness}
The protocol $\Pi^\CV$ (\cref{fig:CVQC}) with $r(\secp) = \secp^2$ and $k(\secp) = \secp$ satisfies \emph{completeness}, which stipulates that for any family $\{Q_\secp,P_\secp\}_{\secp \in \bbN}$ such that $\{P_\secp \circ Q_\secp\}_{\secp \in \bbN}$ is pseudo-deterministic and sequence of inputs $\{x_\secp\}_{\secp \in \bbN}$,

\[\Pr\left[\begin{array}{l} \sV^\CV_\Ver(Q,x,\sparam,\pi) = \{\{q_{i,t}\}_{t \in [\secp]}\}_{i : T_i = 1} ~~ \wedge\\ \MM_\secp(\{\{P(q_{i,t})\}_{t \in [\secp]}\}_{i: T_i = 1}) = P(Q(x)) \end{array}: \begin{array}{r}(\regB_1,\dots,\regB_r) \gets \sP_\Prep^\CV(1^\secp,Q,x) \\ (\pp,\sparam) \gets \sV^\CV_\Gen(1^\secp,Q) \\ \{y_{i,j},z_{i,j}\}_{i \in [r], j \in [\ell]} \gets \sP_\Prove^\CV(\regB_1,\dots,\regB_r,\pp) \\ \{b_{i,j}\}_{i \in [r], j \in [\ell]} \gets \sP_\Meas^\CV(\regB_1,\dots,\regB_r) \\ \pi \coloneqq \{b_{i,j},y_{i,j},z_{i,j}\}_{i \in [r], j \in [\ell]}\end{array}\right] = 1-\negl(\secp).\]

\end{lemma}

\begin{lemma}[Soundness]\label{lemma:soundness}
The protocol $\Pi^\CV$ (\cref{fig:CVQC}) with $r(\secp) = \secp^2$ and $k(\secp) = \secp$ satisfies \emph{soundness}, which stipulates that for any family $\{Q_\secp,P_\secp\}_{\secp \in \bbN}$ such that $\{P_\secp \circ Q_\secp\}_{\secp \in \bbN}$ is pseudo-deterministic, sequence of inputs $\{x_\secp\}_{\secp \in \bbN}$, and QPT adversary $\{\sA_\secp\}_{\secp \in \bbN}$, it holds that 
    \[\Pr\left[\begin{array}{l} \sV^\CV_\Ver(Q,x,\sparam,\pi) = \{\{q_{i,t}\}_{t \in [\secp]}\}_{i : T_i = 1} ~~ \wedge\\ \MM_\secp(\{\{P(q_{i,t})\}_{t \in [\secp]}\}_{i: T_i = 1}) = 1-P(Q(x)) \end{array}: \begin{array}{r}(\pp,\sparam) \gets \sV^\CV_\Gen(1^\secp,Q) \\ \pi \gets \sA(\pp) \end{array}\right] = \negl(\secp).\]
\end{lemma}

\begin{lemma}[$D_{\inside}$ if accept]\label{lemma:Q-in}
The protocol $\Pi^\CV$ (\cref{fig:CVQC}) with $r(\secp) = \secp^2$ and $k(\secp) = \secp$ satisfies the following property. For any family $\{Q_\secp,P_\secp\}_{\secp \in \bbN}$ such that $\{P_\secp \circ Q_\secp\}_{\secp \in \bbN}$ is pseudo-deterministic, sequence of inputs $\{x_\secp\}_{\secp \in \bbN}$, and QPT adversary $\{\sA_\secp\}_{\secp \in \bbN}$, it holds that 
        \[\Pr\left[\begin{array}{l}\sV^\CV_\Ver(Q,x,\sparam,\pi) \neq \bot ~~ \wedge \\ w \notin D_{\inside}[P,P (Q(x))]\end{array} : \begin{array}{r}(\pp,\sparam) \gets \sV^\CV_\Gen(1^\secp,Q) \\ \pi \gets \sA(\pp) \\ w \coloneqq \mathsf{TestRoundOutputs}[\sparam](\pi)\end{array}\right] = \negl(\secp).\]
\end{lemma}

\begin{lemma}[$D_{\out}$ if accept wrong output]\label{lemma:Q-out}
The protocol $\Pi^\CV$ (\cref{fig:CVQC}) with $r(\secp) = \secp^2$ and $k(\secp) = \secp$ satisfies the following property. For any family $\{Q_\secp,P_\secp\}_{\secp \in \bbN}$ such that $\{P_\secp \circ Q_\secp\}_{\secp \in \bbN}$ is pseudo-deterministic, sequence of inputs $\{x_\secp\}_{\secp \in \bbN}$, and QPT adversary $\{\sA_\secp\}_{\secp \in \bbN}$, it holds that 
        \[\Pr\left[\begin{array}{l} \sV^\CV_\Ver(Q,x,\sparam,\pi) = \{\{q_{i,t}\}_{t \in [\secp]}\}_{i : T_i = 1} ~~ \wedge\\ \MM_\secp(\{\{P(q_{i,t})\}_{t \in [\secp]}\}_{i: T_i = 1}) = 1-P(Q(x)) ~~ \wedge \\ w \notin D_{\out}[P,P(Q(x))]\end{array} : \begin{array}{r}(\pp,\sparam) \gets \sV^\CV_\Gen(1^\secp,Q) \\ \pi \gets \sA(\pp,\sparam) \\ w \coloneqq \mathsf{TestRoundOutputs}[\sparam](\pi)\end{array}\right] = \negl(\secp).\]
\end{lemma}

Note that in this final lemma, $\sA_\secp$ is given access to $\sparam$, so this does not trivially follow from soundness.

\subsection{Public verification}\label{subsec:public-verification}

Next, we compile the above protocol into a \emph{publicly-verifiable} protocol for quantum partitioning circuits in the oracle model. We will use the following ingredients in addition to $\Pi^\CV$ (\proref{fig:CVQC}).

\begin{itemize}
    \item A Pauli functional commitment $\PFC = (\PFC.\Gen,\allowbreak\PFC.\Com,\allowbreak\PFC.\OpenZ,\allowbreak\PFC.\OpenX,\allowbreak\PFC.\DecZ,\allowbreak\PFC.\DecX)$ that satisfies \emph{string binding with public decodability} (\cref{def:predicate-binding}).
    \item A strongly unforgeable signature token scheme $\Tok = (\Tok.\Gen,\Tok.\Sign,\Tok.\Verify)$ (\cref{def:strong-unforgeability}).
    \item A pseudorandom function $F_k$ secure against superposition-query attacks \cite{6375347}.
\end{itemize}

\protocol{Publicly-verifiable protocol $\Pi^{\PV} = \left(\PV.\Gen, \PV.\Prove, \PV.\Ver, \PV.\Out\right)$}{Publicly-verifiable non-interactive classical verification of quantum partitioning circuits.}{fig:publicly-verifiable}{
\textbf{Parameters}: Let $\secp$ be the security parameter and define parameters $(\ell,r,k)$ as in $\Pi^\CV$ (\cref{fig:CVQC}).
\begin{itemize}
    \item $\PV.\Gen(1^\secp,Q) \to (\vk,\ket{\pk},\cPK)$:
    \begin{itemize}
        \item Sample $\{(\dk_{i,j},\ket{\ck_{i,j}},\cCK_{i,j}) \gets \PFC.\Gen(1^\secp)\}_{i \in [r], j \in [\ell]}$.
        \item Sample $(\vk_\Tok,\ket{\sk_\Tok}) \gets \Tok.\Gen(1^\secp)$.
        \item Sample PRF keys $k_1,k_2 \gets \{0,1\}^\secp$.
        \item Define the functionality $\sH(\cdot) \coloneqq F_{k_1}(\cdot)$, which will be used as the random oracle $H$ in $\Pi^\CV$.
        \item Define the functionality $\CVGen(\cdot)$ as follows, where its input is parsed as $(x,c,\sigma)$. 
        \begin{itemize}
            \item If $\Tok.\Verify(\vk_\Tok,(x,c),\sigma) = \top$ then continue, and otherwise return $\bot$.
            \item Compute $(\pp,\sparam) \coloneqq \sV_\Gen^\CV(1^\secp,Q;F_{k_2}(x,c,\sigma))$ and output $\pp$. 
        \end{itemize} 
        \item Set $\vk \coloneqq (Q,k_1,k_2,\vk_\Tok,\{\dk_{i,j}\}_{i \in [r],j \in [\ell]})$, $\ket{\pk} \coloneqq (\ket{\sk_\Tok},\{\ket{\ck_{i,j}}\}_{i \in [r], j \in [\ell]})$, and $\cPK \coloneqq (\sH,\CVGen,\{\cCK_{i,j}\}_{i \in [r], j \in [\ell]}).$
    \end{itemize}
    
    \item $\PV.\Prove^{\cPK}(\ket{\pk},Q,x) \to \pi$:
    \begin{itemize}
        \item Prepare $\ket{\psi_1}^{\regB_1},\dots,\ket{\psi_r}^{\regB_r} \gets \sP_\Prep^\CV(1^\secp,Q,x)$.
        \item For each $i \in [r], j \in [\ell]$ apply $\PFC.\Com^{\cCK_{i,j}}(\regB_{i,j},\ket{\ck_{i,j}}) \to (\regB_{i,j},\regU_{i,j},c_{i,j})$ (see \cref{def:PFC}).
        \item Set $c \coloneqq (c_{1,1},\dots,c_{r,\ell})$, compute $\sigma \gets \Tok.\Sign((x,c),\ket{\sk_\Tok})$, and compute $\pp \coloneqq \CVGen(x,c,\sigma)$.
        \item Apply $\sP^\CV_\Prove(\regB_1,\dots,\regB_r,\pp) \to (\regB_1,\dots,\regB_r,\{y_{i,j},z_{i,j}\}_{i \in [r], j \in [\ell]})$, and define $T \coloneqq \sH(y_{1,1},\dots,y_{r,\ell})$
        \item For each $i: T_i = 0, j \in [\ell]$, apply $\PFC.\OpenZ(\regB_{i,j},\regU_{i,j}) \to u_{i,j}$.
        \item For each $i: T_i = 1, j \in [\ell]$, apply $\PFC.\OpenX(\regB_{i,j},\regU_{i,j}) \to u_{i,j}$.
        \item Set $\pi \coloneqq (c,\sigma,\{u_{i,j},y_{i,j},z_{i,j}\}_{i \in [r], j \in [\ell]})$.
    \end{itemize}
        
    \item $\PV.\Ver(\vk,x,\pi) \to \{\{\{q_{i,t}\}_{t \in [\secp]}\}_{i : T_i = 1}\} \cup \{\bot\}$:   
        \begin{itemize}
            \item Parse $\vk \coloneqq (Q,k_1,k_2,\vk_\Tok,\{\dk_{i,j}\}_{i \in [r],j \in [\ell]})$ and $\pi \coloneqq (c,\sigma,\mu)$. 
            \item If $\Tok.\Verify(\vk_\Tok,(x,c),\sigma) = \top$, then set $(\pp,\sparam) \coloneqq \sV_\Gen^\CV(1^\secp,Q;F_{k_2}(x,c,\sigma))$, and let $\{h_i\}_{i \in [r]}$ be the string of basis choices defined by $\sparam$. Otherwise, return $\bot$.
            \item Parse $\mu$ as $\{u_{i,j},y_{i,j},z_{i,j}\}_{i \in [r], j \in [\ell]}$, and define $T \coloneqq F_{k_1}(y_{1,1},\dots,y_{r,\ell})$.
            \item For all $i: T_i = 0, j \in [\ell]$, compute $b_{i,j} \coloneqq \PFC.\DecZ(\dk_{i,j},c_{i,j},u_{i,j})$, and return $\bot$ if $b_{i,j} = \bot$.
            \item For all $i: T_i = 1, j \in [\ell]$ such that $h_{i,j} = 1$, compute $b_{i,j} \coloneqq \PFC.\DecX(\dk_{i,j},c_{i,j},u_{i,j})$, and return $\bot$ if $b_{i,j} = \bot$.
            \item For all $i: T_i = 1, j \in [\ell]$ such that $h_{i,j} = 0$, set $b_{i,j} = 0$.
            \item Let $\widetilde{\pi} \coloneqq \{b_{i,j},y_{i,j},z_{i,j}\}_{i \in [r], j \in [\ell]}$ and return $\{\{q_{i,t}\}_{t \in [\secp]}\}_{i : T_i = 1} \coloneqq \sV_\Ver^\CV(Q,x,\sparam,\widetilde{\pi})$.
        \end{itemize}
        
    \item $\PV.\Combine \equiv \MM_\secp$ (see \cref{def:maj-maj}).

\end{itemize}
}

\begin{theorem}\label{thm:PV-soundness}
The protocol $\Pi^\PV$ (\cref{fig:publicly-verifiable}) satisfies \cref{def:PV-completeness} and \cref{def:PV-soundness}.
\end{theorem}

\begin{proof} We argue completeness (\cref{def:PV-completeness}) and soundness (\cref{def:PV-soundness}).\\

\noindent\underline{Completeness.} Consider some circuit $Q$, input $x$, and sample $(\vk,\ket{\pk},\cPK) \gets \PV.\Gen(1^\secp,Q)$. By the correctness of $\Tok$ (\cref{def:token-correctness}), we know that the call to $\CVGen$ during $\PV.\Prove^{\cPK}(\ket{\pk},Q,x)$ only outputs $\bot$ with $\negl(\secp)$ probability. Also, by the security of the PRF, we can answer this query using uniformly sampled random coins $s$ in place of $F_{k_2}(x,c,\sigma)$. 

Now, imagine sampling $s$ and fixing $(\pp,\sparam) \coloneqq \sV_\Gen^\CV(1^\secp,Q;s)$ before computing $\PV.\Prove^{\cPK}(\ket{\pk},Q,x)$. Then, since $\pp$ no longer depends on $c$, we can move the application of each $\PFC.\Com^{\cCK_{i,j}}(\regB_{i,j},\ket{\ck_{i,j}})$ past the computation of $\pp$, and thus right before $\sP_\Prove^\CV(\regB_1,\dots,\regB_r,\pp)$. Moreover, since both $\PFC.\Com$ and $\sP_\Prove^\CV$ are \emph{classically controlled} on registers $\regB_1,\dots,\regB_r$, and otherwise operate on disjoint registers, we can further commute each $\PFC.\Com$ past $\sP_\Prove^\CV$.

Then, the bits $\{b_{i,j}\}_{i,j}$ for $i: T_i = 0$ computed during $\PV.\Ver(\vk,x,\pi)$ are now computed by applying $\PFC.\Com, \PFC.\OpenZ$, and $\PFC.\DecZ$ in succession to $\regB_{i,j}$, and the bits $\{b_{i,j}\}_{i,j}$ for $i: T_i = 1, h_{i,j} = 1$ computed during $\PV.\Ver(\vk,x,\pi)$ are now computed by applying $\PFC.\Com, \PFC.\OpenX$, and $\PFC.\DecX$ in succession to $\regB_{i,j}$. Thus, by the correctness of $\PFC$ (\cref{def:poc-correctness}), we can replace these operations by directly measuring $\regB_{i,j}$ in the standard (resp. Hadamard) basis. Now, completeness follows directly from the completeness of $\Pi^\CV$ (\cref{lemma:CV-completeness}), since the remaining bits $\{b_{i,j}\}_{i,j}$ for $i: T_i = 1, h_{i,j} = 0$ (which are arbitrarily set to 0 in $\PV.\Ver$) are ignored by $\sV_\Ver^\CV$, and the rest of $\widetilde{\pi}$ is now computed by applying $\sP_\Prove^\CV$ followed by $\sP_\Meas^\CV$ to $\regB_1,\dots,\regB_r$.\\


\noindent\underline{Soundness.} Before getting into the formal proof, we provide a high-level overview. We will go via the following steps.

\begin{itemize}
    \item $\sA_1$: Begin with an adversary $\sA_1$ that is assumed to violate soundness of the protocol. Thus, with $\nonnegl(\secp)$ probability, it's final (classical) output consists of an input $x^*$ and a proof $\pi^*$ such that $\PV.\Ver(\vk,x^*,\pi^*) \neq \bot$ and $\PV.\Out(\PV.\Ver(\vk,x^*,\pi^*),P) \neq P(Q(x^*))$.
    \item $\sA_2$: Replace $F_{k_2}$ with a random oracle, and call the resulting oracle algorithm $\sA_2$.
    \item $\sA_3$: Apply Measure-and-Reprogram (\cref{thm:measure-and-reprogram}) to obtain a two-stage adversary $\sA_3$, where the first stage outputs $x^*$, a $\PFC$ commitment $c^*$, and a token signature $\sigma^*$, and the second stage outputs the remainder $\mu^*$ of the proof $\pi^* \coloneqq (c^*,\sigma^*,\mu^*)$. The parameters $(\pp_{x^*,c^*,\sigma^*},\sparam_{x^*,c^*,\sigma^*})$ for $\Pi^\CV$ are re-sampled at the beginning of the second stage.
    \item $\sA_4$: Use the strong unforgeability of the signature token scheme (\cref{def:strong-unforgeability}) to argue that during the second stage of $\sA_3$, all queries to $\PV.\Ver$ except for $(x^*,c^*,\sigma^*)$ can be ignored. Call the resulting adversary $\sA_4$.
    \item $D_\out[P,P(Q(x^*))]$: Appeal to \cref{lemma:Q-out} to show that whenever $\sA_4$ breaks soundness, its output yields a proof $\widetilde{\pi}$ for $\Pi^\CV$ such that \[\TestOut[\sparam_{x^*,c^*,\sigma^*}](\widetilde{\pi}) \in D_\out[P,P(Q(x^*))].\]  
    \item $\cH_0,\dots,\cH_p$: Define a hybrid for each of the $p = \poly(\secp)$ queries that the second stage of $\sA_4$ makes to $\PV.\Ver$. In each hybrid $\iota$, begin answering query $\iota$ with $\bot$, and let $\Pr[\cH_\iota = 1]$ be the probability that $\sA_4$ still breaks soundness.
    \item $\Pr[\cH_0 = 1] = \nonnegl(\secp)$: This has already been proven, by assumption that $\sA_1$ breaks soundness with $\nonnegl(\secp)$ probability, and the hybrids above.
    \item $\Pr[\cH_p = 1] = \negl(\secp)$: This is implied by the soundness of $\Pi^\CV$ (\cref{lemma:soundness}) because in this experiment, $\sA_4$ does not have access to $\sparam_{x^*,c^*,\sigma^*}$ before producing its final proof.
    \item $\Pr[\cH_\iota = 1] \geq \Pr[\cH_{\iota - 1} = 1] - \negl(\secp)$: This is proven in two parts. 
    
    \begin{enumerate}
        \item By \cref{lemma:Q-in}, we can say that since $\sA_4$ does not have access to $\sparam_{x^*,c^*,\sigma^*}$ before preparing its $\iota$'th query, each classical basis state in the query superposition that is not answered with $\bot$ yields a proof $\widetilde{\pi}$ for $\Pi^\CV$ such that \[\TestOut[\sparam_{x^*,c^*,\sigma^*}](\widetilde{\pi}) \in D_\inside[P,P(Q(x^*))].\] 
        \item We appeal to the string binding with public decodability of $\PFC$ (\cref{def:predicate-binding}) to show that replacing these answers with $\bot$ only affects the probability that $\sA_4$ breaks soundness by a negligible amount. 
        
        This follows because any part of the query that contains $\PFC$ openings for a string in $D_\inside[P,P(Q(x^*))]$ cannot have noticeable overlap with the part of the state (after running the rest of $\sA_4$) that contains $\PFC$ openings for a string in $D_\out[P,P(Q(x^*))]$. Otherwise, we can prepare an adversarial committer, where the part of $\sA_4$ up to query $\iota$ is the ``Commit'' stage, and the remainder of $\sA_4$ is the ``Open'' stage. Crucially, since all queries to $\PV.\Ver$ except $(x^*,c^*,\sigma^*)$ are ignored during the Open stage, we do not have to give the Open stage access to the receiver's Hadamard basis decoding functionalities on the indices that are checked by $D_\inside[P,P(Q(x^*))]$ and $D_\out[P,P(Q(x^*))]$, which are all standard basis positions with respect to the parameters $(\pp_{x^*,c^*,\sigma^*},\sparam_{x^*,c^*,\sigma^*})$.
    \end{enumerate}
    \item This completes the proof, as the previous three bullet points produce a contradiction.
\end{itemize}

Now we provide the formal proof. Suppose there exists $Q, P$ and $\sA_1^{\cPK,\PV.\Ver[\vk]}$ that violates \cref{def:PV-soundness}, where we have dropped the indexing by $\secp$ for convenience. Our first step will be to replace the PRF $F_{k_2}(\cdot)$ with a random oracle $G$. Note that $\sA_1$ only has polynomially-bounded oracle access to this functionality, so this has a negligible affect on the output of $\sA_1$ \cite{6375347}. This defines an oracle algorithm $\sA_2^G$ based on $\sA_1^{\cPK,\PV.\Ver[\vk]}$ that operates as follows.

\begin{itemize}
    \item Sample $(\vk,\ket{\pk},\cPK)$ as in $\PV.\Gen(1^\secp,Q)$, except $F_{k_2}(\cdot)$ is replaced with $G(\cdot)$.
    \item Run $\sA_1^{\cPK,\PV.\Ver[\vk]}(\ket{\pk})$, forwarding calls to $G$ (which occur as part of calls to $\CVGen$ and $\PV.\Ver[\vk]$) to the external random oracle $G$.
    \item Measure $\sA_1$'s output $(x^*,\pi^*)$, parse $\pi^*$ as $(c^*,\sigma^*,\mu^*)$ and output $a \coloneqq (x^*,c^*,\sigma^*)$ and $\aux \coloneqq (\mu^*,\vk)$.
\end{itemize}

\protocol{Functionalities used in the proof of \cref{thm:PV-soundness}}{Description of functionalities used in the proof of \cref{thm:PV-soundness}.}{fig:hybrids}{
\textbf{Fixed parameters}: Security parameter $\secp$, circuit $Q$, and predicate $P$.

\begin{itemize}

    \item $\PV.\Ver[\vk](x,\pi)$: Same as $\PV.\Ver(\vk,x,\pi)$.

    \item $\PV.\Ver[\vk,s](x,\pi)$: Same as $\PV.\Ver[\vk](x,\pi)$ except that $s$ is used instead of $F_{k_2}(x,c,\sigma)$ when generating $(\pp,\sparam) \coloneqq \sV_\Gen^{\CV}(1^\secp,Q;s)$.
    
    \item $\PV.\Ver[\vk,s,(x^*,c^*,\sigma^*)](x,\pi)$: Same as $\PV.\Ver[\vk,s](x,\pi)$, except that after the input is parsed as $x$ and $\pi \coloneqq (c,\sigma,\mu)$, output $\bot$ if \[(x,c,\sigma) \neq (x^*,c^*,\sigma^*).\]
    
    \item $\PV.\Ver[\vk,s,(x^*,c^*,\sigma^*),\mathsf{in}](x,\pi)$: Same as $\PV.\Ver[\vk,s,(x^*,c^*,\sigma^*)](x,\pi)$ except that after $\widetilde{\pi} \coloneqq \{b_{i,j},y_{i,j},z_{i,j}\}_{i \in [r], j \in [\ell]}$ has been computed, output $\bot$ if \[\TestOut[\sparam]\left(\widetilde{\pi}\right) \notin D_{\mathsf{in}}[P,P(Q(x))].\]
    
    \item $\PV.\Ver[\vk,s,(x^*,c^*,\sigma^*),\out](x,\pi)$: Same as $\PV.\Ver[\vk,s,(x^*,c^*,\sigma^*)](x,\pi)$ except that after $\widetilde{\pi} \coloneqq \{b_{i,j},y_{i,j},z_{i,j}\}_{i \in [r], j \in [\ell]}$ has been computed, output $\bot$ if \[\TestOut[\sparam]\left(\widetilde{\pi}\right) \notin D_{\out}[P,P(Q(x))].\]

    \item $V(a,s,\aux)$:
    \begin{itemize}
        \item Parse $a \coloneqq (x^*,c^*,\sigma^*)$ and $\aux \coloneqq (\mu^*,\vk)$.
        \item Compute $q \coloneqq \PV.\Ver[\vk,s](x^*,(c^*,\sigma^*,\mu^*))$.
        \item Output 1 iff $q \neq \bot$ and $\PV.\Out(q,P) = 1 - P(Q(x))$.
    \end{itemize}
    
    \item $V[\out](a,s,\aux)$:
    \begin{itemize}
        \item Parse $a \coloneqq (x^*,c^*,\sigma^*)$ and $\aux \coloneqq (\mu^*,\vk)$.
        \item Compute $q \coloneqq \PV.\Ver[\vk,s,(x^*,c^*,\sigma^*),\out](x^*,(c^*,\sigma^*,\mu^*))$.
        \item Output 1 iff $q \neq \bot$ and $\PV.\Out(q,P) = 1 - P(Q(x))$.
    \end{itemize}

\end{itemize}

}

Note that $\sA_2$ makes $p = \poly(\secp)$ total queries to $G$, since $\sA_1$ makes $\poly(\secp)$ queries. Now, define $V$ as in \cref{fig:hybrids}. Then since $\sA_1$ breaks soundness, \[\Pr\left[V(a,G(a),\aux) = 1 : (a,\aux) \gets \sA_2^G\right] = \nonnegl(\secp).\]

Next, since $p = \poly(\secp)$, by \cref{thm:measure-and-reprogram} there exists an algorithm $\sA_3 \coloneqq \Sim[\sA_2]$ such that

\[\Pr\left[V((x^*,c^*,\sigma^*),s,(\mu^*,\vk)) = 1 : \begin{array}{r} ((x^*,c^*,\sigma^*),\state) \gets \sA_3 \\ s \gets \{0,1\}^\secp \\ (\mu^*,\vk) \gets \sA_3(s,\state)\end{array}\right] = \nonnegl(\secp).\]

Moreover, $\sA_3$ operates as follows.

\begin{itemize}
    \item Sample $G$ as a $2p$-wise independent function and $(i,d) \gets (\{0,\dots,p-1\} \times \{0,1\}) \cup \{(p,0)\}$.
    \item Run $\sA_2$ for $i$ oracle queries, answering each query using the function $G$. 
    \item When $\sA_2$ is about to make its $(i+1)$'th oracle query, measure its query register in the standard basis to obtain $a \coloneqq (x^*,c^*,\sigma^*)$. In the special case that $(i,d) = (p,0)$, just measure (part of) the final output register of $\sA_2$ to obtain $a$.
    \item Receive $s$ externally.
    \item If $d = 0$, answer $\sA_2$'s $(i+1)$'th query with $G$. If $d=1$, answer $\sA_2$'s $(i+1)$'th query instead with $G[(x^*,c^*,\sigma^*) \to s]$.
    \item Run $\sA_2$ until it has made all $p$ queries to $G$. For queries $i+2$ through $p$, answer with $G[(x^*,c^*,\sigma^*) \to s]$.
    \item Measure $\sA_2$'s output $\aux \coloneqq (\mu^*,\vk)$.
\end{itemize}

Recall that $\sA_3$ is internally running $\sA_1$, who expects oracle access to $\sH$, $\CVGen$, $\{\cCK_{i,j}\}_{i,j}$ and $\PV.\Ver[\vk]$. These oracle queries will be answered by $\sA_3$. Next, we define $\sA_4$ to be the same as $\sA_3$, except that after $(x^*,c^*,\sigma^*)$ is measured by $\sA_3$, $\sA_1$'s queries to $\PV.\Ver[\vk]$ are answered instead with $\PV.\Ver[\vk,s,(x^*,c^*,\sigma^*)]$ from \cref{fig:hybrids}. 

\begin{claim}\label{claim:A4}
\[\Pr\left[V((x^*,c^*,\sigma^*),s,(\mu^*,\vk)) = 1 : \begin{array}{r} ((x^*,c^*,\sigma^*),\state) \gets \sA_4 \\ s \gets \{0,1\}^\secp \\ (\mu^*,\vk) \gets \sA_4(s,\state)\end{array}\right] = \nonnegl(\secp).\]
\end{claim}

\begin{proof}
We can condition on $\Tok.\Ver(\vk_\Tok,(x^*,c^*),\sigma^*) = \top$, since otherwise $V$ would output 0. Then, by the strong unforgeability of $\Tok$ (\cref{def:strong-unforgeability}), once $(x^*,c^*,\sigma^*)$ is measured, $\sA_1$ cannot produce any query that has noticeable amplitude on any $(x,c,\sigma)$ such that \[(x,c,\sigma) \neq (x^*,c^*,\sigma^*) ~~ \text{and} ~~ \Tok.\Ver(\vk_\Tok,(x,c),\sigma) = \top.\] But after $(x^*,c^*,\sigma^*)$ is measured and $s$ is sampled, $\PV.\Ver[\vk]$ and $\PV.\Ver[\vk,s,(x^*,c^*,\sigma^*)]$ can only differ on $(x,c,\sigma)$ such that \[(x,c,\sigma) \neq (x^*,c^*,\sigma^*) ~~ \text{and} ~~ \Tok.\Ver(\vk_\Tok,(x,c),\sigma) = \top.\] Thus, since $\sA_1$ only has polynomially-many queries, changing the oracle in this way can only have a negligible affect on the final probability, which completes the proof.
\end{proof}

Next, we claim the following, where $V[\out]$ is defined in \proref{fig:hybrids}.

\begin{claim}\label{claim:hybrid0}
\[\Pr\left[V[\out]((x^*,c^*,\sigma^*),s,(\mu^*,\vk)) = 1 : \begin{array}{r} ((x^*,c^*,\sigma^*),\state) \gets \sA_4 \\ s \gets \{0,1\}^\secp \\ (\mu^*,\vk) \gets \sA_4(s,\state)\end{array}\right] = \nonnegl(\secp).\]
\end{claim}

\begin{proof}
First, if we replace the PRF $F_{k_1}(\cdot)$ with an external random oracle $H$, then the probabilities in \cref{claim:A4} and \cref{claim:hybrid0} remain the same up to a negligible difference \cite{6375347}. Next, note that the only event that differentiates \cref{claim:A4} and \cref{claim:hybrid0} is when $\sA_4$ outputs $(x^*,c^*,\sigma^*,\mu^*)$ such that

\[q \neq \bot ~~ \wedge ~~ \Out_\secp[P](q) = 1 - P(Q(x^*)) ~~ \wedge ~~ \TestOut[\sparam](\widetilde{\pi}) \notin D_{\out}[P,P(Q(x^*))],\] where $(\pp,\sparam) \coloneqq \sV_\Gen^{\CV}(1^\secp,Q;s)$, $\widetilde{\pi} \coloneqq \{b_{i,j},y_{i,j},z_{i,j}\}_{i \in [r], j \in [\ell]}$ is computed during \[\PV.\Ver[\vk,s,(x^*,c^*,\sigma^*)](x^*,(c^*,\sigma^*,\mu^*)),\] and $q \coloneqq \sV_\Ver^\CV(Q,x^*,\sparam,\widetilde{\pi})$. If this event occurs with noticeable probability, there must be some fixed $x^*$ such that it occurs with noticeable probability conditioned on $x^*$. However, this would contradict \cref{lemma:Q-out}. Thus, the difference in probability must be negligible, completing the proof.
\end{proof}

Finally, we will define a sequence of hybrids $\{\cH_\iota\}_{\iota \in [0,p]}$ based on $\sA_4$. Hybrid $\cH_\iota$ is defined as follows.

\begin{itemize}
    \item Run $((x^*,c^*,\sigma^*),\state) \gets \sA_4$.
    \item Sample $s \gets \{0,1\}^\secp$.
    \item Run $(\mu^*,\vk) \gets \sA_4(s,\state)$ with the following difference. Recall that at some point, $\sA_4$ begins using the oracle $G[(x^*,c^*,\sigma^*) \to s]$ while answering $\sA_1$'s queries. For the first $\iota$ times that $\sA_1$ queries $\PV.\Ver[\vk,s,(x^*,c^*,\sigma^*)]$ after this point, respond using the oracle $O_\bot$ that outputs $\bot$ on every input.
    \item Output $V[\out]((x^*,c^*,\sigma^*),s,(\mu^*,\vk))$.
\end{itemize}

Note that \cref{claim:hybrid0} is stating exactly that $\Pr[\cH_0 = 1] = \nonnegl(\secp)$. Next, we have the following claim.

\begin{claim}\label{claim:hybridq}
$\Pr[\cH_p = 1] = \negl(\secp)$.
\end{claim}

\begin{proof}
First, if we replace the PRF $F_{k_1}(\cdot)$ with an external random oracle $H$, then the probability remains the same up to a negligible difference \cite{6375347}. Now, the claim follows by a reduction to the soundness of $\Pi^\CV$ (\cref{lemma:soundness}). Note that $\sA_4$ never needs to know the $\sparam$ such that $(\pp,\sparam) \coloneqq \sV_\Gen^{\CV}(1^\secp,Q;s)$, since all of the (at most $p$) calls that $\sA_1$ makes to $\PV.\Ver[\vk,s,(x^*,c^*,\sigma^*)]$ once $G$ is programmed so that $G[(x^*,c^*,\sigma^*) \to s]$ are answered with $O_\bot$. Thus, we can view $\sA_4^H$ as an adversarial prover for $\Pi^\CV$, where the first stage of $\sA_4^H$ outputs $x^*$, and the second stage receives $\pp$ and outputs $\widetilde{\pi} \coloneqq \{b_{i,j},y_{i,j},z_{i,j}\}_{i,j}$ (which can be computed from $\mu^*$). By the definition of the predicate $V[\out]$, the probability that $\cH_p = 1$ is at most the probability that $\MM_\secp[P](q) = 1-P(Q(x))$, where $q \coloneqq \sV_\Ver^\CV(Q,x^*,\sparam,\widetilde{\pi})$, which by \cref{lemma:soundness} must be $\negl(\secp)$.
\end{proof}

Finally, we prove the following \cref{lemma:intermediate-hybrids}. Since $p = \poly(\secp)$, this contradicts \cref{claim:hybrid0} and \cref{claim:hybridq}, which completes the proof.

\end{proof}

\begin{claim}\label{lemma:intermediate-hybrids}
For any $\iota \in [p]$, $\Pr[\cH_\iota = 1] \geq \Pr[\cH_{\iota - 1} = 1] - \negl(\secp)$.
\end{claim}

\begin{proof}

Throughout this proof, when we refer to ``query $\iota$'' in some hybrid, we mean the $\iota$'th query that $\sA_1$ makes to $\PV.\Ver[\vk,x,(x^*,c^*,\sigma^*)]$ after $\sA_4$ has begun using the oracle $G[(x^*,c^*,\sigma^*) \to s]$ (if such a query exists).

Now, we introduce an intermediate hybrid $\cH_{\iota-1}'$ which is the same as $\cH_{\iota-1}$ except that query $\iota$ is answered with the functionality $\PV.\Ver[\vk,s,(x^*,c^*,\sigma^*),\mathsf{in}]$ defined in \proref{fig:hybrids}.

So, it suffices to show that 
\begin{itemize}
    \item $\Pr[\cH_{\iota-1}' = 1] \geq \Pr[\cH_{\iota-1} = 1] - \negl(\secp)$, and 
    \item $\Pr[\cH_{\iota} = 1] \geq \Pr[\cH_{\iota-1}' = 1] - \negl(\secp)$.
\end{itemize}

We note that the only difference between the three hybrids is how query $\iota$ is answered:
\begin{itemize}
    \item In $\cH_{\iota - 1}$, query $\iota$ is answered with $\PV.\Ver[\vk,s,(x^*,c^*,\sigma^*)]$.
    \item In $\cH_{\iota - 1}'$, query $\iota$ is answered with $\PV.\Ver[\vk,s,(x^*,c^*,\sigma^*),\mathsf{in}]$.
    \item In $\cH_\iota$, query $\iota$ is answered with $O_\bot$.
\end{itemize}

Now, the proof is completed by appealing to the following two claims.
\end{proof}

\begin{claim}
$\Pr[\cH_{\iota-1}' = 1] \geq \Pr[\cH_{\iota-1} = 1] - \negl(\secp).$
\end{claim}

\begin{proof}

First, if we replace the PRF $F_{k_1}(\cdot)$ with an external random oracle $H$, then $\Pr[\cH_{\iota - 1} = 1]$ and $\Pr[\cH_{\iota - 1}' = 1]$ remain the same up to negligible difference \cite{6375347}. Now, this follows from a reduction to \cref{lemma:Q-in}. Indeed, note that if $|\Pr[\cH_{\iota-1}' = 1] - \Pr[\cH_{\iota-1} = 1]| = \nonnegl(\secp)$, then in $\cH_{\iota-1}$, $\sA_1$'s $\iota$'th query must have noticeable amplitude on $(x^*, \pi^* = (c^*,\sigma^*,\mu^*))$ such that 
\[q \neq \bot ~~ \wedge ~~ \TestOut[\sparam](\widetilde{\pi}) \notin D_{\inside}[P,P(Q(x^*))],\] where $(\pp,\sparam) \coloneqq \sV_\Gen^{\CV}(1^\secp,Q;s)$, $\widetilde{\pi} \coloneqq \{b_{i,j},y_{i,j},z_{i,j}\}_{i \in [r], j \in [\ell]}$ is computed during \[\PV.\Ver[\vk,s,(x^*,c^*,\sigma^*)](x^*,(c^*,\sigma^*,\mu^*)),\] and $q \coloneqq \sV_\Ver^\CV(Q,x^*,\sparam,\widetilde{\pi})$. However, $\sA_4$ never needs to know $\sparam$ prior to this query, since all of the calls that $\sA_1$ makes to $\PV.\Ver[\vk,s,(x^*,c^*,\sigma^*)]$ once $G$ is programmed so that $G[(x^*,c^*,\sigma^*) \to s]$ are answered with $O_\bot$. Thus, we can view $\sA_4^H$ has an adversarial prover for $\Pi^\CV$, where the first part of $\sA_4^H$ outputs $x^*$, and the second part receives $\pp$ and outputs $\widetilde{\pi} \coloneqq \{b_{i,j},y_{i,j},z_{i,j}\}_{i,j}$ (which can be computed from $\mu^*$). Then, by \cref{lemma:Q-in}, the above event occurs with negligible probability.

\end{proof}

\begin{claim}
$\Pr[\cH_{\iota} = 1] \geq \Pr[\cH_{\iota-1}' = 1] - \negl(\secp)$
\end{claim}

\begin{proof}

We will show this by reduction to the string binding with public decodability property of $\PFC$. Recall from \cref{subsec:classical-verification} that based on any $(\pp,\sparam) \in \sV_\Gen^\CV(1^\secp,Q)$, we define a subset of indices $S \coloneqq \{S_i\}_{i \in [r]} \subset [r] \times [\ell]$ by the subsets $\{S_i\}_{i \in [r]}$ defined by $\sparam$. This subset $S$ is used in turn to define the predicates $D_\mathsf{in}[P,b]$ and $D_\out[P,b]$. Throughout this proof, we will always let $S$ be defined based on $(\pp,\sparam) \coloneqq \sV_\Gen^\CV(1^\secp,Q;s)$, where the coins $s$ will always be clear from context. We also define $m \coloneqq |S|$, which we assume is the same for all coins $s$.

Now we define an oracle-aided operation $\sC$ as follows.

\begin{itemize}
    \item $\sC$ takes as input $\{\ket{\ck_\tau}\}_{\tau \in [m]}$, where $\{\dk_\tau, \ket{\ck_\tau}, \cCK_\tau \gets \PFC.\Gen(1^\secp)\}_{\tau \in [m]}$.
    \item $\sC$ samples $s \gets \{0,1\}^\secp$ and sets $(\pp,\sparam) \coloneqq \sV^\CV_\Gen(1^\secp,Q;s)$. For $(i,j) \notin S$, sample $\dk_{i,j},\ket{\ck_{i,j}},\cCK_{i,j} \gets \PFC.\Gen(1^\secp)$. Let $f : [m] \to S$ be an arbitrary bijection, and re-define $\{\dk_\tau,\ket{\ck_\tau},\cCK_\tau\}_{\tau \in [m]}$ as $\{\dk_{f(\tau)},\ket{\ck_{f(\tau)}},\cCK_{f(\tau)}\}_{\tau \in [m]}$.
    \item $\sC$ runs $\sA_4$ as defined by $\cH_{\iota-1}'$ until right before query $\iota$ is answered. All queries to $\cCK_{i,j}$, $\PFC.\DecZ[\dk_{i,j}]$, or $\PFC.\DecX[\dk_{i,j}]$ for $(i,j) \in S$ are forwarded to external oracles. 
\end{itemize}

That is, we can write the operation of $\sC$ as

\[\ket{\psi} \gets \sC^{\bCK,\PFC.\DecZ[\bdk],\PFC.\DecX[\bdk]}(\ket{\bck}),\] where $\bCK$ is the collection oracles $\cCK_1,\dots,\cCK_m$, $\ket{\bck} = (\ket{\ck_1},\dots,\ket{\ck_m})$, $\PFC.\DecZ[\bdk]$ is the collection of oracles $\PFC.\DecZ[\dk_1],\allowbreak\dots,\allowbreak\PFC.\DecZ[\dk_m]$, and $\PFC.\DecX[\bdk]$ is the collection of oracles $\PFC.\DecX[\dk_1],\allowbreak\dots,\allowbreak\PFC.\DecX[\dk_m]$.

Next, we define an oracle-aided unitary $\sU$ as follows.

\begin{itemize}
    \item $\sU$ takes as input the state $\ket{\psi}$ output by $\sC$.
    \item It coherently runs the remainder of $\sA_4$ as defined by $\cH'_{\iota-1}$. Any queries to $\cCK_{i,j}$ or $\PFC.\DecZ[\dk_{i,j}]$ for $(i,j) \in S$ are forwarded to external oracles. Note that this portion of $\sA_4$ does not require access to the Hadamard basis decoding oracles $\PFC.\DecX[\dk_{i,j}]$ for $(i,j) \in S$. This follows because for each such $(i,j)$, $h_{i,j} = 0$, which means that $\PV.\Ver[\vk,s,(x^*,c^*,\sigma^*),\mathsf{in}]$ only requires access to the standard basis decoding oracles at these positions.
\end{itemize}

That is, we can write the operation of $\sU$ as 

\[\ket{\psi'} \coloneqq \sU^{\bCK,\PFC.\DecZ[\bdk]}(\ket{\psi}).\]

Now, we give a name to three registers of the space operated on by $\sU$, as follows. 

\begin{itemize}
    \item $\regQ$ is the query register for $\sA_1$'s $\iota$'th query. That is, the state $\ket{\psi}$ contains a superposition over strings $(x,\pi)$ on register $\regQ$.
    \item $\regA$ holds classical information $(\vk,s,x^*,c^*,\sigma^*)$ that has been sampled previously by $\sC$. Thus, the state $\ket{\psi}$ contains a standard basis state on register $\regA$, and $\sU$ is classically controlled on this register. 
    \item $\regV$ is the register that is measured to produce the string $\mu^*$ output at the end of $\sA_4$'s operation. Thus, the state $\ket{\psi'}$ contains a superposition over $\mu^*$ on register $\regV$.
\end{itemize}

We also define $\widetilde{\sU}$ to be identical to $\sU$ except that it runs the remainder of $\sA_4$ as defined by $\cH_{\iota}$. Note that the only difference between $\sU$ and $\widetilde{\sU}$ is how query $\iota$ is answered at the very beginning.

Next, we define the following two projectors.

\begin{align*}
&\Pi_\inside^{\regQ,\regA} \coloneqq \sum_{\substack{\begin{array}{c}(x,\pi),(\vk,s,x^*,c^*,\sigma^*) ~~ \text{s.t.}\\ \PV.\Ver[\vk,s,(x^*,c^*,\sigma^*),\inside](x,\pi) \neq \bot\end{array}}} \ketbra{(x,\pi),(\vk,s,x^*,c^*,\sigma^*)}{(x,\pi),(\vk,s,x^*,c^*,\sigma^*)} \\
&\Pi_\out^{\regA,\regV} \coloneqq \sum_{\substack{\begin{array}{c}(\vk,s,x^*,c^*,\sigma^*),\mu^* ~~ \text{s.t.}\\ V[\out]((x^*,c^*,\sigma^*),s,(\mu^*,\vk)) = 1\end{array}}} \ketbra{(\vk,s,x^*,c^*,\sigma^*),\mu^*}{(\vk,s,x^*,c^*,\sigma^*),\mu^*}
\end{align*}

Now, observe that \[\Pr[\cH_{\iota-1}' = 1] = \E_{\bCK,\bdk,\ket{\bck}}\left[\Big\| \Pi_{\out}^{\regA,\regV}\sU^{\bCK,\PFC.\DecZ[\bdk]}\ket{\psi}\Big\|^2 : \ket{\psi} \gets \sC^{\bCK,\PFC.\DecZ[\bdk],\PFC.\DecX[\bdk]}(\ket{\bck})\right],\] and \[\Pr[\cH_{\iota} = 1] = \E_{\bCK,\bdk,\ket{\bck}}\left[\Big\| \Pi_{\out}^{\regA,\regV}\widetilde{\sU}^{\bCK,\PFC.\DecZ[\bdk]}\ket{\psi}\Big\|^2 : \ket{\psi} \gets \sC^{\bCK,\PFC.\DecZ[\bdk],\PFC.\DecX[\bdk]}(\ket{\bck})\right].\]

Furthermore, for any state $\ket{\psi}$ output by $\sC$, we can write $\ket{\psi} \coloneqq \ket{\psi_{\mathsf{in}}} + \ket{\psi_{\mathsf{in}}^\bot}$, where $\ket{\psi_{\mathsf{in}}} \coloneqq \Pi_{\mathsf{in}}^{\regQ,\regA}\ket{\psi}$. Notice that for any such $\ket{\psi_{\mathsf{in}}^\bot}$, it holds that $\sU\ket{\psi_{\mathsf{in}}^\bot} = \widetilde{\sU}\ket{\psi_{\mathsf{in}}^\bot}$, since query $\iota$ is answered with $\bot$ on both states and $\sU$ and $\widetilde{\sU}$ are otherwise identical. Thus, defining \begin{align*}
    &\Pi_{\out,\sU} \coloneqq \left(\sU^{\bCK,\PFC.\DecZ[\bdk]}\right)^\dagger\Pi_{\out}\left(\sU^{\bCK,\PFC.\DecZ[\bdk]}\right),\\
    &\Pi_{\out,\widetilde{\sU}} \coloneqq \left(\widetilde{\sU}^{\bCK,\PFC.\DecZ[\bdk]}\right)^\dagger\Pi_{\out}\left(\widetilde{\sU}^{\bCK,\PFC.\DecZ[\bdk]}\right),
\end{align*} we have that for any $\ket{\psi} \coloneqq \ket{\psi_{\mathsf{in}}} + \ket{\psi_{\mathsf{in}}^\bot}$, 

\begin{align*}
    &\Big\| \Pi_{\out,\sU}(\ket{\psi_{\mathsf{in}}} + \ket{\psi_{\mathsf{in}}^\bot})\Big\|^2 - \Big\| \Pi_{\out,\widetilde{\sU}}(\ket{\psi_{\mathsf{in}}} + \ket{\psi_{\mathsf{in}}^\bot})\Big\|^2 \\ 
    &~~ = \bra{\psi_{\mathsf{in}}}\Pi_{\out,\sU}\ket{\psi_{\mathsf{in}}} + \bra{\psi_{\mathsf{in}}}\Pi_{\out,\sU}\ket{\psi_{\mathsf{in}}^\bot} + \bra{\psi_{\mathsf{in}}^\bot}\Pi_{\out,\sU}\ket{\psi_{\mathsf{in}}} \\ &~~~~~~~ - \bra{\psi_{\mathsf{in}}}\Pi_{\out,\widetilde{\sU}}\ket{\psi_{\mathsf{in}}} - \bra{\psi_{\mathsf{in}}}\Pi_{\out,\widetilde{\sU}}\ket{\psi_{\mathsf{in}}^\bot} - \bra{\psi_{\mathsf{in}}^\bot}\Pi_{\out,\widetilde{\sU}}\ket{\psi_{\mathsf{in}}} \\
    &~~ \leq 3\Big\| \Pi_{\out,\sU}\ket{\psi_{\mathsf{in}}}\Big\| + 3\Big\| \Pi_{\out,\widetilde{\sU}}\ket{\psi_{\mathsf{in}}}\Big\|.
\end{align*}

So, we can bound $\Pr[\cH_{\iota - 1}' = 1] - \Pr[\cH_{\iota} = 1]$ by


\begin{align*}
    \E_{\bCK,\bdk,\ket{\bck}}\left[3\Big\| \Pi_{\out,\sU}\ket{\psi_{\mathsf{in}}}\Big\| + 3\Big\| \Pi_{\out,\widetilde{\sU}}\ket{\psi_{\mathsf{in}}}\Big\| : \begin{array}{r}\ket{\psi} \gets \sC^{\bCK,\PFC.\DecZ[\bdk],\PFC.\DecX[\bdk]}(\ket{\bck}) \\ \ket{\psi} \coloneqq \ket{\psi_{\mathsf{in}}} + \ket{\psi_{\mathsf{in}}^\bot}\end{array}\right],
\end{align*} 

and thus it suffices to show that 

\[\E_{\bCK,\bdk,\ket{\bck}}\left[\Big\| \Pi_{\out}\sU^{\bCK,\PFC.\DecZ[\bdk]}\Pi_{\mathsf{in}}\ket{\psi}\Big\|^2 : \ket{\psi} \gets \sC^{\bCK,\PFC.\DecZ[\bdk],\PFC.\DecX[\bdk]}(\ket{\bck})\right] = \negl(\secp),\]

and

\[\E_{\bCK,\bdk,\ket{\bck}}\left[\Big\| \Pi_{\out}\widetilde{\sU}^{\bCK,\PFC.\DecZ[\bdk]}\Pi_{\mathsf{in}}\ket{\psi}\Big\|^2 : \ket{\psi} \gets \sC^{\bCK,\PFC.\DecZ[\bdk],\PFC.\DecX[\bdk]}(\ket{\bck})\right] = \negl(\secp).\]

The rest of this proof will be identical in either case, so we consider $\sU$. Towards proving this, we first recall that $s$ is sampled uniformly at random at the very beginning of $\sC$, and the rest of $\sC$ and $\sU$ are classically controlled on $s$. So, let $\sC_s$ be the same as $\sC$ except that it is initialized with the string $s$. Then it suffices to show that for any fixed $s$,

\[\E_{\bCK,\bdk,\ket{\bck}}\left[\Big\| \Pi_{\out}\sU^{\bCK,\PFC.\DecZ[\bdk]}\Pi_{\mathsf{in}}\ket{\psi}\Big\|^2 : \ket{\psi} \gets \sC_s^{\bCK,\PFC.\DecZ[\bdk],\PFC.\DecX[\bdk]}(\ket{\bck})\right] = \negl(\secp).\]

Now, we observe that the register $\regA$ output by $\sC$ contains a standard basis state holding $(\vk,s,(x^*,c^*,\sigma^*))$, where $c^* \coloneqq \{c^*_{i,j}\}_{i \in [r], j \in [\ell]}$. Define commitments $\bc \coloneqq \{c^*_{i,j}\}_{(i,j) \in S}$ and write the output of $\sC_s$ as $(\ket{\psi},\bc)$ to make these commitments explicit. Then, define the following predicates, where $f$ is the bijection from $[m] \to S$ defined earlier.\\

\noindent $\widetilde{D}_{\mathsf{in}}[\bdk,\bc]$:
\begin{itemize}
    \item Take as input $(b,\pi)$, where $\pi$ is parsed as $(\cdot,\cdot,\{u_{i,j},y_{i,j},z_{i,j}\}_{i \in [r], j \in [\ell]})$.
    \item Output 1 if for some $w \in D_{\mathsf{in}}[P,b]$ and all $(i,j) \in S$, $w_{f^{-1}(i,j)} = \PFC.\DecZ(\dk_{i,j},c^*_{i,j},u_{i,j}).$
\end{itemize}

\noindent $\widetilde{D}_\out[\bdk,\bc]$:
\begin{itemize}
    \item Take as input $(b,\mu^*)$, where $\mu^*$ is parsed as $\{u_{i,j},y_{i,j},z_{i,j}\}_{i \in [r], j \in [\ell]}$.
    \item Output 1 if for some $w \in D_\out[P,b]$ and all $(i,j) \in S$, $w_{f^{-1}(i,j)} = \PFC.\DecZ(\dk_{i,j},c^*_{i,j},u_{i,j}).$
\end{itemize}

Next, we define the following two projectors.

\begin{align*}
    &\Pi_{\bdk,\bc,\mathsf{in}}^{\regQ,\regA} \coloneqq \sum_{\substack{\begin{array}{c}(\cdot,\pi),(\cdot,\cdot,x^*,\cdot,\cdot) ~~ \text{s.t.}\\ \widetilde{D}_\inside[\bdk,\bc](P(Q(x^*)),\pi) = 1\end{array}}} \ketbra{(\cdot,\pi),(\cdot,\cdot,x^*,\cdot,\cdot)}{(\cdot,\pi),(\cdot,\cdot,x^*,\cdot,\cdot)} \\
    &\Pi_{\bdk,\bc,\out}^{\regA,\regV} \coloneqq \sum_{\substack{\begin{array}{c}(\cdot,\cdot,x^*,\cdot,\cdot),\mu^* ~~ \text{s.t.}\\ \widetilde{D}_\out[\bdk,\bc](P(Q(x^*)),\mu^*) = 1\end{array}}} \ketbra{(\cdot,\cdot,x^*,\cdot,\cdot),\mu^*}{(\cdot,\cdot,x^*,\cdot,\cdot),\mu^*}
\end{align*}




Note that $\Pi_{\mathsf{in}}^{\regQ,\regA} \leq \Pi_{\bdk,\bc,\mathsf{in}}^{\regQ,\regA}$ and $\Pi_{\out}^{\regA,\regV} \leq \Pi_{\bdk,\bc,\out}^{\regA,\regV}$, and thus it suffices to show that 

\[\E_{\bCK,\bdk,\ket{\bck}}\left[\Big\|\Pi_{\bdk,\bc,\out}^{\regA,\regV} \sU^{\bCK,\PFC.\DecZ[\bdk]} \Pi_{\bdk,\bc,\mathsf{in}}^{\regQ,\regA}\ket{\psi}\Big\|^2 : (\ket{\psi},\bc) \gets \sC_s^{\bCK,\PFC.\DecZ[\bdk],\PFC.\DecX[\bdk]}(\ket{\bck})\right] = \negl(\secp).\]





Finally, for each $b \in \{0,1\}$, we define

\begin{align*}
    \Pi_{\bdk,\bc,\inside,b}^\regQ \coloneqq \sum_{(\cdot,\pi) : \widetilde{D}_{\mathsf{in}}[\bdk,\bc](b,\pi) = 1} \ketbra{(\cdot,\pi)}{(\cdot,\pi)}, ~~ \Pi_{\bdk,\bc,\out,b}^\regV \coloneqq \sum_{\rho^* : \widetilde{D}_\out[\bdk,\bc](b,\mu^*) = 1}\ketbra{\mu^*}{\mu^*}.
\end{align*}

In fact, these projectors now only operate on the sub-registers of $\regQ$ and $\regV$ that hold the strings \[\{u_{i,j}\}_{(i,j) \in S} = \{u_{f(\tau)}\}_{\tau \in [m]}.\] Naming these sub-registers $\regQ' = (\regQ_1,\dots,\regQ_m)$ and $\regV' = (\regV_1,\dots,\regV_m)$, we can write 

\[\Pi_{\bdk,\bc,\mathsf{in},b}^{\regQ'} \coloneqq \sum_{w \in D_\mathsf{in}[P,b]} \left( \bigotimes_{\tau \in [m]} \Pi_{\dk_{\tau},c^*_\tau,w_\tau}^{\regQ_\tau} \right), \ \ \, \Pi_{\bdk,\bc,\out,b}^{\regV'} \coloneqq \sum_{w \in D_\out[P,b]} \left( \bigotimes_{\tau \in [m]} \Pi_{\dk_{\tau},c^*_\tau,w_\tau}^{\regV_\tau} \right),\] where

\[\Pi_{\dk_\tau,c^*_\tau,w_\tau} \coloneqq \sum_{u : \PFC.\DecZ(\dk_\tau,c^*_\tau,u) = w_\tau} \ketbra{u}{u}.\]

Now, to complete the proof, we note that

\begin{align*}
    & ~~\E_{\bCK,\bdk,\ket{\bck}}\left[\Big\|\Pi_{\bdk,\bc,\out} \sU^{\bCK,\PFC.\DecZ[\bdk]} \Pi_{\bdk,\bc,\mathsf{in}}\ket{\psi}\Big\|^2 : (\ket{\psi},\bc) \gets \sC_s^{\bCK,\PFC.\DecZ[\bdk],\PFC.\DecX[\bdk]}(\ket{\bck})\right] \\
    &\leq \E_{\bCK,\bdk,\ket{\bck}}\left[\Big\|\Pi_{\bdk,\bc,\out,0} \sU^{\bCK,\PFC.\DecZ[\bdk]} \Pi_{\bdk,\bc,\mathsf{in},0}\ket{\psi}\Big\|^2 : (\ket{\psi},\bc) \gets \sC_s^{\bCK,\PFC.\DecZ[\bdk],\PFC.\DecX[\bdk]}(\ket{\bck})\right] \\
    & ~~~~ +\E_{\bCK,\bdk,\ket{\bck}}\left[\Big\|\Pi_{\bdk,\bc,\out,1} \sU^{\bCK,\PFC.\DecZ[\bdk]} \Pi_{\bdk,\bc,\mathsf{in},1}\ket{\psi}\Big\|^2 : (\ket{\psi},\bc) \gets \sC_s^{\bCK,\PFC.\DecZ[\bdk],\PFC.\DecX[\bdk]}(\ket{\bck})\right],
\end{align*}


and by the string binding with public decodability of $\PFC$ (\cref{def:predicate-binding}), and the fact that $D_{\mathsf{in}}[P,b]$ and $D_{\out}[P,b]$ are disjoint sets of strings, we have that for any $b \in \{0,1\}$,

\[\E_{\bCK,\bdk,\ket{\bck}}\left[\Big\|\Pi_{\bdk,\bc,\out,b} \sU^{\bCK,\PFC.\DecZ[\bdk]} \Pi_{\bdk,\bc,\mathsf{in},b}\ket{\psi}\Big\| : (\ket{\psi},\bc) \gets \sC_s^{\bCK,\PFC.\DecZ[\bdk],\PFC.\DecX[\bdk]}(\ket{\bck})\right] = \negl(\secp).\] 

\end{proof}

\subsection{Application: Publicly-Verifiable QFHE}

Now, we apply our general framework for verification of quantum partitioning circuits to the specific case of quantum fully-homomorphic encryption (QFHE). First, we define the notion of publicly-verifiable QFHE for pseudo-deterministic circuits. We write the syntax in the \emph{oracle model}, where the parameters used for proving and verifying include an efficient classical oracle $\mathsf{PP}$. Such a scheme can be heuristically instantiated in the plain model by using post-quantum indistinguishability obfuscation to obfuscate this oracle.

\begin{definition}[Publicly-verifiable QFHE for pseudo-deterministic circuits]\label{def:PVQFHE}
A publicly-verifiable quantum fully-homomorphic encryption scheme for pseudo-deterministic circuits consists of the following algorithms $(\Gen,\Enc,\VerGen,\Eval,\Ver,\Dec)$.
\begin{itemize}
    \item $\Gen(1^\lambda,D) \to (\pk,\sk)$: On input the security parameter $1^\secp$ and a circuit depth $D$, the key generation algorithm returns a public key $\pk$ and a secret key $\sk$.
    \item $\Enc(\pk, x) \to \ct$: On input the public key $\pk$ and a classical plaintext $x$, the encryption algorithm outputs a ciphertext $\ct$.
    \item $\VerGen(\ct,Q) \to (\ket{\pp},\mathsf{PP})$: On input a ciphertext $\ct$ and the description of a quantum circuit $Q$, the verification parameter generation algorithm returns public parameters $(\ket{\mathsf{pp}},\mathsf{PP})$, where $\mathsf{PP}$ is the description of a classical deterministic polynomial-time functionality.
    \item $\Eval^{\mathsf{PP}}(\ct, \ket{\mathsf{pp}}, y) \to (\widetilde{\ct},\pi)$: The evaluation algorithm has oracle access to $\mathsf{PP}$, takes as input a ciphertext $\ct$, a quantum state $\ket{\mathsf{pp}}$, and a classical string $y$, and outputs a ciphertext $\widetilde{\ct}$ and proof $\pi$.
    \item $\Ver^{\mathsf{PP}}(y,\widetilde{\ct},\pi) \to \{\top,\bot\}$: The classical verification algorithm has oracle access to $\mathsf{PP}$, takes as input a string $y$, a ciphertext $\widetilde{\ct}$, and a proof $\pi$, and outputs either $\top$ or $\bot$.
    \item $\Dec(\sk, \ct) \to x$: On input the secret key $\sk$ and a classical ciphertext $\ct$, the decryption algorithm returns a message $x$.
\end{itemize}

These algorithms should satisfy the following properties.

\begin{itemize}
    \item \textbf{Correctness.} For any family $\{Q_\secp,x_\secp,y_\secp\}_{\secp \in \bbN}$ where $Q_\secp$ takes two inputs, $\{Q_\secp(x_\secp,\cdot)\}_{\secp \in \bbN}$ is pseudo-deterministic, and $Q_\secp$ has depth $D=D(\secp)$, it holds that 
    \[\Pr\left[\begin{array}{l}\Ver^{\mathsf{PP}}(y,\widetilde{\ct},\pi) = \top ~~ \wedge\\ \Dec(\sk,\widetilde{\ct}) = Q(x,y)\end{array} : \begin{array}{r}(\pk,\sk) \gets \Gen(1^\secp,D) \\ \ct \gets \Enc(\pk,x) \\  (\ket{\mathsf{pp}},\mathsf{PP}) \gets \VerGen(\ct,Q) \\ (\widetilde{\ct},\pi) \gets \Eval^{\mathsf{PP}}(\ct,\ket{\mathsf{pp}},y) \end{array}\right] = 1-\negl(\secp).\]
    \item \textbf{Security.} For any QPT adversary $\{\sA_\secp\}_{\secp \in \bbN}$, depth $D=D(\secp)$, and messages $\{x_{\secp,0},x_{\secp,1}\}_{\secp \in \bbN}$,
    \begin{align*}
    &\bigg| \Pr\left[\sA(\pk,\ct) = 1 : \begin{array}{r}(\pk,\sk) \gets \Gen(1^\secp,D) \\ \ct \gets \Enc(\pk,x_0)  \end{array}\right]\\ &- \Pr\left[\sA(\pk,\ct) = 1 : \begin{array}{r}(\pk,\sk) \gets \Gen(1^\secp,D) \\ \ct \gets \Enc(\pk,x_1)\end{array}\right]\bigg| = \negl(\secp)
    \end{align*}
    \item \textbf{Soundness.} For any QPT adversary $\{\sA_\secp\}_{\secp \in \bbN}$, depth $D=D(\secp)$, and family $\{Q_\secp,x_\secp\}_{\secp \in \bbN}$, where $Q_\secp$ takes two inputs and $\{Q_\secp(x_\secp,\cdot)\}_{\secp \in \bbN}$ is pseudo-deterministic, 
    \[\Pr\left[\begin{array}{l} \Ver^{\mathsf{PP}}(y,\widetilde{\ct},\pi) = \top ~~ \wedge \\ \Dec(\sk,\widetilde{\ct}) \neq Q(x,y)\end{array} : \begin{array}{r}(\pk,\sk) \gets \Gen(1^\secp,D) \\ \ct \gets \Enc(\pk,x) \\ (\ket{\pp},\mathsf{PP}) \gets \VerGen(\ct,Q) \\  (y,\widetilde{\ct},\pi) \gets \sA^{\mathsf{PP}}(\ct,\ket{\mathsf{pp}})\end{array}\right] = \negl(\secp).\] 
\end{itemize}

\end{definition}

We will now construct publicly-verifiable QFHE for pseudo-deterministic circuits from the following ingredients.

\begin{itemize}
    \item A quantum fully-homomorphic encryption scheme $(\QFHE.\Gen, \QFHE.\Enc, \QFHE.\Eval, \QFHE.\Dec)$ (\cref{subsec:QFHE}).
    \item A protocol for publicly-verifiable non-interactive classical verification of quantum partitioning circuits in the oracle model $(\PV.\Gen,\PV.\Prove,\PV.\Verify,\PV.\Combine)$ (\cref{subsec:public-verification}).
\end{itemize}

Our construction goes as follows.

\begin{itemize}

    \item $\PV\QFHE.\Gen(1^\secp,D)$: Same as $\QFHE.\Gen(1^\secp,D)$.
    \item $\PV\QFHE.\Enc(\pk,x)$: Same as $\QFHE.\Enc(\pk,x)$.
    \item $\PV\QFHE.\VerGen(\ct,Q)$: 
        \begin{itemize}
            \item Define the quantum circuit $E[\ct]: y \to \QFHE.\Eval(Q(\cdot,y),\ct)$.
            \item Sample $(\PV.\vk,\ket{\PV.\pk},\PV.\mathsf{PK}) \gets \PV.\Gen(1^\secp,E[\ct])$. 
            \item Let $\mathsf{VK}(y,\pi)$ be the following classical functionality. First, run $\PV.\Ver(\PV.\vk,y,\pi)$. Output $\bot$ if the output was $\bot$. Otherwise, parse the output as $(\ct_1,\dots,\ct_m)$, compute $\widetilde{\ct} \coloneqq \QFHE.\Eval(\PV.\Combine,(\ct_1,\dots,\ct_m))$, and output $\widetilde{\ct}$.\footnote{Here, we are using the fact that $\QFHE.\Eval$ is a deterministic classical functionality when evaluating a deterministic classical functionality.} 
            \item Output $\ket{\mathsf{pp}} \coloneqq \ket{\PV.\pk}, \mathsf{PP} \coloneqq \left(\PV.\mathsf{PK},\mathsf{VK}\right).$
        \end{itemize}
    \item $\PV\QFHE.\Eval^{\mathsf{PP}}(\ct,\ket{\mathsf{pp}},y)$: 
    \begin{itemize}
        \item Run $\pi \gets \PV.\Prove^{\PV.\mathsf{PK}}(\ket{\mathsf{pp}},E[\ct],y)$.
        \item Compute $\widetilde{\ct} = \mathsf{VK}(y,\pi)$, and output $(\widetilde{\ct},\pi)$.
    \end{itemize}
    \item $\PV\QFHE.\Ver^{\mathsf{PP}}(y,\widetilde{\ct},\pi)$: Output $\top$ iff $\mathsf{VK}(y,\pi) = \widetilde{\ct}$.
    \item $\PV\QFHE.\Dec(\sk,\ct)$: Same as $\QFHE.\Dec(\sk,\ct)$.

\end{itemize}

\begin{theorem}
The scheme described above satisfies \cref{def:PVQFHE}.
\end{theorem}

\begin{proof}
Correctness follows immediately from the evaluation correctness of QFHE (\cref{def:eval correct}) and the completeness of $\PV$ (\cref{def:PV-completeness}). Security follows immediately from the semantic security of QFHE (\cref{def:semantic-security}). Soundness follows immediately from the correctness of QFHE (\cref{def:eval correct}) and soundness  of $\PV$ (\cref{def:PV-soundness}), since $\QFHE.\Dec(\sk,\cdot) \circ E[\ct]$ is pseudo-deterministic and the $\mathsf{VK}$ oracle is nothing but the $\PV.\Ver[\vk]$ oracle plus post-processing. 
\end{proof}


\section{Quantum Obfuscation}\label{sec:obfuscation}

\subsection{Construction}

In this section, we construct virtual black-box (VBB) obfuscation for pseudo-deterministic quantum circuits from the following ingredients.

\begin{itemize}
    \item A VBB obfuscator $(\CObf,\CEval)$ for classical circuits (\cref{def:ideal-obfuscation}).
    \item A publicly-verifiable QFHE for pseudo-deterministic circuits in the oracle model $(\Gen,\allowbreak\Enc,\allowbreak\VerGen,\allowbreak\Eval,\allowbreak\Ver,\allowbreak\Dec)$ (\cref{def:PVQFHE}).
\end{itemize}

The construction is given in \cref{fig:QObf2}.

\protocol{Obfuscation scheme $(\QObf,\QEval)$ for pseudo-deterministic quantum circuits
}{Obfuscation for pseudo-deterministic quantum circuits.}{fig:QObf2}{
\begin{itemize}
    \item $\QObf(1^\secp,Q)$:
    \begin{itemize}
        \item Let $U$ be the universal quantum circuit that takes as input the description of a circuit of size $|Q|$ and an input of size $n$, where $n$ is the length of an input to $Q$. Let $D$ be the depth of $U$.
        \item Sample $(\pk,\sk) \gets \Gen(1^\secp,D)$, $\ct \gets \Enc(\pk,Q)$, and $(\ket{\pp},\mathsf{PP}) \gets \VerGen(\ct,U)$.
        \item Let $\mathsf{DK}(x,\widetilde{\ct},\pi)$ be the following functionality. First, run $\Ver^{\mathsf{PP}}(x,\widetilde{\ct},\pi)$. If the output was $\bot$, then output $\bot$, and otherwise output $\Dec(\sk,\widetilde{\ct})$.
        \item Sample $\widetilde{\mathsf{PP}} \gets \CObf(1^\secp,\mathsf{PP})$ and $\widetilde{\mathsf{DK}} \gets \CObf(1^\secp,\mathsf{DK})$.
        \item Output $\widetilde{Q} \coloneqq \left(\ct,\ket{\pp},\widetilde{\mathsf{PP}},\widetilde{\mathsf{DK}}\right)$. 
    \end{itemize}
    \item $\QEval(\widetilde{Q},x)$:
    \begin{itemize}
        \item Parse $\widetilde{Q}$ as $\left(\ct,\ket{\pp},\widetilde{\mathsf{PP}},\widetilde{\mathsf{DK}}\right)$.
        \item Compute $(\widetilde{\ct},\pi) \gets \Eval^{\widetilde{\mathsf{PP}}}(\ct,\ket{\pp},x)$.
        \item Output $b \coloneqq \widetilde{\mathsf{DK}}(x,\widetilde{\ct},\pi)$.
    \end{itemize}
\end{itemize}
}

\begin{theorem}\label{thm:obfuscation}
$(\QObf,\QEval)$ described in \cref{fig:QObf2} is a virtual black-box obfuscator for pseudo-deterministic quantum circuits, satisfying \cref{def:ideal-obfuscation}.
\end{theorem}

\begin{proof}
First, correctness follows immediately from the correctness of the VBB obfuscator (\cref{def:ideal-obfuscation}) and the correctness of the publicly-verifiable QFHE scheme (\cref{def:PVQFHE}). Note that even though the evaluation procedure may include measurements, an evaluator could run coherently, measure just the output bit $b$, and reverse. By Gentle Measurement (\cref{lemma:gentle-measurement}), this implies the ability to run the obfuscated program on any $\poly(\secp)$ number of inputs.

Next, we show security. For any QPT adversary $\{\sA_\secp\}_{\secp \in \bbN}$, we define a simulator $\{\sS_\secp\}_{\secp \in \bbN}$ as follows, where $\{\widetilde{\sA}_\secp\}_{\secp \in \bbN}$ is the simulator for the classical obfuscation scheme $(\CObf,\CEval)$, defined based on $\{\sA_\secp\}_{\secp \in \bbN}$.
\begin{itemize}
    \item Sample $(\pk,\sk) \gets \Gen(1^\secp,D)$, $\ct \gets \Enc(\pk,0^{|Q|})$, and $(\ket{\pp},\mathsf{PP}) \gets \VerGen(\ct,U)$.
    \item Run $\widetilde{\sA}_\secp^{\mathsf{PP},\mathsf{DK}}(\ct,\ket{\pp})$, answering $\mathsf{PP}$ calls honestly, and $\mathsf{DK}$ calls as follows. 
    \begin{itemize}
        \item Take $(x,\widetilde{\ct},\pi)$ as input.
        \item Run $\Ver^{\mathsf{PP}}(x,\widetilde{\ct},\pi)$. If the output was $\bot$ then output $\bot$.
        \item Otherwise, forward $x$ to the external oracle $O[Q]$, and return the result $b = O[Q](x)$.
    \end{itemize}
    \item Output $\widetilde{\sA}_\secp$'s output.
\end{itemize}

Now, for any circuit $Q$, we define a sequence of hybrids.

\begin{itemize}
    \item $\cH_0$: Sample $\widetilde{Q} \gets \QObf(1^\secp,Q)$ and run $\sA_\secp(1^\secp,\widetilde{Q})$.
    \item $\cH_1$: Sample $(\ct,\ket{\pp},\mathsf{PP},\mathsf{DK})$ as in $\QObf(1^\secp,Q)$, and run $\widetilde{\sA}_\secp^{\mathsf{PP},\mathsf{DK}}(\ct,\ket{\pp})$.
    \item $\cH_2$: Same as $\cH_1$, except that calls to $\mathsf{DK}$ are answered as in the description of $\sS_\secp$.
    \item $\cH_3$: Same as $\cH_2$, except that we sample $(\ct,\ket{\pp},\mathsf{PP}) \gets \Enc(\pk,0^{|Q|},U)$. This is $\sS_\secp$.
\end{itemize}

We complete the proof by showing the following.

\begin{itemize}
    \item $|\Pr[\cH_0 = 1] - \Pr[\cH_1 = 1]| = \negl(\secp)$. This follows from the security of the classical obfuscation scheme $(\CObf,\CEval)$.
    \item $|\Pr[\cH_1 = 1] - \Pr[\cH_2 = 1]| = \negl(\secp)$. Suppose otherwise. Then there must exist some query made by $\widetilde{\sA}_\secp$ to $\mathsf{DK}$ with noticeable amplitude on $(x,\widetilde{\ct},\pi)$ such that $\mathsf{DK}$ does not return $\bot$ but $\Dec(\sk,\widetilde{\ct}) \neq Q(x)$. Thus, we can measure a random one of the $\poly(\secp)$ many queries made by $\widetilde{\sA}_\secp$ to obtain such an $(x,\widetilde{\ct},\pi)$, which violates the soundness of the publicly-verifiable QFHE scheme (\cref{def:PVQFHE}).
    \item $|\Pr[\cH_2 = 1] - \Pr[\cH_3 = 1]| = \negl(\secp)$. Since $\sk$ is no longer used in $\cH_2$ to respond to $\mathsf{DK}$ queries, this follows directly from the security of the publicly-verifiable QFHE scheme (\cref{def:PVQFHE}).
\end{itemize}

\end{proof}

\subsection{Application: Copy-protection}\label{subsec:copy-protection}

We sketch an application of our obfuscation scheme to copy-protection of \emph{quantum} programs. Let $(\QObf,\QEval)$ be a VBB obfuscation scheme for pseudo-deterministic quantum circuits, and let $F_k$ be a pseudo-random function secure against superposition-query attacks. In~\cref{fig:copy-protection}, we describe \cite{C:ALLZZ21}'s construction of a software copy-protection scheme, generalized to copy-protect pseudo-deterministic quantum circuits.

\protocol{Quantum copy-protection scheme \cite{C:ALLZZ21}}{A description of the quantum copy protection scheme from \cite{C:ALLZZ21}, where the Generate algorithm may now take as input the description of a pseudo-deterministic quantum functionality.}{fig:copy-protection}{
\begin{itemize}
    \item $\mathsf{Setup}(1^\secp) \to \sk$: 
    \begin{itemize}
        \item Take as input the security parameter $1^\secp$.
        \item Sample a uniformly random subspace $S < \bbF_2^\secp$ of dimension $\secp/2$.
        \item Sample a PRF key $k \gets \{0,1\}^\secp$.
        \item Out $\sk \coloneqq (S,k)$.
    \end{itemize}
    \item $\mathsf{Generate}(\sk,Q) \to \widehat{Q}$: 
    \begin{itemize}
        \item Take as input $\sk = (S,k)$ and the description of a pseudo-deterministic quantum circuit $Q$.
        \item Let $O_1$ be the functionality that takes $(x,v)$ as input and outputs $Q(x) \oplus F_{k}(x)$ if $v \in S \setminus \{0\}$, and $\bot$ otherwise.
        \item Let $O_2$ be the functionality that takes $(x,v)$ as input and outputs $F_{k}(x)$ if $v \in S^\bot \setminus \{0\}$, and $\bot$ otherwise.
        \item Sample $\widetilde{O}_1 \gets \QObf(1^\secp,O_1)$ and $\widetilde{O}_2 \gets \QObf(1^\secp,O_2)$
        \item Output $\widehat{Q} \coloneqq \left(\ket{S},\widetilde{O}_1,\widetilde{O}_2\right)$.
    \end{itemize}
    \item $\mathsf{Compute}(\widehat{Q},x) \to y$:
    \begin{itemize}
        \item Parse $\widehat{Q}$ as $\ket{S},\widetilde{O}_1,\widetilde{O}_2$, where $\ket{S}$ is on register $\regS$.
        \item Apply $\QEval(\widetilde{O}_1,\cdot)$ coherently to register $\regS$, measure the output to obtain $y_1$, and reverse the computation of $\QEval(\widetilde{O}_1,\cdot)$.
        \item Apply $H^{\otimes \secp}$ to register $\regS$, apply $\QEval(\widetilde{O}_2,\cdot)$ coherently to register $\regS$, measure the output to obtain $y_2$, reverse the computation of $\QEval(\widetilde{O}_2,\cdot)$, and finally apply $H^{\otimes \secp}$ to register $\regS$.
        \item Output $y \coloneqq y_1 \oplus y_2$.
    \end{itemize}

\end{itemize}
}

We refer the reader to \cite{C:ALLZZ21} for definitions of (generalized) quantum unlearnable function families and anti-piracy of quantum copy-protection schemes. Here, we observe that if $Q$ is a pseudo-deterministic circuit, then both $O_1$ and $O_2$ are as well, and thus they can be obfuscated by our scheme. Finally, it is straightforward to see that any classical functionality $f$ sampled from a distribution $\cF$ can be replaced with a pseudo-deterministic quantum functionality $Q$ sampled from a distribution $\cQ$ in the definitions and proofs from \cite{C:ALLZZ21}. Thus, we can generalize their main theorem as follows.

\begin{theorem}
(Corollary of \cite[Theorem 4]{C:ALLZZ21} and \cref{thm:obfuscation}) Let $\cQ$ be a family of pseudo-deterministic quantum ciruits that is $\gamma$-quantum-unlearnable with respect to distribution $\cD$ (where $\gamma$ is a non-negligible function of $\secp$). Then \proref{fig:copy-protection} is a copy protection scheme for $\cQ,\cD$ that has $(\gamma(\secp)-1/\poly(\secp))$-anti-piracy security, for any polynomial $\poly(\secp)$.
\end{theorem}

\subsection{Application: Functional encryption}\label{subsec:FE}

We sketch an application of our obfuscation scheme to functional encryption for pseudo-deterministic quantum functionalities. Let $(\QObf,\QEval)$ be a VBB obfuscation scheme for pseudo-deterministic quantum circuits,\footnote{For this application, we technically only require the weaker notion of indistinguishability obfuscation (\cref{def:iO}).} let $(\Gen,\Enc,\Dec)$ be a (post-quantum) public-key encryption scheme, and let $(\Setup,\Prove,\Verify)$ be a (post-quantum) statistically simulation sound non-interactive zero-knowledge proof system (SSS-NIZK). We refer the reader to \cite{SIAMCOMP:GGHRSW16} for preliminaries on SSS-NIZK, and for definitions of functional encryption.

Consider the following construction of functional encryption for pseudo-deterministic quantum functionalities.

\begin{itemize}
    \item $\FE.\Setup(1^\secp)$: Sample $(\pk_1,\sk_1) \gets \Gen(1^\secp)$, $(\pk_2,\sk_2) \gets \Gen(1^\secp)$, $\crs \gets \Setup(1^\secp)$, and output $\pp \coloneqq (\pk_1,\pk_2,\crs)$ and $\msk \coloneqq \sk_1$.
    \item $\FE.\KeyGen(\msk,Q)$: On input the master secret key $\msk$ and the description of a pseudo-deterministic quantum circuit $Q$, define the following pseudo-deterministic quantum circuit $C[Q,\crs,\sk_1]$.
    \begin{itemize}
        \item Take $(\ct_1,\ct_2,\pi)$ as input.
        \item Check that $\pi$ is a valid SSS-NIZK proof under $\crs$ that there exists $(m,r_1,r_2)$ such that $\ct_1 = \Enc(\pk_1,m;r_1)$ and $\ct_2 = \Enc(\pk_2,m;r_2)$.
        \item If so, output $Q(\Dec(\sk_1,\ct_1))$, and otherwise output $\bot$.
    \end{itemize}
    Finally, sample and output $\sk_Q \gets \QObf(1^\secp,C[Q,\crs,\sk_1])$.
    \item $\FE.\Enc(\pp,m)$: Sample $r_1,r_2 \gets \{0,1\}^\secp$, compute $\ct_1 \coloneqq \Enc(\pk_1,m;r_1)$, $\ct_2 \coloneqq \Enc(\pk_2,m;r_2)$, compute a SSS-NIZK proof $\pi$ that there exists $(m,r_1,r_2)$ such that $\ct_1 = \Enc(\pk_1,m;r_1)$ and $\ct_2 = \Enc(\pk_2,m;r_2)$, and output $\ct \coloneqq (\ct_1,\ct_2,\pi)$.
    \item $\FE.\Dec(\sk_Q,\ct)$: Run the obfuscated program $\sk_Q$ on input $\ct$ to obtain the output.
\end{itemize}

It is straightforward to extend the definitions and proofs in Section 6 of \cite{SIAMCOMP:GGHRSW16} to consider functional encryption and obfuscation of pseudo-deterministic quantum circuits. As a result, we obtain the following theorem.

\begin{theorem}[Corollary of \cite{SIAMCOMP:GGHRSW16} Section 6 and \cref{thm:obfuscation}] 
The above construction is a functional encryption scheme satisfying indistinguishability security for the class of polynomial-size pseudo-deterministic quantum functionalities.
\end{theorem}

\ifsubmission
\bibliographystyle{splncs04}
\else
\bibliographystyle{alpha}
\fi

\bibliography{abbrev3,custom,crypto}

\appendix

\section{Remaining Proofs from \cref{subsec:com-def}}\label{sec:appendix-string}

\begin{lemma}
Any Pauli functional commitment that satisfies \emph{single-bit binding with public decodability} also satisfies \emph{string binding with public decodability}.
\end{lemma}

\begin{proof}

For this proof, we will need a couple of different binding definitions, as well as a couple of imported theorems.

\begin{definition}[Collapse binding]
A Pauli functional commitment $(\Gen,\allowbreak\Com,\allowbreak\OpenZ,\allowbreak\OpenX,\allowbreak\DecZ,\allowbreak\DecX)$ satisfies \emph{collapse binding} if the following holds. For any adversary $\sA \coloneqq \{(\sC_\secp,\sU_\secp)\}_{\secp \in \bbN}$, where each of $\sC_\secp$ and $\sU_\secp$ are oracle-aided quantum operations that make at most $\poly(\secp)$ oracle queries, define the experiment $\Exp^\sA_\CB(\secp)$ as follows.
\begin{itemize}
    \item Sample $\dk,\ket{\ck},\cCK \gets \Gen(1^\secp)$.
    \item Run $\sC_\secp^{\cCK,\DecZ[\dk],\DecX[\dk]}(\ket{\ck})$ until it outputs a commitment $c$ and a state on registers $(\regB,\regU,\regA)$.
    \item Sample $b \gets \{0,1\}$. If $b = 0$, do nothing, and otherwise measure $(\regB,\regU)$ with $\{\Pi_{\dk,c,0},\Pi_{\dk,c,1}\}$.\footnote{These projectors are defined in 
    \cref{def:poc-single-bit-binding}.}
    \item Run $\sU_\secp^{\cCK,\DecZ[\dk]}(\regB,\regU,\regA)$ until it outputs a bit $b'$. The experiment outputs 1 if $b = b'$.
\end{itemize}

We say that $\sA$ is \emph{valid} if the state on $(\regB,\regU)$ output by $\sC_\secp$ is in the image of $\Pi_{\dk,c,0} + \Pi_{\dk,c,1}$. Then, it must hold that for all valid adversaries $\sA$, 
\[\left|\Pr\left[\Exp_\CB^\sA(\secp) = 1\right]-\frac{1}{2}\right| = \negl(\secp).\]

\end{definition}

\begin{definition}[Unique message binding]
A Pauli functional commitment $(\Gen,\allowbreak\Com,\allowbreak\OpenZ,\allowbreak\OpenX,\allowbreak\DecZ,\allowbreak\DecX)$ satisfies \emph{unique message binding} if for any polynomial $m(\secp)$ and any adversary $\{(\sC_\secp,\sU_\secp)\}_{\secp \in \bbN}$, where each of $\sC_\secp$ and $\sU_\secp$ are oracle-aided quantum operations that make at most $\poly(\secp)$ oracle queries, the following experiment outputs 1 with probability $\negl(\secp)$.
\begin{itemize}
    \item Sample $\{\dk_i,\ket{\ck_i},\cCK_i \gets \Gen(1^\secp)\}_{i \in [m]}$.
    \item Run $\sC_\secp^{\bCK,\DecZ[\bdk],\DecX[\bdk]}(\ket{\bck})$ until it outputs a commitment $\bc \coloneqq (c_1,\dots,c_m)$, a message $x_1 \in \{0,1\}^m$, and a state on registers $(\regB_1,\regU_1,\dots,\regB_m,\regU_m,\regA)$.
    \item For each $i \in [m]$, apply $\Pi_{\dk_i,c_i,x_{1,i}}$ to $(\regB_i,\regU_i)$ and abort and output 0 if this projection rejects.
    \item Run $\sU_\secp^{\bCK,\DecZ[\bdk]}(\regB_1,\regU_1,\dots,\regB_m,\regU_m,\regA)$ until it outputs a message $x_2 \in \{0,1\}^m$, and a state on registers $(\regB_1,\regU_1,\dots,\regB_m,\regU_m)$. If $x_1 = x_2$, abort and output 0.
    \item For each $i \in [m]$, apply $\Pi_{\dk_i,c_i,x_{2,i}}$ to $(\regB_i,\regU_i)$ and abort and output 0 if this projection rejects. Otherwise, output 1.
\end{itemize}
\end{definition}

\begin{importedtheorem}[\cite{9996877}]\label{impthm:LMS}
Any commitment that satisfies collapse binding also satisfies unique message binding.
\end{importedtheorem}

\begin{importedtheorem}[\cite{cryptoeprint:2022/786}]\label{impthm:DS}
Let $\sD$ be a projector, $\Pi_0,\Pi_1$ be orthogonal projectors, and $\ket{\psi} \in \mathsf{Im}\left(\Pi_0+\Pi_1\right)$. Then,

\[\|\Pi_1\sD\Pi_0\ket{\psi}\|^2 + \|\Pi_0\sD\Pi_1\ket{\psi}\|^2 \geq \frac{1}{2}\left(\|\sD\ket{\psi}\|^2 - \left(\|\sD\Pi_0\ket{\psi}\|^2 + \|\sD\Pi_1\ket{\psi}\|^2 \right)\right)^2.\]
\end{importedtheorem}

Given these imported theorems, the proof of our lemma is quite straightforward.

\begin{itemize}
    \item First, we establish using \cref{impthm:DS} that any Pauli functional commitment that satisfies single-bit binding also satisfies collapse binding. To see this, suppose there exists an adversary $(\sC,\sU)$ that breaks collapse binding, let $\Pi_0 = \Pi_{\dk,c,0}$, $\Pi_1 = \Pi_{\dk,c,1}$, let $\sD$ be a projective implementation of $\sU^{\cCK,\DecZ[\dk]}$, and let $\ket{\psi}$ be the state of the collapse binding experiment that is output by $\sC^{\cCK,\DecZ[\dk],\DecX[\dk]}$. Then the RHS of \cref{impthm:DS} is half the squared advantage of the adversary in the collapse binding game. This implies that at least one of the terms on the LHS is non-negligible, which immediately implies that this adversary can be used to break the single-bit binding game.
    \item Next, appealing to \cref{impthm:LMS}, we see that any Pauli functional commitment that satisfies single-bit binding also satisfies unique message binding.
    \item Finally, suppose there is a Pauli functional commitment that is single-bit binding, but there exists an adversary that breaks the string binding of this commitment for some pair of disjoint sets $W_0,W_1$. We define an experiment where we insert a measurement of $\DecZ(\bdk,\bc,\cdot)$ applied to the state $\Pi_{\bdk,\bc,W_0}\ket{\psi}$, which by definition will return some string $x_0 \in W_0$. By the collapse binding of the commitment, inserting this measurement will only have a negligible affect on the experiment. But now, since $W_0$ and $W_1$ are disjoint sets, this adversary breaks the unique message binding of the commitment. This completes the proof.
\end{itemize}

\end{proof}
\section{Remaining Proofs from \cref{subsec:binding-proof}}\label{appendix:inner-product}

In this appendix, we prove the following theorem.

\begin{theorem}\label{thm:inner-product}
Let $n,m,d \in \bbN, \epsilon \in (0,1/8)$ be such that $d \geq 2$ and $n-d+1 > 10\log(1/\epsilon)+6$. Let $\sU^{\regX,\regY}$ be any $(2^{n+m})$-dimensional unitary, where register $\regX$ is $2^n$ dimensions and register $\regY$ is $2^m$ dimensions. Let $\cA$ be the set of $d$-dimensional balanced affine subspaces $A = A_0 \cup A_1$ of $\bbF_2^n$, where $A_0$ is the affine subspace of vectors in $A$ that start with 0 and $A_1$ is the affine subspace of vectors in $A$ that start with 1. For any $A = A_0 \cup A_1$, let 

\[\Pi_{A_0} \coloneqq \sum_{v \in A_0}\ketbra{v}{v}^{\regX} \otimes \bbI^{\regY}, ~~~ \Pi_{A_1} \coloneqq \sU^\dagger\left(\sum_{v \in A_1}\ketbra{v}{v}^{\regX} \otimes \bbI^{\regY}\right)\sU.\]

Let $\cR$ be the set of pairs $(A,B)$ of $d$-dimensional affine subspaces of $\bbF_2^n$ such that $\dim(A_0 \cap B_0) = d-2$ and $\dim(A_1 \cap B_1) = d-2$. Then for any set of states $\{\ket{\psi_A}\}_A$ such that for all $A \in \cA$, $\ket{\psi_A} \in \mathsf{Im}(\Pi_{A_0})$, and $\|\Pi_{A_1}\ket{\psi_A}\| \geq \epsilon$, 
\[\E_{(A,B) \gets \cR}[|\braket{\psi_A | \psi_B}|] < \frac{1}{2}-\epsilon^{13}.\]
\end{theorem}

We will first simplify the problem by reducing to the case where each $A$ is two-dimensional, consisting of just four vectors. This case is proven later in \cref{subsec:two-dim}. In the reduction, which follows below, we begin with the observation that each $(A,B) \in \cR$ consists of six cosets of a particular $(d-2)$-dimensional subspace $S$. Then, we partition $\cR$ based on this underlying subspace, and prove the claim separately for each $S$. Finally, the process of sampling $(A,B)$ from $\cR$ conditioned on an underlying subspace $S$ can be seen as sampling $A$ and $B$ as two-dimensional spaces in the subspace of cosets of $S$.

\begin{proof}(of \cref{thm:inner-product})
First, note that for any $(A,B) \in \cR$, $A_0 \cap B_0$ is an intersection of affine subspaces, so is an affine subspace itself. So, we write $A_0 \cap B_0 = S+v_0$ for some $(d-2)$-dimensional subspace $S$. Since all vectors in $S+v_0$ start with 0, it must be the case that all vectors in $S$ start with 0 and $v_0$ starts with 0. Moreover, $A = A_0 \cup A_1$ and $B = B_0 \cup B_1$ are both cosets of superspaces of $S$, and thus we can write 

\[A = (S + v_0) \cup (S + w_0) \cup (S + v_1) \cup (S + w_1), ~~ B = (S + v_0) \cup (S+ u_0) \cup (S + v_1) \cup (S + u_1)\] for $v_0, w_0, u_0$ that start with 0, $v_1, w_1, u_1$ that start with 1, and where $v_0 + w_0 = v_1 + w_1$ and $v_0 + u_0 = v_1 + u_1$.

Now, for any $(d-2)$-dimensional subspace $S \coloneqq \mathsf{span}(z_1,\dots,z_{d-2})$ such all vectors in $S$ start with 0, let $z_{d-1},\dots,z_n$ be such that $(z_1,\dots,z_n)$ is an orthonormal basis of $\bbF_2^n$ and $z_{d-1}$ is the only basis vector that starts with 1. Define the subspace $\co(S) \coloneqq \mathsf{span}(z_{d-1},\dots,z_n)$. Furthermore, let $\co(S)_0$ be the subspace of vectors in $\co(S)$ that start with 0, and let $\co(S)_1$ be the affine subspace of vectors in $\co(S)$ that starts with 1.



Then we can sample from $\cR$ by first sampling a random $(d-2)$-dimensional subspace $S$ such that all vectors in $S$ start with 0, then sampling distinct $v_0,w_0,u_0 \gets \co(S)_0$ and distinct $v_1,w_1,u_1 \gets \co(S)_1$ such that $v_0 + w_0 = v_1 + w_1$ and $v_0 + u_0 = v_1 + u_1$, and finally setting \[A = (S + v_0) \cup (S + w_0) \cup (S + v_1) \cup (S + w_1), ~~ B = (S + v_0) \cup (S+ u_0) \cup (S + v_1) \cup (S + u_1)\]

For any subspace $S$, let $\cR[S]$ be the set of $(A,B) \in \cR$ such that $A_0 \cap B_0$ is a coset of $S$. Thus, it suffices to prove that for \emph{each fixed} $S$, \[\E_{(A,B) \gets \cR[S]}[|\braket{\psi_A | \psi_B}|] < \frac{1}{2}-\epsilon^{13}.\]

Now consider any fixed $S$. For each $A$ that could be sampled by $\cR[S]$, we write \[A = (S+v_0) \cup (S+w_0) \cup (S+v_1) \cup (S+w_1)\] for $v_0,w_0 \in \co(S)_0$ and $v_1,w_1 \in \co(S)_1$ such that $v_0 + w_0 = v_1 + w_1$. Moreover, we can express $v_0,w_0$ as $(0,v_0'),(0,w_0') \in \bbF_2^{n-d+2}$ and $v_1,w_1$ as $(1,v_1'),(1,w_1') \in \bbF_2^{n-d+2}$  in the $(z_{d-1},\dots,z_n)$-basis. Thus we can associate each $A$ with vectors $v_0',w_0',v_1',w_1' \in \bbF_2^{n-d+1}$ such that $v_0'+w_0' = v_1' + w_1'$. 

Let $\sU_{S,\co(S)}$ be the unitary that implements the change of basis $(e_1,\dots,e_n) \to (z_1,\dots,z_n)$, where the $e_i$ are the standard basis vectors, and let \[\widetilde{\sU} \coloneqq \left(\sU_{S,\co(S)} \otimes \bbI^\regY \right)\sU^{\regX,\regY}\left({\sU_{S,\co(S)}^\dagger \otimes \bbI^\regY}\right).\]

Then, re-defining

\begin{align*}
&\ket{\widetilde{\psi}_A} \coloneqq \sU_{S,\co(S)}\ket{\psi_A}, \\
&\widetilde{\Pi}_{A_0} \coloneqq \bbI^{\otimes d-2} \otimes \ketbra{0}{0} \otimes \left(\ketbra{v_0'}{v_0'} + \ketbra{w_0'}{w_0'}\right) \otimes \bbI^{\regY},\\ 
&\widetilde{\Pi}_{A_1} \coloneqq \widetilde{\sU}^\dagger \left(\bbI^{\otimes d-2} \otimes \ketbra{1}{1} \otimes \left(\ketbra{v'_1}{v'_1} + \ketbra{w'_1}{w'_1}\right) \otimes \bbI^{\regY}\right)\widetilde{\sU},
\end{align*}

we have that $\ket{\widetilde{\psi}_A} \in \mathsf{Im}(\widetilde{\Pi}_{A_0})$ and $\|\widetilde{\Pi}_{A_1}\ket{\widetilde{\psi}_A}\| \geq \epsilon$ for all $A$ that could be sampled by $\cR[S]$. Moreover, we can replace the projections on the $d-1$'st qubit with identities, defining 

\begin{align*}
&\widetilde{\Pi}'_{A_0} \coloneqq \bbI^{\otimes d-1} \otimes \left(\ketbra{v_0'}{v_0'} + \ketbra{w_0'}{w_0'}\right) \otimes \bbI^{\regY},\\ 
&\widetilde{\Pi}'_{A_1} \coloneqq \widetilde{\sU}^\dagger \left(\bbI^{\otimes d-1} \otimes \left(\ketbra{v'_1}{v'_1} + \ketbra{w'_1}{w'_1}\right) \otimes \bbI^{\regY}\right)\widetilde{\sU},
\end{align*}

and still have that $\ket{\widetilde{\psi}_A} \in \mathsf{Im}(\widetilde{\Pi}'_{A_0})$ and $\|\widetilde{\Pi}'_{A_1}\ket{\widetilde{\psi}_A}\| \geq \epsilon$ for all $A$ that could be sampled by $\cR[S]$. Thus, we have reduced this problem to the ``two-dimensional'' case, which is covered in the next section. Since $n-d+1 > 10\log(1/\epsilon)+6$, \cref{thm:two-dimensional} implies that 

\[\E_{(A,B) \gets \cR[S]}[|\braket{\widetilde{\psi}_A | \widetilde{\psi}_B}|] < \frac{1}{2} - \epsilon^{13},\] which implies that

\[\E_{(A,B) \gets \cR[S]}[|\braket{\psi_A | \psi_B}|] < \frac{1}{2}-\epsilon^{13},\] completing the proof.



\end{proof}

\subsection{Two-dimensional case}\label{subsec:two-dim}
\begin{theorem}\label{thm:two-dimensional}
Let $n,m \in \bbN, \epsilon \in (0,1/8)$ be such that $n > 10\log(1/\epsilon)+6$. Let $\sU^{\regX,\regY}$ be a $(2^{n+m})$-dimensional unitary, where register $\regX$ is $2^n$ dimensions and register $\regY$ is $2^m$ dimensions. Let $\cA$ be the set of pairs of sets $(\{v_0,w_0\},\{v_1,w_1\})$ such that $v_0,w_0,v_1,w_1 \in \bbF_2^n$ and $v_0 + w_0 = v_1 + w_1$.\footnote{Note that this theorem is not strictly the two-dimensional version of \cref{thm:inner-product}, since $\cA$ is not exactly defined to be the set of two-dimensional affine subspaces. Rather it consists of pairs of two sets $\{v_0,w_0\},\{v_1,w_1\}$ where the vectors are arbitrary but satisfy $v_0 + w_0 = v_1 + w_1$. That is, $v_0,w_0,v_1,w_1$ here play the role of $v_0',w_0',v_1',w_1'$ in the proof of \cref{thm:inner-product}, and in particular $v_0,w_0$ do not necessarily start with 0 and $v_1,w_1$ do not necessarily start with 1.} We will write any $A \in \cA$ as $A \coloneqq (A_0, A_1)$, where $A_0 \coloneqq \{v_0,w_0\}$ and $A_1 = \{v_1,w_1\}$. For any such $A$, let \[\Pi_{A_0} \coloneqq \left(\ketbra{v_0}{v_0} + \ketbra{w_0}{w_0}\right)^{\regX} \otimes \bbI^{\regY}, \ \ \ \ \Pi_{A_1} \coloneqq \sU^\dagger \left(\left(\ketbra{v_1}{v_1} + \ketbra{w_1}{w_1}\right)^{\regX} \otimes \bbI^{\regY}\right)\sU.\]


Let $\cR$ be the set of pairs $(A,B)$ such that $|A_0 \cap B_0| = 1$ and $|A_1 \cap B_1| = 1$. Then for any set of states $\{\ket{\psi_A}\}_A$ such that for all $A \in \cA$, $\ket{\psi_A} \in \mathsf{Im}(\Pi_{A_0})$ and $\| \Pi_{A_1}\ket{\psi_A}\| \geq \epsilon$, \[\E_{(A,B) \gets \cR}[|\braket{\psi_A | \psi_B}|] < \frac{1}{2} - \epsilon^{13}.\]
\end{theorem}

First, we provide a high-level overview the proof. We note that it is easy to show that \[\E_{(A,B) \gets \cR}[|\braket{\psi_A | \psi_B}|] \leq \frac{1}{2},\] which only requires the condition that for all $A \in \cA$, $\ket{\psi_A} \in \mathsf{Im}(\Pi_{A_0})$. Adding the condition that $\|\Pi_{A_1}\ket{\psi_A}\| \geq \epsilon$ should intuitively only decrease this expected inner product, since many of the $\Pi_{A_1}$ are orthogonal. In particular, for any $A_0$, all the $\Pi_{A_1}$ such that $(A_0,A_1) \in \cA$ are orthogonal. To formalize this intuition, we proceed by contradiction, and assume that \[\E_{(A,B) \gets \cR}[|\braket{\psi_A | \psi_B}|] \geq \frac{1}{2} - \epsilon^{13}.\]

For each $A = (\{v_0,w_0\},\{v_1,w_1\})$, we will write $\ket{\psi_A}$ as \[\ket{\psi_A} \coloneqq \alpha_A^{v_0}\ket{v_0}^\regX\ket{\phi_A^{v_0}}^\regY + \alpha_A^{w_0} \ket{w_0}^\regX\ket{\phi_A^{w_0}}^\regY,\] and note that 
\[\E_{(A,B) \gets \cR}[|\braket{\psi_A | \psi_B}|] \leq \E_{(A,B) \gets \cR}[|\alpha_A^{v_{A,B}}| \cdot |\alpha_B^{v_{A,B}}| \cdot |\braket{\phi_A^{v_{A,B}}|\phi_B^{v_{A,B}}}|],\] where $\{v_{A,B}\} \coloneqq A_0 \cap B_0$.

Then, we proceed via the following steps.

\begin{enumerate}
    \item If we only require that $\ket{\psi_A} \in \mathsf{Im}(\Pi_{A_0})$, then one way to obtain the maximum expected inner product of $1/2$ is to set each $|\alpha_A^{v_0}| = 1/\sqrt{2}$ and for each $v_0$, let all $\ket{\phi_A^{v_0}}$ be the same vector. Then, each $|\alpha_A^{v_{A,B}}| \cdot |\alpha_B^{v_{A,B}}| = 1/2$ and each $|\braket{\phi_A^{v_{A,B}}|\phi_B^{v_{A,B}}}| = 1$. We show that this way of defining the $\alpha_A^{v_0}$ is ``robust'' in the sense that if the expected inner product is \emph{close} to 1/2, then for  \emph{many} of the $(A,B)$, $|\alpha_A^{v_{A,B}}| \cdot |\alpha_B^{v_{A,B}}|$ is \emph{close} to $1/2$ (\cref{claim:alphas}).
    \item We show that Step 1 implies that this way of defining $\ket{\phi_A^{v_0}}$ is also ``robust'', in the sense that for \emph{many} of the $(A,B)$, $|\braket{\phi_A^{v_{A,B}}|\phi_B^{v_{A,B}}}|$ is \emph{close} to 1 (\cref{claim:phis}). Thus, this property must be satisfied if our expected inner product is at least $1/2 - \epsilon^{13}$.
    \item By analyzing the graph of ``connections'' induced by $\cR$ between the elements of $\cA$, we show that Step 2 implies that there must exist \emph{some} $A_0^* = \{v_0^*,w_0^*\}$ with the following property. There any \emph{many} (exponential in $n$) states \[\left\{\ket{\psi_{(A_0^*,A_1)}} \coloneqq \alpha_{(A_0^*,A_1)}^{v_0^*} \ket{v_0^*}\ket{\phi_{(A_0^*,A_1)}^{v_0^*}} + \alpha_{(A_0^*,A_1)}^{w_0^*} \ket{w_0^*}\ket{\phi_{(A_0^*,A_1)}^{w_0^*}}\right\}_{A_1 : (A_0^*,A_1) \in \cA}\] such that the $\{\ket{\phi_{(A_0^*,A_1)}^{v_0^*}}\}$ are all close to each other, \emph{and} the $\{\ket{\phi_{(A_0^*,A_1)}^{w_0^*}}\}$ are all close to each other (\cref{claim:intersection}).
    \item Step 3 implies that there exists a \emph{large} (exponential in $n$) collection of states $\ket{\psi_{(A_0^*,A_1)}}$ such that (i) all $\ket{\psi_{(A_0^*,A_1)}}$ are \emph{close} to the \emph{same two-dimensional subspace}, and (ii) each $\ket{\psi_{(A_0^*,A_1)}}$ has $\epsilon$ overlap with a \emph{different orthogonal subspace} $\Pi_{A_1}$. We complete the proof by showing that this is impossible when $n$ is large enough compared to $1/\epsilon$. This relies on a Welch bound, which bounds the number of distinct vectors of some minimum distance from each other that can be packed into a low-dimensional subspace.
\end{enumerate}

\begin{proof}(of \cref{thm:two-dimensional})
Assume that \[\E_{(A,B) \gets \cR}[|\braket{\psi_A | \psi_B}|] \geq \frac{1}{2} - \epsilon^{13}.\] Using the fact that each $\ket{\psi_A} \in \mathsf{Im}(\Pi_{A_0})$, write each \[\ket{\psi_A} \coloneqq \alpha_A^{v_0} \ket{v_0}^\regX\ket{\phi_A^{v_0}}^\regY + \alpha_A^{w_0} \ket{w_0}^\regX\ket{\phi_A^{w_0}}^\regY,\] where $A_0 = \{v_0,w_0\}$. For any $(A,B) \in \cR$, define $\{v_{A,B}\} = A_0 \cap B_0$. 
Then, we have the following series of inequalities.

\begin{align*}
    \frac{1}{2} - \epsilon^{13} &\leq \E_{(A,B) \gets \cR}[|\braket{\psi_A | \psi_B}|] \\ &= \E_{(A,B) \gets \cR}[|{\alpha_A^{v_{A,B}}}^* \alpha_B^{v_{A,B}}\braket{\phi_A^{v_{A,B}}|\phi_B^{v_{A,B}}}|]\\ &\leq \E_{(A,B) \gets \cR}[|\alpha_A^{v_{A,B}}| \cdot |\alpha_B^{v_{A,B}}| \cdot |\braket{\phi_A^{v_{A,B}}|\phi_B^{v_{A,B}}}|] \\ &\leq \E_{(A,B) \gets \cR}[|\alpha_A^{v_{A,B}}| \cdot |\alpha_B^{v_{A,B}}|].
\end{align*}

Next, we show the following.

\begin{claim}\label{claim:alphas}
\[\Pr_{(A,B) \gets \cR}\left[|\alpha_A^{v_{A,B}}| \cdot |\alpha_B^{v_{A,B}}| \geq \frac{1}{2} - 2\epsilon^2\right] \geq 1-\epsilon^6.\]
\end{claim}

\begin{proof}

First, note that for any $(A,B) \in \cR$ where $A = (\{v_0,w_0\},\{v_1,w_1\})$ and $B = (\{v_0,u_0\},\{v_1,u_1\})$, the set $C = (\{w_0,u_0\},\{w_1,u_1\}) \in \cA$. This follows because 

\begin{align*}
    &A \in \cA \implies v_0 + w_0 = v_1 + w_1 \implies w_0 = v_0 + v_1 + w_1 \\
    &B \in \cA \implies v_0 + u_0 = v_1 + w_1 \implies u_0 = v_0 + v_1 + u_1, \\
    &\text{so} ~~ w_0 + u_0 = w_1 + u_1 \implies C \in \cA.
\end{align*}

This means that each $(A,B) \in \cR$ uniquely define a $C \in \cA$ such that all \[(A,B),(B,C),(C,A) \in \cR.\]

Thus, we will imagine sampling $(A,B) \gets \cR$ as follows. First, sample distinct $v_0,w_0,u_0 \gets \bbF_2^n$. Then, sample $v_1,w_1,u_1$ such that \[C_1 \coloneqq (\{v_0,w_0\},\{v_1,w_1\}), ~~ C_2 \coloneqq (\{v_0,u_0\},\{v_1,u_1\}), ~~ C_3 \coloneqq (\{w_0,u_0\},\{w_1,u_1\}) \in \cA.\] Let $(C_1,C_2,C_3) \gets \cS$ denote this sampling procedure. Finally, choose \[(A,B) \gets \cR[C_1,C_2,C_3] \coloneqq \{(C_1,C_2),(C_2,C_3),(C_3,C_1)\}.\]




Let \[E[C_1,C_2,C_3] \coloneqq \E_{(A,B) \gets \cR[C_1,C_2,C_3]}\left[|\alpha_A^{v_{A,B}}| \cdot |\alpha_B^{v_{A,B}}|\right].\]

Then, 

\[\E_{(A,B) \gets \cR}\left[|\alpha_A^{v_{A,B}}| \cdot |\alpha_B^{v_{A,B}}|\right] = \E_{(C_1,C_2,C_3) \gets \cS}\left[E[C_1,C_2,C_3]\right] \geq \frac{1}{2} - \epsilon^{13} > \frac{1}{2} - \epsilon^{12}.\]

Now, given any $(C_1,C_2,C_3)$ and corresponding 
\begin{align*}
&\ket{\psi_{C_1}} \coloneqq \alpha_{C_1}^{v_0}\ket{v_0}\ket{\phi_{C_1}^{v_0}} + \alpha_{C_1}^{w_0}\ket{w_0}\ket{\phi_{C_1}^{w_0}},\\ 
&\ket{\psi_{C_2}} \coloneqq \alpha_{C_2}^{v_0}\ket{v_0}\ket{\phi_{C_2}^{v_0}} + \alpha_{C_2}^{u_0}\ket{u_0}\ket{\phi_{C_2}^{u_0}},\\ 
&\ket{\psi_{C_3}} \coloneqq \alpha_{C_3}^{w_0}\ket{w_0}\ket{\phi_{C_3}^{w_0}} + \alpha_{C_3}^{u_0}\ket{u_0}\ket{\phi_{C_3}^{u_0}}, 
\end{align*}

we have that 

\[E[C_1,C_2,C_3] \leq \frac{1}{3}\left(|\alpha_{C_1}^{v_0}| \cdot |\alpha_{C_2}^{v_0}| + |\alpha_{C_1}^{w_0}| \cdot |\alpha_{C_3}^{w_0}| + |\alpha_{C_2}^{u_0}| \cdot |\alpha_{C_3}^{u_0}|\right).\]

By \cref{fact:robustness}, $E[C_1,C_2,C_3] \leq 1/2$, so by Markov, 

\[\Pr_{(C_1,C_2,C_3) \gets \cS}\left[\frac{1}{2} - E[C_1,C_2,C_3] \geq \epsilon^6\right] \leq \epsilon^6 \implies \Pr_{(C_1,C_2,C_3) \gets \cS}\left[E[C_1,C_2,C_3] \geq \frac{1}{2} - \epsilon^6\right] \geq 1 - \epsilon^6.\]

Moreover, whenever $E[C_1,C_2,C_3] \geq 1/2 - \epsilon^6$, we have that  

\[|\alpha_{C_1}^{v_0}| \cdot |\alpha_{C_2}^{v_0}| + |\alpha_{C_1}^{w_0}| \cdot |\alpha_{C_3}^{w_0}| + |\alpha_{C_2}^{u_0}| \cdot |\alpha_{C_3}^{u_0}| \geq \frac{3}{2} - \frac{6\epsilon^6}{2},\] so by \cref{fact:robustness}, \[|\alpha_{C_1}^{v_0}| \cdot |\alpha_{C_2}^{v_0}|, ~~ |\alpha_{C_1}^{w_0}| \cdot |\alpha_{C_3}^{w_0}|, ~~ |\alpha_{C_2}^{u_0}| \cdot |\alpha_{C_3}^{u_0}| \geq \frac{1}{2} - 2\epsilon^2,\] which completes the proof of the claim.

\end{proof}

\begin{claim}\label{claim:phis}
\[\Pr_{(A,B) \gets \cR}\left[|\braket{\phi_A^{v_{A,B}} | \phi_B^{v_{A,B}}}| \geq 1-\epsilon^6\right] \geq 1-2\epsilon^6.\]
\end{claim}

\begin{proof}

First, note that the proof of \cref{claim:alphas} also shows that \[\E_{(A,B) \gets \cR}\left[|\alpha_A^{v_{A,B}}| \cdot |\alpha_B^{v_{A,B}}|\right] \leq \frac{1}{2},\] since each $E[C_1,C_2,C_3] \leq 1/2$.

By our assumption that \[\E_{(A,B) \gets \cR}\left[|\alpha_A^{v_{A,B}}| \cdot |\alpha_B^{v_{A,B}}| \cdot |\braket{\phi_A^{v_{A,B}} | \phi_B^{v_{A,B}}}|\right] \geq \frac{1}{2}-\epsilon^{13}\] and linearity of expectation, 
\[\E_{(A,B) \gets \cR}\left[|\alpha_A^{v_{A,B}}| \cdot |\alpha_B^{v_{A,B}}| \cdot \left(1 - |\braket{\phi_A^{v_{A,B}} | \phi_B^{v_{A,B}}}|\right)\right] \leq \epsilon^{13}.\]

Now, assume for contradiction that \[\Pr_{(A,B) \gets \cR}\left[|\braket{\phi_A^{v_{A,B}} | \phi_B^{v_{A,B}}}| < 1-\epsilon^6\right] > 2\epsilon^6 \implies \Pr_{(A,B) \gets \cR}\left[1-|\braket{\phi_A^{v_{A,B}} | \phi_B^{v_{A,B}}}| > \epsilon^6\right] > 2\epsilon^6.\]

By \cref{claim:alphas}, this implies that \[\Pr_{(A,B) \gets \cR}\left[\left(1 - |\braket{\phi_A^{v_{A,B}} | \phi_B^{v_{A,B}}}| > \epsilon^6\right) \wedge \left(|\alpha_A^{v_{A,B}}| \cdot |\alpha_B^{v_{A,B}}| \geq \frac{1}{2} - 2\epsilon^2\right) \right] \geq \epsilon^6.\]

But then,
\[\E_{(A,B) \gets \cR}\left[|\alpha_A^{v_{A,B}}| \cdot |\alpha_B^{v_{A,B}}| \cdot \left(1 - |\braket{\phi_A^{v_{A,B}} | \phi_B^{v_{A,B}}}|\right)\right] > \epsilon^6 \cdot \epsilon^6 \cdot \left(\frac{1}{2} - 2\epsilon^2\right) \geq \frac{\epsilon^{12}}{4} > \epsilon^{13},\] whenever $\epsilon < 1/4$.

\end{proof}

\begin{claim}\label{claim:intersection}
There exists an $A_0^* = \{v_0^*,w_0^*\}$ and two unit vectors $\ket{\tau^{v_0^*}}, \ket{\tau^{w_0^*}}$ such that the following holds. Let \[\left\{\ket{\psi_{(A_0^*,A_1)}} \coloneqq \alpha_{(A_0^*,A_1)}^{v_0^*} \ket{v_0^*}\ket{\phi_{(A_0^*,A_1)}^{v_0^*}} + \alpha_{(A_0^*,A_1)}^{w_0^*} \ket{w_0^*}\ket{\phi_{(A_0^*,A_1)}^{w_0^*}}\right\}_{A_1: (A_0^*,A_1) \in \cA}\] be the set of $2^{n-1}$ states indexed by $A_1$ such that $(A_0^*,A_1) \in \cA$.\footnote{Note that there are $2^{n-1}$ possible states because the $A_1$ partition of the set $\bbF_2^n$ into disjoint unordered pairs of vectors, where each pair $\{v_1,w_1\}$ is such that $v_1 + w_1 = v_0^* + w_0^*$.} Then there exists a set $\cA_1^*$ of size at least $2^{n-2}$ such that for all $A_1 \in \cA_1^*$, \[|\braket{\phi_{(A_0^*,A_1)}^{v_0^*} | \tau^{v_0^*}}| \geq 1-2\epsilon^3 ~~ \text{and} ~~ |\braket{\phi_{(A_0^*,A_1)}^{w_0^*} | \tau^{w_0^*}}| \geq 1-2\epsilon^3.\]
\end{claim}


\begin{proof}

For each ordered pair $(v_0,w_0)$ where $v_0 \neq w_0 \in \bbF_2^n$, define \[\cR[(v_0,w_0)] \coloneqq \left\{(A,B) \in \cR : A_0 = \{v_0,w_0\} \wedge v_{A,B} = v_0\right\}.\] 
Then \cref{claim:phis} implies that there exists some set $\{v_0^*,w_0^*\}$ such that \begin{align*}&\Pr_{(A,B) \gets \cR[(v_0^*,w_0^*)]}\left[|\braket{\phi_A^{v_{A,B}} | \phi_B^{v_{A,B}}}| = |\braket{\phi_A^{v_0^*} | \phi_B^{v_0^*}}|  \geq 1-\epsilon^6\right] \geq 1 - 4\epsilon^6, ~~ \text{and} \\ &\Pr_{(A,B) \gets \cR[(w_0^*,v_0^*)]}\left[|\braket{\phi_A^{v_{A,B}} | \phi_B^{v_{A,B}}}| = |\braket{\phi_A^{w_0^*} | \phi_B^{w_0^*}}| \geq 1-\epsilon^6\right] \geq 1 - 4\epsilon^6.\end{align*}

Let $A_0^* = \{v_0^*,w_0^*\}$, let $\cA_1 \coloneqq \{\{v_1,w_1\}\}_{v_1 + w_1 = v_0^* + w_0^*}$ be the set of $A_1$ such that $(A_0^*,A_1) \in \cA$, let

\[\left\{\ket{\psi_{(A^*_0,A_1)}} \coloneqq \alpha_{(A^*_0,A_1)}^{v_0^*}\ket{v_0^*}\ket{\phi_{(A^*_0,A_1)}^{v_0^*}} + \alpha_{(A^*_0,A_1)}^{w_0^*}\ket{w_0^*}\ket{\phi_{(A^*_0,A_1)}^{w_0^*}}\right\}_{A_1 \in \cA_1},\]


and let \[\cA_1^{\times 2} = \{\{A_1,A_1'\}\}_{A_1 \neq A_1' \in \cA_1}.\] Note that by the definition of $\cA_1$, for any $\{A_1,A_1'\}\in \cA_1^{\times 2}$, it holds that $A_1\cap A_1'=\emptyset$. Now, we will argue that there exists a vector $\ket{\tau^{v_0^*}}$ and a set $\cA_1^{v_0^*}$ of size at least $\frac{3}{4}2^{n-1}$ such that for all $A_1 \in \cA_1^*$,

\[|\braket{\phi_{(A_0^*,A_1)}^{v_0^*} | \tau^{v_0^*}}| \geq 1-2\epsilon^3.\]

Consider any $\{A_1,A_1'\} \in \cA_1^{\times 2}$, where $A_1 = \{v_1,w_1\}$ and $A_1' = \{v_1',w_1'\}$. There are exactly four $B$ such that \[((A_0^*,A_1),B) \in \cR[(v_0^*,w_0^*)] ~~ \text{and} ~~ ((A_0^*,A_1'),B) \in \cR[(v_0^*,w_0^*)],\] which are\footnote{Note that $v_0^* + v_1 + v_1' \neq w_0^*$ since otherwise $w_1 = v_1 + (v_0^* + w_0^*) = v_1'$ and $w_1' = v_1 + (v_0^* + w_0^*) = v_1$ which would mean that $A_1 = A_1'$. Thus, for the first $B$ listed, $((A_0^*,A_1),B) \in \cR[(v_0^*,w_0^*)]$, and a similar argument holds for the rest of the $B$.} 
\begin{align*}
    B \in \begin{Bmatrix} \begin{array}{l} (\{v_0^*,v_0^*+v_1+v_1'\},\{v_1,v_1'\}), \\
    (\{v_0^*,v_0^*+w_1+w_1'\},\{w_1,w_1'\}), \\
    (\{v_0^*,v_0^* + v_1 + w_1'\}, \{v_1,w_1'\}), \\
    (\{v_0^*,v_0^* + w_1 + v_1'\}, \{w_1,v_1'\})\end{array}\end{Bmatrix}.
\end{align*}

Define \[\cR[(v_0^*,w_0^*),\{A_1,A_1'\}] \coloneqq \{((A_0^*,A_1),B)\}_B \cup \{((A_0^*,A_1'),B)\}_B\] where the indexing is over the four $B$ such that \[((A_0^*,A_1),B) \in \cR[(v_0^*,w_0^*)] ~~ \text{and} ~~ ((A_0^*,A_1'),B) \in \cR[(v_0^*,w_0^*)].\]

Note that for any two $\{A_1,A_1'\} \neq \{\widetilde{A}_1,\widetilde{A}_1'\} \in \cA_1^{\times 2}$, the sets $\cR[(v_0^*,w_0^*),\{A_1,A_1'\}]$ and $\cR[(v_0^*,w_0^*),\{\widetilde{A}_1,\widetilde{A}_1'\}]$ are disjoint, which can be seen by noting that $B_1$ always includes one vector from $A_1$ and one from $A_1'$.



Next, we claim that 

\[\cR[(v_0^*,w_0^*)]  = \bigcup_{\{A_1,A_1'\} \in \cA_1^{\times 2}} \cR[(v_0^*,w_0^*),\{A_1,A_1'\}],\] which follows from a counting argument. First, \[\bigg| \bigcup_{\{A_1,A_1'\} \in \cA_1^{\times 2}} \cR[(v_0^*,w_0^*),\{A_1,A_1'\}]\bigg| = 8 \cdot \binom{2^{n-1}}{2} = 2^{2n} - 2^{n+1}.\] Then, counting $|\cR[(v_0^*,w_0^*)]|$ directly, we can choose from any of the $2^{n-1}$ possible $A_1$, any $2^n - 2$ of the possible $B_0$, and then, given $B_0$, the two possible $B_1$ that intersect $A_1$.  Thus, \[\big| \cR[(v_0^*,w_0^*)]\big| = 2^{n-1} \cdot (2^n - 2) \cdot 2 = 2^{2n} - 2^{n+1}.\] This establishes that the sets 

\[\left\{\cR[(v_0^*,w_0^*),\{A_1,A_1'\}]\right\}_{\{A_1,A_1'\} \in \cA_1^{\times 2}}\]

partition $\cR[(v_0^*,w_0^*)]$ equally into sets of size 8. Thus,\footnote{Here, we show that there exists a large fraction of $\{A_1,A_1'\}$ such that \emph{all} $(A,B) \in \cR[(v_0^*,w_0^*),\{A_1,A_1'\}]$ are ``good'', meaning that $|\braket{\phi_A^{v_{A,B}} | \phi_B^{v_{A,B}}}| \geq 1-\epsilon^6$. As we will see later, it would have sufficed to prove the slightly weaker claim that there exists a large fraction of $\{A_1,A_1'\}$ such that \emph{at least} 5/8 of the $(A,B) \in \cR[(v_0^*,w_0^*),\{A_1,A_1'\}]$ are good. This is because for each such $\{A_1,A_1'\}$, we will just need a single $B$ (rather that all four) such that $((A_0^*,A_1),B)$ \emph{and} $((A_0^*,A_1'),B)$ are good.}  \[\Pr_{\{A_1,A_1'\} \gets \cA_1^{\times 2}}\left[\forall (A,B) \in \cR[(v_0^*,w_0^*),\{A_1,A_1'\}], |\braket{\phi_A^{v_0^*} | \phi_B^{v_0^*}}| \geq 1-\epsilon^6\right] \geq 1-32\epsilon^6,\]


which means that there exists some $A_1^* = \{v_1^*,w_1^*\}$ such that

\[\Pr_{A_1 \gets \cA_1 \setminus \{A_1^*\}} \left[\forall (A,B) \in \cR[(v_0^*,w_0^*),\{A_1^*,A_1\}], |\braket{\phi_A^{v_0^*} | \phi_B^{v_0^*}}| \geq 1-\epsilon^6\right] \geq 1 - 32\epsilon^6 \geq \frac{7}{8},\] which holds for all $\epsilon \leq 1/8$. 

Let $\cA_1^{v_0^*}$ be the set of $A_1$ such that \[\forall (A,B) \in \cR[(v_0^*,w_0^*),\{A_1^*,A_1\}], |\braket{\phi_A^{v_0^*} | \phi_B^{v_0^*}}| \geq 1-\epsilon^6,\] and note that $|\cA_1^{v_0^*}| \geq \frac{7}{8}(2^{n-1}-1) > \frac{3}{4}2^{n-1}$.


Now consider any $A_1  = \{v_1,w_1\} \in \cA_1^{v_0^*}$, and note that for $B = (\{v_0^*,v_0^*+v_1^*+v_1\},\{v_1^*,v_1\})$, we have that \[((A_0^*,A_1^*),B), ((A_0^*,A_1),B) \in \cR[(v_0^*,w_0^*),\{A_1^*,A_1\}].\]

Thus, we know that \[|\braket{\phi_{(A_0^*,A_1^*)}^{v_0^*} | \phi^{v_0^*}_{B}}| \geq 1-\epsilon^6, ~~ \text{and} ~~ |\braket{\phi_{(A_0^*,A_1)}^{v_0^*} | \phi^{v_0^*}_{B}}| \geq 1 - \epsilon^6,\]

so by \cref{fact:inner-product}, 

\[|\braket{\phi_{(A_0^*,A_1)}^{v_0^*} | \phi_{(A_0^*,A_1^*)}^{v_0^*}}| \geq (1-\epsilon^6)^2 - \sqrt{2\epsilon^6} \geq 1-2\epsilon^3.\]

Then if we set $\ket{\tau^{v_0^*}} \coloneqq \ket{\phi_{(A_0^*,A_1^*)}^{v_0^*}}$, we have that for all $A_1 \in \cA_1^{v_0^*}$, 

\[|\braket{\phi_{(A_0^*,A_1)}^{v_0^*} | \tau^{v_0^*}}| \geq 1 - 2\epsilon^3.\]

Finally, repeating the analysis for $\cR[(w_0^*,v_0^*)]$, there exists a $\ket{\tau^{w_0^*}}$ and a set $\cA_1^{w_0^*}$ of size at least $\frac{3}{4}2^{n-1}$ such that for all $A_1 \in \cA_1^{w_0^*}$, 
\[|\braket{\phi_{(A_0^*,A_1)}^{w_0^*} | \tau^{w_0^*}}| \geq 1 - 2\epsilon^3.\]

Thus, setting $\cA^*_1 \coloneqq \cA_1^{v_0^*} \cap \cA_1^{w_0^*}$ (which has size $\geq 2^{n-2}$) completes the proof.

\end{proof}

Finally, we can reach a contradiction by using the fact that for any fixed $A_0^*$, all of the $\Pi_{A_1}$ such that $(A_0^*,A_1) \in \cA$ are  orthogonal, which follows from the definition of the $\Pi_{A_1}$. 

Now, define the rank-two projector 

\[\Pi^* \coloneqq \ket{v_0^*}\ketbra{\tau^{v_0^*}}{\tau^{v_0^*}}\bra{v_0^*} + \ket{w_0^*}\ketbra{\tau^{w_0^*}}{\tau^{w_0^*}}\bra{w_0^*}.\]

By \cref{claim:intersection} and the assumption of the theorem, for each $A_1 \in \cA_1^*$ we know that 

\[\| \Pi^* \ket{\psi_{(A_0^*,A_1)}}\| \geq 1-2\epsilon^3 ~~ \text{and} ~~ \| \Pi_{A_1} \ket{\psi_{(A_0^*,A_1)}}\| \geq \epsilon.\]

For each $A_1 \in \cA_1^*$, define

\[\ket{\psi^*_{A_1}} \coloneqq \frac{\Pi^*\ket{\psi_{(A_0^*,A_1)}}}{\|\Pi^*\ket{\psi_{(A_0^*,A_1)}}\|}.\]

Thus, since $|\braket{\psi_{A_1}^* | \psi_{(A_0^*,A_1)}}| \geq 1-2\epsilon^3$ and $\|\Pi_{A_1}\ket{\phi_{(A_0^*,A_1)}}\| \geq \epsilon$, by \cref{fact:inner-product} (second part) it holds that \[\| \Pi_{A_1}\ket{\psi_{A_1}^*}\| \geq \epsilon(1-2\epsilon^3) - 2\epsilon^{3/2} \geq \frac{\epsilon}{2},\] which holds for all $\epsilon \leq 1/8$. 


Consider the following algorithm, which will eventually select all $\{\ket{\psi^*_{A_1}}\}_{A_1 \in \cA_1^*}$.

\begin{enumerate}
    \item Set $i = 1$.
    \item Select an arbitrary (not yet selected) $\ket{\psi^*_{A_1}}$, and define $\ket{\psi_i} \coloneqq \ket{\psi^*_{A_1}}$.
    \item Select all (not yet selected) $\ket{\psi^*_{A_1}}$ such that $|\braket{\psi^*_{A_1} | \psi_i}| \geq 1-\epsilon^4$.
    \item Set $i = i+1$ and go back to Step 2.
\end{enumerate}

First, we claim that in each invocation of Step 3, we select at most $16/\epsilon^2$ vectors. To see this, note that for each $\ket{\psi^*_{A_1}}$ selected in Step 3 during the $i$'th loop of the procedure, $|\braket{\psi_{A_1}^* | \psi_i}| \geq 1-\epsilon^4$ and $\|\Pi_{A_1}\ket{\psi^*_{A_1}}\| \geq \epsilon/2$. Thus, by \cref{fact:inner-product} (second part), 

\[\|\Pi_{A_1}\ket{\psi_i}\| \geq \frac{\epsilon}{2}(1-\epsilon^4)-\sqrt{2}\epsilon^2 \geq \frac{\epsilon}{4},\] which holds for all $\epsilon \leq 1/8$. Since the $\Pi_{A_1}$ are all orthogonal, and $\ket{\psi_i}$ has a component of at least $\epsilon^2/16$ squared norm on each, we conclude that there can be at most $16/\epsilon^2$ such $A_1$. 

Second, let $I$ be the value of $i$ when the procedure terminates. Note that the $\{\ket{\psi_i}\}_{i \in [I]}$ are all in the image of a two-dimensional subspace $\mathsf{Im}(\Pi^*)$, and for all $i \neq j,$ $|\braket{\psi_i | \psi_j}| < 1-\epsilon^4$. 

Now, we use a Welch bound.

\begin{importedtheorem}[\cite{1055219}]
Let $\{x_1,\dots,x_I\}$ be unit vectors in $\bbC^d$, and define $c = \max_{i \neq j}|\braket{x_i | x_j}|$. Then for every $k \in \bbN$,

\[c^{2k} \geq \frac{1}{I-1}\left(\frac{I}{\binom{k+d-1}{k}} - 1\right).\]
\end{importedtheorem}

Setting $d = 2$ and $k = I/2 - 1$, we have that \[\frac{1}{I-1} \leq (1-\epsilon^4)^{I-2} \leq e^{-\epsilon^4 (I-2)} \implies \frac{1}{\epsilon^4} \geq \frac{I-2}{\ln(I-1)} \geq \sqrt{I} \implies I \leq \frac{1}{\epsilon^8}.\]

Putting these two facts together, we have that the size of $\cA_1^*$ is at most $16/\epsilon^{10}$, meaning that \[2^{n-2} \leq \frac{16}{\epsilon^{10}} \implies 2^n \leq \frac{64}{\epsilon^{10}},\] and contradicting the fact that $n > 10\log(1/\epsilon) + 6$.

\end{proof}

\subsection{Useful facts}



\begin{fact}\label{fact:inner-product}
Let $\ket{\phi_a},\ket{\phi_b}$ be complex unit vectors such that $|\braket{\phi_a | \phi_b}| \geq 1-\alpha$. Then the following hold.
\begin{enumerate}
    \item If $\ket{\phi_c}$ is a complex unit vector such that $|\braket{\phi_b | \phi_c}| \geq \beta$, then $|\braket{\phi_a | \phi_c}| \geq \beta(1-\alpha) - \sqrt{2\alpha}$.
    \item If $\Pi$ is a projector such that $\| \Pi \ket{\phi_b}\| \geq \beta$, then $\| \Pi \ket{\phi_a}\| \geq \beta(1-\alpha) - \sqrt{2\alpha}$.
\end{enumerate}

\end{fact}

\begin{proof}
To show the first part, write $\ket{\phi_a} = e^{i \theta}(1-\alpha)\ket{\phi_b} + \sqrt{2\alpha-\alpha^2}\ket{\phi_b^\bot}$ for some $\theta$ and $\ket{\phi_b^\bot}$ orthogonal to $\ket{\phi_b}$. 
Then

\begin{align*}
    |\braket{\phi_a | \phi_c}| &= |e^{i\theta}(1-\alpha)\braket{\phi_b | \phi_c} + \sqrt{2\alpha-\alpha^2}\braket{\phi_b^\bot | \phi_c}| \\ 
    &\geq |e^{i\theta}(1-\alpha)\braket{\phi_b | \phi_c}| - \sqrt{2\alpha-\alpha^2} \\
    &\geq \beta(1-\alpha) - \sqrt{2\alpha}.
\end{align*}

To show the second part, define \[\ket{\phi_c} \coloneqq \frac{\Pi\ket{\phi_b}}{\|\Pi\ket{\phi_b}\|},\] and note that \[|\braket{\phi_b | \phi_c}| = \frac{\bra{\phi_b}\Pi\ket{\phi_b}}{\| \Pi \ket{\phi_b}\|} = \frac{\|\Pi\ket{\phi_b}\|^2}{\|\Pi\ket{\phi_b}\|} = \|\Pi\ket{\phi_b}\| \geq \beta.\]

Thus, \[\|\Pi\ket{\phi_a}\| \geq \|\ketbra{\phi_c}{\phi_c}\ket{\phi_a}\| = |\braket{\phi_a | \phi_c}| \geq \beta(1-\alpha) -\sqrt{2\alpha},\] where the first inequality follows because $\ket{\phi_c} \in \mathsf{Im}(\Pi)$ and the second inequality follows from the first part.

\end{proof}

\begin{fact}\label{fact:robustness}
Let \[u_1 \coloneqq \begin{pmatrix} a_1 \\ a_2 \\ 0\end{pmatrix}, u_2 \coloneqq \begin{pmatrix} b_1 \\ 0 \\ b_2 \end{pmatrix}, u_3 \coloneqq \begin{pmatrix} 0 \\ c_1 \\ c_2 \end{pmatrix}\] be three unit vectors in $\bbR_{\geq 0}^3$. Then, \[u_1 \cdot u_2 + u_1 \cdot u_3 + u_2 \cdot u_3 \leq \frac{3}{2}.\] Moreover, for any $\delta \in [0,1/2]$, if \[u_1 \cdot u_2 + u_1 \cdot u_3 + u_2 \cdot u_3 \geq \frac{3}{2}-\frac{\delta^3}{2},\] then \[u_1 \cdot u_2 \geq \frac{1}{2}- \delta, ~~ u_1 \cdot u_3 \geq \frac{1}{2} - \delta, ~~ \text{and} ~~ u_2 \cdot u_3 \geq \frac{1}{2} - \delta.\]
\end{fact}

\begin{proof}
We begin with the first part of the claim. Let $v_1 \coloneqq (a_1 \ a_2 \ b_1 \ b_2 \ c_1 \ c_2)$ and $v_2 \coloneqq (b_1 \ c_1 \ a_1 \ c_2 \ a_2 \ b_2)$. Then,

\begin{align*}
    u_1 \cdot u_2 + u_1 \cdot u_3 + u_2 \cdot u_3 = \frac{1}{2}v_1 \cdot v_2^\top \leq \frac{1}{2}(a_1^2 + a_2^2 + b_1^2 + b_2^2 + c_1^2 + c_2^2) = \frac{3}{2},
\end{align*}

where the inequality is Cauchy-Schwartz.

Now, we prove the ``moreover'' part. This is trivial when $\delta = 1/2$, so suppose that $u_1 \cdot u_2 = 1/2 - \delta$ for some $\delta \in [0,1/2)$. We will show that this implies that \[u_1 \cdot u_2 + u_1 \cdot u_3 + u_2 \cdot u_3 \leq \frac{3}{2} - \frac{\delta^3}{2},\] which, by symmetry, would complete the proof.

Define the value \[m \coloneqq \max_{\substack{\begin{array}{c}u_1,u_2,u_3, \\ u_1 \cdot u_2 = 1/2 - \delta\end{array}}}\left\{u_1 \cdot u_3 + u_2 \cdot u_3\right\},\] and let $a_1 = \sqrt{1-x}, a_2 = \sqrt{x}, b_1 = \sqrt{1-y}$ and $b_2 = \sqrt{y}$ for some $x,y \in [0,1)$. Then,

\begin{align*}
    m &= \max_{x,y \in [0,1), \sqrt{1-x}\sqrt{1-y} = 1/2 - \delta}\left\{\sqrt{x}c_1 + \sqrt{y}c_2\right\} \\
    &\leq \max_{x,y \in [0,1), \sqrt{1-x}\sqrt{1-y} = 1/2 - \delta}\left\{\sqrt{x+y}\sqrt{c_1 + c_2}\right\} \\
    &= \max_{x,y \in [0,1), \sqrt{1-x}\sqrt{1-y} = 1/2 - \delta}\left\{\sqrt{x+y}\right\},
\end{align*}

where the inequality is Cauchy-Schwartz.





Next, we solve for \[y = 1 - \frac{(\frac{1}{2}-\delta)^2}{1-x},\] and see that 

\begin{align*}
    m^2 &= \max_{x \in [0,1)}\left\{x + 1 - \frac{(\frac{1}{2}-\delta)^2}{1-x}\right\} \\
    &= \max_{x \in [0,1)} \left\{2 - \frac{2}{1-x}\left(\frac{(1-x)^2 + \left(\frac{1}{2}-\delta\right)^2}{2}\right)\right\} \\
    &\leq 2 - 2\left(\frac{1}{2}-\delta\right) \\
    &= 1+2\delta,
\end{align*}

where the inequality is AM-GM.


Thus, to complete the proof it suffices to show that
\[\frac{1}{2} - \delta + \sqrt{1+2\delta} \leq \frac{3}{2} - \frac{\delta^3}{2}.\] If $\delta = 0$, then both sides are 1, so now assume that $\delta > 0$. Then

\begin{align*}
\frac{1}{2} - \delta + \sqrt{1+2\delta} \leq \frac{3}{2} - \frac{\delta^3}{2} &\iff \sqrt{1+2\delta} \leq 1 + \delta - \frac{\delta^3}{2} \\ &\iff 1+2\delta \leq 1+2\delta + \delta^2 - (1+\delta)\delta^3 + \frac{\delta^6}{4} \\ &\iff \delta + \delta^2 - \frac{\delta^4}{4} \leq 1,
\end{align*}

which is true for all $\delta \in (0,1/2)$.
\end{proof}

\section{Remaining Proofs from \cref{subsec:classical-verification}}\label{sec:appendix-CVQC}

In this appendix, we prove \cref{lemma:soundness}, \cref{lemma:Q-in}, and \cref{lemma:Q-out}. We proceed via three steps.

\begin{enumerate}
    \item Compile the information-theoretic protocol $\Pi^\QV$ from 
    \cref{subsec:QPIP} into a 4-message quantum ``commit-challenge-response'' protocol $\Pi^\CCR$ with a classical verifier. This compilation is achieved via the use of Mahadev's measurement protocol \cite{SIAMCOMP:Mahadev22}. As argued in \cite{TCC:Bartusek21}, the resulting protocol satisfies a ``computationally orthogonal projectors'' property, which was first described by \cite{TCC:ACGH20}.
    \item Apply parallel repetition to $\Pi^\CCR$ to obtain $\Pi^\parl$, and observe that the parallel repetition theorem of \cite{TCC:Bartusek21} implies that the analogues of \cref{lemma:soundness}, \cref{lemma:Q-in}, and \cref{lemma:Q-out} hold in $\Pi^\parl$.
    \item Apply Fiat-Shamir to $\Pi^\parl$ to obtain the protocol $\Pi^\CV$ from \proref{fig:CVQC}, and observe that Measure and Re-program (\cref{thm:measure-and-reprogram}) implies that \cref{lemma:soundness}, \cref{lemma:Q-in}, and \cref{lemma:Q-out} must also hold with respect to $\Pi^\CV$.
\end{enumerate}



\protocol{Commit-challenge-response protocol $\Pi^\CCR = (\sV_\Gen^\CCR, \sP_\Com^\CCR, \sP_\Prove^\CCR, \sV^\CCR_\Ver)$}{A quantum ``commit-challenge-response'' protocol for verifying quantum partitioning circuits.}{fig:CVQC-single}{

\textbf{Parameters}: Number of qubits $\ell = \ell(\secp)$ in the prover's state.\\

\begin{itemize}
    \item $\sV_\Gen^\CCR(1^\secp,Q) \to (\pp,\sparam)$: Sample $(h,S) \gets \sV_\Gen^\QV(1^\secp,Q)$ and $\{(\pk_j,\sk_j) \gets \TCF.\Gen(1^\secp,h_j)\}_{j \in [\ell]}$, and set \[\pp \coloneqq \{\pk_j\}_{j \in [\ell]}, \ \ \ \sparam \coloneqq (h,S,\{\sk_j\}_{j \in [\ell]}).\]
    \item $\sP_\Com^\CCR(1^\secp,Q,x,\pp) \to (\regB,\regZ,y)$: Prepare the state $\ket{\psi} \gets \sP^\QV(1^\secp,Q,x)$ on register $\regB = (\regB_1,\dots,\regB_\ell)$, which we write as 
    \[\ket{\psi} \coloneqq \sum_{v \in \{0,1\}^\ell}\alpha_v\ket{v}^{\regB},\] and then for each $j \in [\ell]$, apply $\TCF.\Eval[\pk_{j}](\regB_{j}) \to (\regB_{j},\regZ_{j},\regY_{j})$, resulting in the state
    \[\sum_{v \in \{0,1\}^\ell}\alpha_v\ket{v}^{\regB}\ket{\psi_{\pk_{1},v_1}}^{\regZ_{1},\regY_{1}},\dots,\ket{\psi_{\pk_{\ell},v_\ell}}^{\regZ_{\ell},\regY_{\ell}}.\] Finally, measure registers $\regY_1,\dots,\regY_\ell$ in the standard basis to obtain string $y \coloneqq \{y_j\}_{j \in [\ell]}$.
    \item The verifier samples a random bit $d \gets \{0,1\}$, and sends $d$ to the prover.
    \item $\sP_\Prove^\CCR(\regB,\regZ,d) \to z$: If $d = 0$, the prover measures registers $\regB,\regZ$ in the standard basis to obtain $z \coloneqq \{b_j,z_j\}_{j \in [\ell]}$. If $d = 1$, the prover applies $J(\cdot)$ coherently to each register $\regZ_j$ and then measures registers $\regB,\regZ$ in the Hadamard basis to obtain $z \coloneqq \{b_j,z_j\}_{j \in [\ell]}$.
    \item $\sV_\Ver^\CCR(Q,x,\sparam,y,d,z) \to \{\{q_t\}_{t \in [\secp]}\} \union \{\top,\bot\}$:
    
    \begin{itemize}
        \item Parse $y \coloneqq \{y_j\}_{j \in [\ell]}$ and $z \coloneqq \{b_j,z_j\}_{j \in [\ell]}$.
        \item If $d = 0$, for each $j \in [\ell]$ compute $\TCF.\Ch(\pk_{j},b_{j},z_{j},y_{j})$. If any are $\bot$, then output $\bot$, and otherwise output $\top$.
        \item If $d = 1$, do the following for each $j \in [\ell]$. 
        \begin{itemize}
            \item If $h_{j} = 0$, compute $\TCF.\Invert(0,\sk_{j},y_{j})$, abort and output $\bot$ if the output is $\bot$, and otherwise parse the output as $(m_{j},x_{j})$.
            \item If $h_{j} = 1$, compute $\TCF.\Invert(1,\sk_{j},y_{j})$, abort and output $\bot$ if the output is $\bot$, and otherwise parse the output as $(0,x_{j,0}),(1,x_{j,1})$. Then, check $\TCF.\IsValid(x_{j,0},x_{j,1},z_{j})$ and abort and output $\bot$ if the result is $\bot$. Next, set $m_{j} \coloneqq b_{j} \oplus z_{j} \cdot (J(x_{j,0}) \oplus J(x_{j,1}))$.
        \end{itemize}
        Then, let $m \coloneqq (m_1,\dots,m_\ell)$ and compute $\sV_\Ver^\QV(Q,x,h,m)$. Output $\bot$ if the result is $\bot$, and otherwise output $\{q_t\}_{t \in [\secp]} \coloneqq m[S]$.
    \end{itemize}
             
\end{itemize}

}

\begin{proof}(of \cref{lemma:soundness}, \cref{lemma:Q-in}, and \cref{lemma:Q-out})
\medskip

\noindent\underline{Step 1}. We first describe the syntax of a generic commit-challenge-response protocol between a quantum prover $\sP$ and a classical verifier $\sV$. 

\begin{itemize}
    \item Commit: $\sP(1^\secp)$ and $\sV(1^\secp;r)$ engage in a two-message commitment protocol, where $r$ are the random coins used by $\sV$ to generate the first message of the protocol, and the prover responds with a classical commitment string.
    \item Challenge: $\sV$ samples a random bit $d \gets \{0,1\}$ and sends it to $\sP$.
    \item Response: $\sP$ computes a (classical) response $z$ and sends it to $\sV$. 
    \item Output: $\sV$ receives $z$ and decides to either accept and output $\top$ or reject and output $\bot$.
\end{itemize}

Consider any QPT adversarial prover $\sP^*$, and let $\ket{\psi^{\sP^*}_{\secp,r}}^{\regA,\regC}$ be the (purified) state of the prover after interacting with $\sV(1^\secp;r)$ in the commit phase, where $\regC$ holds the (classical) prover message output during this phase, and $\regA$ holds its remaining state.

The remaining strategy of the prover can be described by family of unitaries $\left\{\sU^{\sP^*}_{\secp,0},\sU^{\sP^*}_{\secp,1}\right\}_{\secp \in \mathbb{N}}$, where $\sU^{\sP^*}_{\secp,0}$ is applied to $\ket{\psi^{\sP^*}_{\secp,r}}$ on challenge 0 (followed by a measurement of $z$), and $\sU^{\sP^*}_{\secp,1}$ is applied to $\ket{\psi^{\sP^*}_{\secp,r}}$ on challenge 1 (followed by a measurement of $z$).

Let $\sV_{\secp,r,0}$ denote the accept projector applied by the verifier to the prover messages when $d=0$, and define $\sV_{\secp,r,1}$ analogously. Then define the following projectors on registers $(\regA,\regC)$.

\[\Pi^{\sP^*}_{\secp,r,0} \coloneqq {\sU^{\sP^*}_{\secp,0}}^\dagger \sV_{\secp,r,0}\sU^{\sP^*}_{\secp,0}, \ \ \ \Pi^{\sP^*}_{\secp,r,1} \coloneqq {\sU^{\sP^*}_{\secp,1}}^\dagger \sV_{\secp,r,1}\sU^{\sP^*}_{\secp,1}.\]

\begin{definition}\label{def:computationally-orthogonal-projectors}
A \emph{commit-challenge-response} protocol has \emph{computationally orthogonal projectors} if for any QPT prover $\{\sP^*_\secp\}_{\secp \in \bbN}$, \[\E_{r}\left[\bra{\psi^{\sP^*}_{\secp,r}}\Pi^{\sP^*}_{\secp,r,0} \Pi^{\sP^*}_{\secp,r,1} \Pi^{\sP^*}_{\secp,r,0} \ket{\psi^{\sP^*}_{\secp,r}}\right] = \negl(\secp).\]
\end{definition}

Now, consider running protocol $\Pi^\CCR$ with some fixed circuit $Q$ and input $x$, and suppose that $P$ is a predicate such that $P(Q(\cdot))$ is pseudo-deterministic. We define the verifier acceptance predicates as follows.

\begin{itemize}
    \item $\sV_{\secp,r,0}$ runs $\sV_\Ver^\CCR$ on $d = 0$.
    \item $\sV_{\secp,r,1}$ runs $\sV_\Ver^\CCR$ on $d = 1$ to obtain either $\bot$ or $\{q_t\}_{t \in [\secp]}$. In the latter case, it outputs $\top$ if $\Maj\left(\{P(q_t)\}_{t \in [\secp]}\right) = 1-P(Q(x))$ and $\bot$ otherwise.
\end{itemize}

Then, by \cite[Lemma 4.4]{TCC:Bartusek21}, which uses the soundness of $\Pi^\QV$ (\cref{impthm:QV}) and the soundness of the measurement protocol (\cite{SIAMCOMP:Mahadev22}), we have the following claim.

\begin{claim}\label{claim:comp-orthogonal}
For any $\{\sP^*_\secp\}_{\secp \in \bbN}$ attacking $\Pi^\CCR$ (Protocol in~\cref{fig:CVQC-single}), it holds that 
\[\E_{r}\left[\bra{\psi^{\sP^*}_{\secp,r}}\Pi^{\sP^*}_{\secp,r,0} \Pi^{\sP^*}_{\secp,r,1} \Pi^{\sP^*}_{\secp,r,0} \ket{\psi^{\sP^*}_{\secp,r}}\right] = \negl(\secp),\] where the verifier acceptance predicates $\sV_{\secp,r,0},\sV_{\secp,r,1}$ used to define $\Pi_{\secp,r,0}^{\sP^*}$ and $\Pi_{\secp,r,1}^{\sP^*}$ are as described above.
\end{claim}

\noindent\underline{Step 2}. In this step, we will use the following imported theorem.


\begin{importedtheorem}[\cite{TCC:Bartusek21}, Theorem 3.1]\label{thm:parallel-repetition} Let $\epsilon > 0$ and $0 < \delta < 1$ be constants. Let $\Pi$ be a commit-challenge-response protocol with computationally orthogonal projectors, and where the verifier's $d=0$ acceptance predicate is publicly computable given the verifier's first message. Let $\Pi^\parl$ be the $\secp^{1+\epsilon}$ parallel repetition of $\Pi$, where the verifier's challenge string $T$ is sampled as a uniformly random $\secp^{1+\epsilon}$ bit string with Hamming weight $\secp$. Then for any QPT adversarial prover $\sP^*$ attacking $\Pi^\parl$, the probability that the verifier accepts all rounds $i$ such that $T_i = 0$ and $\geq \delta \cdot \secp$ rounds $i$ such that $T_i = 1$ is $\negl(\secp)$.
\end{importedtheorem}

Now, we define the protocol $\Pi^\parl = (\sV_\Gen^\parl, \sP_\Com^\parl, \sP_\Prove^\parl, \sV^\parl_\Ver)$ to be the $\secp^2$ parallel repetition of $\Pi^\CCR$, where the verifier's challenge string $T$ is sampled as a uniformly random $\secp^2$ bit string with Hamming weight $\secp$. Then, we can prove the following lemmas about $\Pi^\parl$.

\begin{lemma}[$\Pi^\parl$ analogue of \cref{lemma:soundness}]\label{lemma:soundness-4-msg}
For any family $\{Q_\secp,P_\secp\}_{\secp \in \bbN}$ such that $\{P_\secp \circ Q_\secp\}_{\secp \in \bbN}$ is pseudo-deterministic, sequence of inputs $\{x_\secp\}_{\secp \in \bbN}$, and QPT adversary $\{\sA_\secp\}_{\secp \in \bbN}$, it holds that 
    \[\Pr\left[\begin{array}{l}\sV_\Ver^\parl(Q,x,\sparam,y,T,z) = \{\{q_{i,t}\}_{t \in [\secp]}\}_{i:T_i =1} ~~ \wedge \\ \MM_\secp\left(\{\{P(q_{i,t})\}_{t \in [\secp]}\}_{i:T_i =1}\right) = 1-P(Q(x)) \end{array} : \begin{array}{r}(\pp,\sparam) \gets \sV^\parl_\Gen(1^\secp,Q) \\ y \gets \sA(\pp) \\ T \gets \{0,1\}^{\binom{\secp^2}{\secp}} \\ z \gets \sA(T) \end{array}\right] = \negl(\secp),\]
    
    where $\sA$ maintains an internal state, which we leave implicit above.
\end{lemma}

\begin{proof}
We have to rule out a prover that makes the verifier of $\Pi^\CCR$ accept each of the $\secp^2-\secp$ rounds where $T_i = 0$, and, for a majority of the rounds $i$ where $T_i = 1$, accepts and outputs $\{q_{i,t}\}_{t \in [\secp]}$ such that $\Maj\left(\{P(q_{i,t})\}_{t \in [\secp]}\right) = 1 - P(Q(x))$. This is directly ruled out by \cref{claim:comp-orthogonal} and \cref{thm:parallel-repetition} with $\epsilon = 1$ and $\delta = 1/2$.
\end{proof}

\begin{lemma}[$\Pi^\parl$ analogue of \cref{lemma:Q-in}]\label{lemma:Q-in-4-msg}
For any family $\{Q_\secp,P_\secp\}_{\secp \in \bbN}$ such that $\{P_\secp \circ Q_\secp\}_{\secp \in \bbN}$ is pseudo-deterministic, sequence of inputs $\{x_\secp\}_{\secp \in \bbN}$, and QPT adversary $\{\sA_\secp\}_{\secp \in \bbN}$, it holds that 
        \[\Pr\left[\begin{array}{l}\sV^\parl_\Ver(Q,x,\sparam,y,T,z) \neq \bot ~~ \wedge \\ w \notin D_{\inside}[P,P(Q(x))]\end{array} : \begin{array}{r}(\pp,\sparam) \gets \sV^\parl_\Gen(1^\secp,Q) \\ y \gets \sA(\pp) \\ T \gets \{0,1\}^{\binom{\secp^2}{\secp}} \\ z \gets \sA(T) \\  w \coloneqq \mathsf{TestRoundOutputs}[\sparam](y,T,z)\end{array}\right] = \negl(\secp),\] where $\sA$ maintains an internal state, which we leave implicit above, and where $\TestOut$ is defined as in \cref{subsec:classical-verification}, except that string $T$ is explicitly given rather than being computed by a random oracle $H$.
\end{lemma}

\begin{proof}
First, we make the following observation. For every $i \in [\secp^2]$, the strings $\{q_{i,t}\}_{t \in [\secp]}$ that the verifier would output conditioned on accepting and on $T_i = 1$ are already determined by the prover's first message $y_i \coloneqq (y_{i,1},\cdots,y_{i,\ell})$ and the secret parameters $\sparam$. Indeed, recall from the description of $\Pi^\QV$ that the bits in $\{q_{i,t}\}_{t \in [\secp]}$ are computed from indices $j \in [\ell]$ where the basis $h_{i,j} = 0$ (that is, they are the result of standard basis measurements). Moreover, when $h_{i,j} = 0$, $\pk_{i,j}$ defines an injective function, which follows from \cref{def:clawfree}, correctness properties (a) and (c). Thus, each string $y_{i,j}$ either has one or zero pre-images. If it has zero, the verifier would never accept when $T_i = 1$, and if it has one, the verifier would only accept the first bit $b_{i,j}$ of the pre-image.

So, we can define $\{\{q_{i,t}\}_{t \in [\secp]}\}_{i \in [\secp^2]}$ based on the prover's first message $\{y_i\}_{i \in [\secp^2]}$. Then,
\begin{itemize}
    \item Let $a$ be the fraction of $\{q_{i,t}\}_{t \in [\secp]}$ such that $\Maj\left(\{P(q_{i,t})\}_{t \in [\secp]}\right) = P(Q(x))$ over $i \in [\secp^2]$.
    \item Let $b$ be the fraction of $\{q_{i,t}\}_{t \in [\secp]}$ such that $\Maj\left(\{P(q_{i,t})\}_{t \in [\secp]}\right) = P(Q(x))$ over $i : T_i = 1$.
\end{itemize}

By the definition of $D_{\mathsf{in}}[P,P(Q(x))]$, 

\[w \notin D_{\mathsf{in}}[P,P(Q(x))] \implies a \leq \frac{3}{4} + \frac{1}{\secp}.\]
Moreover, by \cref{claim:comp-orthogonal} and \cref{thm:parallel-repetition} with $\epsilon = 1$ and $\delta = 1/5$,

\[\Pr\left[\sV^\parl_\Ver(Q,x,\sparam,y,T,z) \neq \bot \wedge b < \frac{4}{5}\right] = \negl(\secp).\]

Thus, the proof is completed by showing that 

\[\Pr\left[b-a \geq \frac{4}{5} - \left(\frac{3}{4} + \frac{1}{\secp}\right) > \frac{1}{30}\right] \leq e^{-2(\secp/30)^2}  = \negl(\secp),\] where the expression inside the probability holds for large enough $\secp$, and the inequality is Hoeffding's inequality (using the case where the random variables are sampled without replacement).

\end{proof}

\begin{lemma}[$\Pi^\parl$ analogue of \cref{lemma:Q-out}]\label{lemma:Q-out-4-msg}
For any family $\{Q_\secp,P_\secp\}_{\secp \in \bbN}$ such that $\{P_\secp \circ Q_\secp\}_{\secp \in \bbN}$ is pseudo-deterministic, sequence of inputs $\{x_\secp\}_{\secp \in \bbN}$, and QPT adversary $\{\sA_\secp\}_{\secp \in \bbN}$, it holds that 
        \[\Pr\left[\begin{array}{l} \sV_\Ver^\parl(Q,x,\sparam,y,T,z) = \{\{q_{i,t}\}_{t \in [\secp]}\}_{i : T_i = 1} ~~ \wedge \\ \MM_\secp\left(\{\{P(q_{i,t})\}_{t \in [\secp]}\}_{i : T_i = 1}\right) = 1-P(Q(x)) ~~ \wedge \\ w \notin D_\out[P,P(Q(x))] \end{array} : \begin{array}{r}(\pp,\sparam) \gets \sV^\parl_\Gen(1^\secp,Q) \\ y \gets \sA(\pp,\sparam) \\ T \gets \{0,1\}^{\binom{\secp^2}{\secp}} \\ z \gets \sA(T) \\ w \coloneqq \mathsf{TestRoundOutputs}[\sparam](y,T,z)\end{array}\right] = \negl(\secp),\] where $\sA$ maintains an internal state, which we leave implicit above, and where $\TestOut$ is defined as in \cref{subsec:classical-verification}, except that string $T$ is explicitly given rather than being computed by a random oracle $H$.
\end{lemma}

\begin{proof}
We again define $\{\{q_{i,t}\}_{t \in [\secp]}\}_{i \in [\secp^2]}$ based on the prover's first message $\{y_i\}_{i \in [\secp^2]}$, and 

\begin{itemize}
    \item Let $a$ be the fraction of $\{q_{i,t}\}_{t \in [\secp]}$ such that $\Maj\left(\{P(q_{i,t})\}_{t \in [\secp]}\right) = 1-P(Q(x))$ over $i \in [\secp^2]$.
    \item Let $b$ be the fraction of $\{q_{i,t}\}_{t \in [\secp]}$ such that $\Maj\left(\{P(q_{i,t})\}_{t \in [\secp]}\right) = 1-P(Q(x))$ over $i : T_i = 1$.
\end{itemize}

By the definition of $D_{\out}[P,P(Q(x))]$, \[w \notin D_{\out}[P,P(Q(x))] \implies a \leq \frac{1}{3}+\frac{1}{\secp}.\]
Thus, the proof is completed by showing that 
\[\Pr\left[b - a \geq \frac{1}{2} - \left(\frac{1}{3} + \frac{1}{\secp}\right) > \frac{1}{10}\right] \leq e^{-2(\secp/10)^2} = \negl(\secp),\]
which again follows from Hoeffding's inequality. Note that this argument is entirely statistical, and holds even if $\sA_\secp$ has $\sparam$. 
\end{proof}

\noindent\underline{Step 3}. Note that the protocol $\Pi^\CV$ is exactly Fiat-Shamir applied to $\Pi^\parl$. That is, take $\Pi^\parl$ and let the verifier's challenge $T$ be computed by applying a random oracle $H$ to the prover's first message $y$. This results in exactly the protocol $\Pi^\CV$, where we have re-defined the prover operations $(\sP_\Com^\parl,\sP_\Prove^\parl)$ as $(\sP_\Prep^\CV,\sP_\Prove^\CV,\sP_\Meas^\CV)$. Then, straightforward applications of Measure-and-Reprogram (\cref{thm:measure-and-reprogram}) show that \cref{lemma:soundness-4-msg}, \cref{lemma:Q-in-4-msg}, and \cref{lemma:Q-out-4-msg} imply \cref{lemma:soundness}, \cref{lemma:Q-in}, and \cref{lemma:Q-out} respectively.

In more detail, suppose that \cref{lemma:soundness} is false, and fix $P,Q,x$, and an adversary $\sA$ that breaks that claim. Define a predicate $V$ that takes as input $y$, $H(y)$, the rest of the transcript of the protocol, and the verifier's secret parameters $\sparam$, and outputs whether \[\sV^\CV_\Ver(Q,x,\sparam,\pi) = \{\{q_{i,t}\}_{t \in [\secp]}\}_{i : T_i = 1} ~~ \wedge ~~ \MM_\secp(\{\{P(q_{i,t})\}_{t \in [\secp]}\}_{i : T_i = 1}) = 1-P(Q(x)) .\]

Define adversary $\sB^H$ to run an interaction between $\sA$ and the verifier $\sV^\CV$, forwarding random oracles calls to an external oracle $H$, and output $y$ along with auxiliary information $\mathsf{aux}$ that includes the rest of the transcript and $\sparam$. Then we have that  

\[\Pr\left[V(y,H(y),\mathsf{aux}) = 1 : (y,\mathsf{aux}) \gets \sB^H\right] = \nonnegl(\secp).\]

Since $\sB$ makes $\poly(\secp)$ queries to $H$, \cref{thm:measure-and-reprogram} implies that there exists a simulator $\Sim$ such that 

\[\Pr\left[V(y,T,\mathsf{aux}) = 1: \begin{array}{r} (y,\state) \gets \Sim[\sB] \\ T \gets \{0,1\}^{\binom{\secp^2}{\secp}} \\ \mathsf{aux} \gets \Sim[\sB](T,\state) \end{array}\right] = \nonnegl(\secp).\]

Moreover, by definition (\cref{thm:measure-and-reprogram}), $\Sim[\sB]$ runs $\sB$ honestly except that it simulates $H$ and measures one of $\sB$'s queries to $H$. Thus, $\Sim[\sB]$ can be used as an adversarial prover interacting in $\Pi^{\parl}$, where $y$ is sent to the verifier as the prover's first message, and $T$ is sampled and given in response. Thus, $\Sim[\sB]$ can be used to violate \cref{lemma:soundness-4-msg}.

Finally, the fact that \cref{lemma:Q-in-4-msg} implies \cref{lemma:Q-in} and \cref{lemma:Q-out-4-msg} implies \cref{lemma:Q-out} can be shown in exactly the same way, by defining the appropriate predicate $V$. This completes the proof.

\end{proof}

\end{document}